\documentclass[11pt]{article}
\usepackage[margin=1in]{geometry}

\usepackage{mathtools,amssymb,amsthm}
\usepackage{bm}
\usepackage{algorithm,algpseudocode}
\usepackage{multirow,booktabs}
\usepackage{natbib}
\usepackage{microtype}
\usepackage[hidelinks]{hyperref}
\allowdisplaybreaks

\newtheorem{theorem}{Theorem}
\newtheorem{lemma}{Lemma}
\newtheorem{definition}{Definition}
\newtheorem{remark}{Remark}
\newtheorem{claim}{Claim}
\newenvironment{claimproof}
  {\begin{proof}[Claim Proof]}
  {\end{proof}}

\newcommand{\size}[1]{\ensuremath{|#1|}}
\newcommand{\Ceil}[1]{\ensuremath{\left\lceil#1\right\rceil}}
\newcommand{\lra}[1]{\ensuremath{(#1)}}
\newcommand{\lrc}[1]{\ensuremath{\{#1\}}}
\newcommand{\lrd}[1]{\ensuremath{\langle #1\rangle}}
\newcommand{\lrA}[1]{\ensuremath{\left(#1\right)}}
\newcommand{\lrAl}[1]{\ensuremath{\left(#1\right.}}
\newcommand{\lrAr}[1]{\ensuremath{\left. #1\right)}}
\newcommand{\lrB}[1]{\ensuremath{\left[#1\right]}}
\newcommand{\lrBl}[1]{\ensuremath{\left[#1\right.}}
\newcommand{\lrBr}[1]{\ensuremath{\left. #1\right]}}
\newcommand{\lrC}[1]{\ensuremath{\left\{#1\right\}}}
\newcommand{\PP}[1]{\ensuremath{\Pr[#1]}}
\newcommand{\PPP}[1]{\ensuremath{\Pr\lrB{#1}}}
\newcommand{\EE}[1]{\ensuremath{\mathbb{E}[#1]}}
\newcommand{\EEE}[1]{\ensuremath{\mathbb{E}\lrB{#1}}}
\def\A{\mathcal{A}}
\def\E{\mathcal{E}}
\def\M{\mathcal{M}}
\def\O{\mathcal{O}}
\def\P{\mathcal{P}}
\def\Q{\mathcal{Q}}
\def\T{\mathcal{T}}
\def\V{\mathcal{V}}
\def\Z{\mathcal{Z}}
\newcommand{\df}[2]{\ensuremath{\varphi_{#1}(#2)}}
\def\NM{N_{\max}}
\def\RLB{\mathrm{rlb}}
\def\TSP{\mathrm{tsp}}
\def\OPT{\mathrm{opt}}
\def\SOL{\mathrm{sol}}
\def\OBJ{\mathrm{obj}}

\title{Approximating Unsplittable Capacitated Vehicle Routing: Analysis via Factor-Revealing LPs}
\author{Jingyang Zhao\thanks{Kyung Hee University. Email: \texttt{jingyangzhao1020@gmail.com}.}
\and
Mingyu Xiao\thanks{University of Electronic Science and Technology of China. Email: \texttt{myxiao@gmail.com}.}}
\date{}

\begin{document}
\maketitle
\begin{abstract}
The capacitated vehicle routing problem (CVRP) is one of the most extensively studied problems in combinatorial optimization. Given a depot and a set of customers with demands in a metric space, the objective is to find a set of tours, each starting and ending at the depot, operated by a capacitated vehicle located at the depot, such that all customers are served and the total travel cost is minimized. We consider the unsplittable variant, where the demand of each customer must be served entirely by a single tour. The previous best approximation ratio was about $\alpha+1.694$, where $\alpha$ is the approximation ratio for the traveling salesman problem. We propose improved approximation algorithms. First, we introduce a unified algorithmic framework that generalizes and integrates known linear programming (LP) rounding and matching-based approaches. We show that these techniques are complementary and can be combined to yield improved approximation ratios. Second, we develop a novel analytical framework based on customer demand discretization, which yields new structural bounds and, through factor-revealing LPs, achieves expected approximation ratios of about $\alpha+1.647$ for the general-capacity case and $\alpha+1.548$ for the fixed-capacity case. Finally, we derandomize our algorithms using a priori 2-tour sets and the method of conditional expectations while maintaining nearly the same approximation ratios.

\end{abstract}

% \noindent\textbf{Keywords:}
% Combinatorial Optimization, Capacitated Vehicle Routing, Metric, Approximation Algorithms, Graph Algorithms

\section{Introduction}\label{sec1}
The vehicle routing problems (VRPs) form a fundamental class of combinatorial optimization problems with numerous applications in logistics, transportation, and supply chain management~\citep{book:toth:BOOK:14}, and have been studied widely in operations research~\citep{journal:gupta:OR:12}.
Among the various VRPs, the capacitated vehicle routing problem (CVRP)~\citep{journal:dantzig:MS:59} is one of the most famous problems.

In the CVRP, we are given an undirected complete graph $G=(V\cup\lrc{r}, E)$ with a metric edge-cost function $c: E\to\mathbb{R}_{\geq 0}$ that satisfies the triangle inequality.
The vertex set $V=\{v_1,\ldots,v_n\}$ represents $n$ customers, where each customer $v\in V$ has an integer demand $d_v\in\mathbb{Z}_{\geq 1}$.
A vehicle with integer capacity $k\in\mathbb{Z}_{\geq 1}$ is initially located at the depot $r$.
A tour is a walk that starts and ends at the depot, serves customers with total demand at most $k$, and has a cost equal to the sum of the costs of all edges in the tour.

The objective of the CVRP is to find a set of tours for the vehicle that serve all customers while minimizing the total travel cost.
In the \emph{unsplittable} version of the problem, each customer's demand must be served entirely by a single tour, whereas in the \emph{splittable} version, a customer's demand may be served by multiple tours.
If each customer has unit demand, the problem is referred to as the \emph{unit-demand} version.

When $k=1$ or $k=2$, the CVRP can be solved in polynomial time~\citep{conference:asano:STOC:97}. However, for each fixed $k\geq3$, the problem becomes APX-hard, even in the unit-demand case~\citep{report:asano:REPORT:96} (see also the restatement in the appendix of \citep{journal:zhao:ORL:25}).
Additional complexity results for several CVRP variants can be found in~\citep{journal:hassin:ORL:05}.
When $k=\infty$, the CVRP reduces to the well-known metric traveling salesman problem (TSP), which admits a $\frac{3}{2}$-approximation algorithm~\citep{journal:christofides:ORF:22}. 
Recently, \citet{conference:karlin:IPCO:23,journal:karlin:OR:24} improved the ratio to less than $\frac{3}{2} - 10^{-36}$.
We use $\alpha$ to denote the current best-known approximation ratio for the metric TSP, regardless of the specific algorithm achieving it.

\citet{journal:haimovich:MOR:85} introduced the iterated tour partition (ITP) algorithm for the unit-demand CVRP. 
Based on an $\alpha$-approximate TSP tour on $V$, the ITP algorithm achieves an approximation ratio of $\alpha(1-\frac{\Ceil{n/k}}{n})+\frac{\Ceil{n/k}}{n/k}$.
Later, \citet{journal:altinkemer:TS:90} showed that, given an $\alpha$-approximate TSP tour on $V\cup \lrc{r}$, the ITP algorithm achieves an approximation ratio of $\alpha+1-\frac{\alpha}{k}$ for both the unit-demand CVRP and the splittable CVRP.
Moreover, \citet{journal:altinkemer:ORL:87} modified the ITP algorithm to handle the unsplittable CVRP, achieving an approximation ratio of $\alpha+2-\frac{2\alpha}{k}$ for even $k$ and $\alpha+2-\frac{\alpha}{k}$ for odd $k$.

Subsequently, \citet{journal:jansen:NETWORKS:93} obtained an approximation ratio of $\alpha+\frac{k-\alpha}{\Ceil{k/2}}$ for the unsplittable CVRP, which equals $\alpha+2-\frac{2\alpha}{k}$ for even $k$, and $\alpha+2-\frac{2\alpha}{k}-\frac{2(k-\alpha)}{k(k+1)}$ for odd $k$.
Hence, this result improves upon the previous approximation ratio of $\alpha+2-\frac{\alpha}{k}$ for odd $k$.
Later, \citet{journal:bompadre:DO:06} made several improvements.
For any $\alpha \geq 1$, they improved the approximation ratio by a term of at least $\frac{1}{3k^3}$ for all three versions of the CVRP.
In the specific case of $\alpha = \frac{3}{2}$, the improvement is at least $\frac{1}{4k^2}$ for the splittable and unit-demand cases, and at least $\frac{1}{3k^2}$ for the unsplittable case.

For the case of general vehicle capacity, significant progress was made by \citet{journal:blauth:MP:23}. They improved the approximation ratio to $\alpha+1-\varepsilon_\alpha$ for the splittable and unit-demand cases, and to $\alpha+2-2\varepsilon_\alpha$ for the unsplittable case, where $\varepsilon_\alpha$ is a constant related to $\alpha$, with $\varepsilon_\alpha>\frac{1}{3000}$ when $\alpha=\frac{3}{2}$.
\citet{journal:friggstad:MOR:25} obtained two further improvements for the unsplittable CVRP. The first is an $(\alpha+1.75)$-approximation algorithm using a matching-based approach, while the second is an $(\alpha+1+\ln 2+\varepsilon)$-approximation algorithm based on linear programming (LP) rounding, with a running time of $n^{\O({1}/{\varepsilon})}$.
They also showed that both approximations can be further improved by a small constant through the method from~\citep{journal:blauth:MP:23}.

The CVRP has attracted independent interest for small values of $k$, such as $k=3$~\citep{journal:bazgan:DAM:05} and $k=4$~\citep{journal:anily:NRL:99}. 
For the unit-demand CVRP with $k=3$, \citet{journal:bazgan:DAM:05} obtained an approximation ratio of $3-\frac{4}{3}\rho+\varepsilon<1.934$, where $\rho=\frac{4}{5}$ denotes the current best-known approximation ratio for the maximum-cost TSP on general graphs~\citep{conference:dudycz:IPCO:17}.
For the unit-demand CVRP with $k=4$, \citet{journal:anily:NRL:99} obtained an approximation ratio of $\frac{7}{4}$.

More recently, \citet{journal:zhao:IC:26} proposed a $(\frac{5}{2}-\Theta(\frac{1}{\sqrt{k}}))$-approximation algorithm for the splittable and unit-demand cases, and a $(\frac{5}{2}+\ln 2-\Theta(\frac{1}{\sqrt{k}}))$-approximation algorithm for the unsplittable case.
They also designed several algorithms for the cases $k\in\{3,4,5\}$.
In particular, for the splittable and unit-demand cases, the approximation ratio is $\frac{3}{2}$ for both $k=3$ and $k=4$; 
for the unsplittable case, the approximation ratio is $\frac{3}{2}$ for $k=3$, $\frac{7}{4}$ for $k=4$, and $2.157$ for $k=5$.
When $k \leq 1.7\times 10^7$, the approximation ratio achieved in \citep{journal:zhao:IC:26} attains the current best-known bound for all three variants of the CVRP.

\paragraph{\textbf{The Asymmetric CVRP.}}

In this problem, the edge-cost function does not necessarily satisfy the symmetry condition $c(x,y)=c(y,x)$ for distinct $x$ and $y$.
When $k=\infty$, the problem is known as the asymmetric TSP.

For the asymmetric TSP, the first nontrivial approximation ratio of $\log_2 n$ was established by \citet{journal:frieze:NETWORKS:82}, which was subsequently improved to $0.999\cdot\log_2 n$~\citep{journal:blaser:TALG:08}, then to $0.842\cdot\log_2 n$~\citep{journal:kaplan:JACM:05}, and further to $\frac{2}{3}\log_2 n$~\citep{conference:feige:APPROX:07}. 
\citet{journal:asadpour:OR:17} later proposed a randomized $\O\lra{\frac{\log n}{\log\log n}}$-approximation algorithm. 
Subsequently, \citet{journal:svensson:JACM:20} established the first constant-factor approximation with ratio $506$. 
This ratio was later improved to $22+\varepsilon$~\citep{journal:traub:SICOMP:22}, then to $17+\varepsilon$~\citep{book:traub:BOOK:24}, and most recently to $14.7$~\citep{journal:vygen:arXiv:26}.

As noted in \citep{journal:friggstad:MOR:25,book:traub:BOOK:24}, several existing approximation algorithms for the metric CVRP extend to the asymmetric case.
Suppose that we are given an $\alpha$-approximate asymmetric TSP tour on $V\cup\lrc{r}$.
Most notably, the ITP algorithm~\citep{journal:haimovich:MOR:85,journal:altinkemer:TS:90} still achieves an approximation ratio of $\alpha+1$ for the splittable and unit-demand cases, which remains the current best-known result.
For the unsplittable CVRP, the modified ITP algorithm~\citep{journal:altinkemer:ORL:87} still achieves an approximation ratio of $\alpha+2$.
Moreover, the matching-based and LP-based algorithms achieve approximation ratios of $\alpha+1.75$ and $\alpha+1+\ln2+\varepsilon$, respectively.
Since it is unclear whether the method from~\citep{journal:blauth:MP:23} can be extended to the asymmetric case, the current best-known approximation ratio for the unsplittable case is $\alpha+1+\ln2+\varepsilon$~\citep{journal:friggstad:MOR:25}.

\subsection{Other Related Work}
Although only limited progress has been made for the CVRP on general metrics, a substantial body of work has focused on structured metrics.
We briefly review some representative results below.

\paragraph{\textbf{The CVRP in Euclidean Spaces.}}
Consider the unit-demand CVRP. 
In $\mathbb{R}^2$, a polynomial-time approximation scheme (PTAS) is known for $k=\O(\log\log n)$~\citep{journal:haimovich:MOR:85}, $k=\O(\frac{\log n}{\log\log n})$ or $k=\Omega(n)$~\citep{conference:asano:STOC:97}, and $k\leq 2^{\log^{f(\varepsilon)}n}$~\citep{journal:adamaszek:IJFCS:10}.
Moreover, a quasi PTAS (QPTAS) with running time $n^{\log^{\O(1/\varepsilon)}n}$ is known for arbitrary $k$~\citep{journal:das:ALGCA:15}.
This running time was later improved to $n^{\O(\log^6 n/\varepsilon^5)}$~\citep{journal:jayaprakash:TALG:23}, where the authors also established a QPTAS for graphs of bounded treewidth, bounded doubling metrics, and bounded highway dimension with arbitrary $k$.
In $\mathbb{R}^\ell$, for any fixed $\ell$, a PTAS is known for $k=\O(\log^{1/\ell}n)$~\citep{conference:khachay:DOOR:16}.

Furthermore, \citet{conference:grandoni:ITCS:23} proposed a $(2+\varepsilon)$-approximation algorithm for the unsplittable CVRP in $\mathbb{R}^2$, while \citet{conference:friggstad:SODA:26} proposed a $(2-\tau)$-approximation algorithm for the unit-demand CVRP in $\mathbb{R}^2$, for some constant $\tau>10^{-5}$.

\paragraph{\textbf{The CVRP on Tree Metrics.}}
On tree metrics, the splittable CVRP is NP-hard~\citep{journal:labbe:OR:91}, while the unsplittable CVRP is NP-hard to approximate within a factor better than $\frac{3}{2}$~\citep{journal:golden:NETWORKS:81} unless $\textrm{P}=\textrm{NP}$.
For the splittable CVRP on tree metrics, \citet{conference:hamaguchi:ISAAC:98} proposed a $\frac{3}{2}$-approximation algorithm.
This ratio was subsequently improved to $\frac{\sqrt{41}-4}{4}$~\citep{journal:asano:JOCO:01}, then to $\frac{4}{3}$~\citep{conference:becker:APPROX:18}, and eventually to a PTAS~\citep{journal:mathieu:TALG:23}. 
For the unsplittable CVRP on tree metrics, \citet{journal:labbe:OR:91} proposed a 2-approximation algorithm, which was later improved to a tight approximation ratio of $\frac{3}{2}+\varepsilon$~\citep{journal:mathieu:MP:25}.

\paragraph{\textbf{The CVRP on Other Structured Metrics.}}
Consider the unit-demand CVRP with $k=\O(1)$. 
\citet{conference:becker:ESA:17,conference:becker:ESA:18} proposed a QPTAS for planar metrics and graphs of bounded genus, a PTAS for graphs of bounded highway dimension, and an exact algorithm with running time $\O(n^{\text{tw}\cdot k})$ for graphs of treewidth $\text{tw}$.
Later, \citet{conference:becker:WADS:19b} proposed a PTAS for planar metrics.
\citet{conference:cohen-addad:FOCS:20} proposed an efficient PTAS (EPTAS) for graphs of bounded treewidth and bounded genus, and the first QPTAS for minor-free metrics. 
More recently, \citet{conference:filtser:FOCS:22} proposed an EPTAS for planar metrics with \emph{almost linear} running time, as well as the first EPTAS for minor-free metrics.
\citet{conference:cohen-addad:FOCS:23} proposed a QPTAS for minor-free metrics with arbitrary $k$.
Finally, for the CVRP on graphic metrics, \citet{conference:momke:SOSA:23} proposed a $1.95$-approximation algorithm.

Further results on CVRPs and TSPs can be found in recent surveys~\citep{conference:chen:ICCMSO:23,journal:saller:AOR:25} and the book~\citep{book:traub:BOOK:24}.

\subsection{Our Contributions}
In this paper, we propose improved approximation algorithms for the unsplittable CVRP.
We consider the asymmetric case.\footnote{In the metric case, no notable further improvement is observed beyond those resulting from improved approximations for the TSP.}
Our contributions are summarized below.

\paragraph{\textbf{A Unified Algorithm for the Unsplittable CVRP.}}

We propose an algorithm, denoted by UITP-$(\gamma,\sigma,\delta,\mu)$, which takes as input four constants $\gamma,\sigma,\delta,\mu\geq 0$ with $\sigma\leq\delta\leq\mu\leq\frac{1}{2}$. This framework unifies and generalizes several existing approaches based on LP-rounding and matching techniques for the unsplittable CVRP.

Specifically, we first revisit the algorithm for the unsplittable CVRP introduced by \citet{journal:friggstad:MOR:25}, which employs a constant $0\leq \delta\leq \frac{1}{2}$ to distinguish customers based on whether their demands exceed $\delta\cdot k$. 
Note that when $\delta=0$, their algorithm coincides with the classic ITP algorithm proposed by \citet{journal:altinkemer:ORL:87}. 
We refine their algorithm by dedicating a single tour to each customer with demand $d_v>\frac{1}{2}k$, and denote the resulting algorithm by ITP-$\delta$.

The ITP-$\delta$ algorithm serves as a flexible building block, which, as shown by \citet{journal:friggstad:MOR:25}, can be combined with either the LP-rounding or the matching-based techniques.

We first integrate ITP-$\delta$ with the LP-rounding technique to obtain the LP-ITP-$(\gamma,\sigma,\delta)$ algorithm, which takes as input three constants $\gamma,\sigma,\delta\geq 0$ with $\sigma\leq\delta\leq\frac{1}{2}$.
Let $\T$ denote the set of feasible tours serving only \emph{big} customers, namely customers with demands greater than $\sigma\cdot k$.
The LP-ITP-$(\gamma,\sigma,\delta)$ algorithm formulates a configuration LP, where a variable $x_T$ is introduced for each tour $T\in\T$ to minimize $\sum_{T\in\T}c(T)\cdot x_T$ subject to covering all big customers.
Since the number of feasible tours serving only big customers is bounded, an optimal fractional solution $\{x^*_T\}_{T\in\T}$ can be computed in $n^{\O(1/\sigma)}$ time when $\sigma>0$, or in $n^{\O(k)}$ time when $\sigma=0$. 
The algorithm then constructs a set of tours $\T_{\mathrm{round}}$ by independently selecting each tour $T\in\T$ with probability $\min\{1,\gamma\cdot x^*_T\}$.
Finally, ITP-$\delta$ is invoked to serve all remaining unserved customers. 
Note that when $\gamma=0$, LP-ITP-$(\gamma,\sigma,\delta)$ coincides with ITP-$\delta$. 
Moreover, we may set $\sigma=\O(\varepsilon)$ for general vehicle capacity $k$, where $\varepsilon$ is a sufficiently small constant, and $\sigma=0$ when $k$ is fixed.

Furthermore, we incorporate the matching-based technique into LP-ITP-$(\gamma,\sigma,\delta)$ to obtain the full UITP-$(\gamma,\sigma,\delta,\mu)$ algorithm.
While \citet{journal:friggstad:MOR:25} applied the matching-based technique directly to the ITP-$\delta$ algorithm---specifically, by computing a minimum-cost \emph{$2$-tour set} serving all customers with demands greater than $\frac{1}{3}k$ before invoking ITP-$\delta$---we instead apply this technique after the LP-rounding phase. 
Here, a $2$-tour set is defined as a set of tours, each serving at most two customers, and a minimum-cost $2$-tour set can be computed in polynomial time via a reduction to the minimum-weight $2$-set cover problem (MW$2$SCP).

Our UITP-$(\gamma,\sigma,\delta,\mu)$ algorithm first performs the same LP-rounding procedure as LP-ITP-$(\gamma,\sigma,\delta)$ to obtain an initial collection of tours.
Recall that $0\leq\sigma\leq\delta\leq\mu\leq\frac{1}{2}$. 
Due to the randomness of the LP-rounding step, certain customers with demands greater than $\mu\cdot k$ may remain unserved. 
The algorithm then computes a minimum-cost $2$-tour set restricted to these unserved customers with demands greater than $\mu\cdot k$. 
Finally, ITP-$\delta$ is invoked to serve all remaining unserved customers. 
In this way, the UITP-$(\gamma,\sigma,\delta,\mu)$ algorithm integrates both the LP-rounding and matching-based techniques within a unified framework.

Since ITP-$\delta$ serves each customer with demand greater than $\frac{1}{2}k$ via a single tour, UITP-$(\gamma,\sigma,\delta,\mu)$ coincides with LP-ITP-$(\gamma,\sigma,\delta)$ when $\mu=\frac{1}{2}$.
Thus, by appropriately choosing the parameter tuple $(\gamma,\sigma,\delta,\mu)$, the UITP-$(\gamma,\sigma,\delta,\mu)$ algorithm can recover both previous matching-based and LP-based algorithms of~\citet{journal:friggstad:MOR:25}, thereby providing a highly general and unified approach for the unsplittable CVRP.

\paragraph{\textbf{Improvements via Simple Combinations.}}

For the unsplittable CVRP, \citet{journal:friggstad:MOR:25} proposed a matching-based $\lra{\alpha+1.75}$-approximation algorithm and an LP-based $\lra{\alpha+1+\ln2+\O(\varepsilon)}$-approximation algorithm, which correspond to a trade-off between UITP-$(0,-,\frac{1}{3},\frac{1}{3})$\footnote{When $\gamma=0$, the choice of $\sigma$ is irrelevant; we adopt `-' for convenience.} and UITP-$(0,-,\frac{1}{3},\frac{1}{2})$, and to UITP-$(\ln 2,\varepsilon,\varepsilon,\frac{1}{2})$, respectively.
The matching-based approach appears to be dominated by the LP-based approach.
However, perhaps somewhat surprisingly, we show that these two approaches are complementary and can be combined to yield improved approximation ratios via an analysis similar to that in \citep{journal:friggstad:MOR:25}.

\begin{itemize}
\item
\textbf{General Vehicle Capacity Case:}
By carefully choosing two constants $\gamma_1,\gamma_2\geq 0$ and returning the best solution among the three solutions computed by UITP-$(0,-,\frac{1}{3},\frac{1}{3})$, UITP-$(\gamma_1,\varepsilon,\frac{1}{3},\frac{1}{2})$, and UITP-$(\gamma_2,\varepsilon,\varepsilon,\frac{1}{2})$, we obtain an expected approximation ratio of at most $\alpha+1.6759$, improving upon $\alpha+1+\ln2\approx \alpha + 1.6932$.
\item
\textbf{Fixed Vehicle Capacity Case:}
More simply, by carefully choosing a single constant $\gamma\geq 0$ and returning the better of the two solutions computed by UITP-$(0,-,\frac{1}{3},\frac{1}{3})$ and UITP-$(\gamma,0,\frac{1}{3},\frac{1}{2})$, we obtain an expected approximation ratio of at most $\alpha+1.5897$.
\end{itemize}
While these improvements are obtained via simple combinations, they are somewhat unintuitive to discover.\footnote{In fact, we initially believed that the matching-based approach was dominated by the LP-based approach. 
However, while attempting to improve the approximation algorithm for the unsplittable CVRP with capacity $k=5$ \citep{journal:zhao:IC:26}, we discovered that combining these approaches could yield further improvements. 
This observation motivated the developments in the present work.
}

\paragraph{\textbf{Further Improvements via Novel Analysis.}}

The above improvements highlight the potential of the matching-based approach.
Recall that UITP-$(\gamma,\sigma,\delta,\mu)$ coincides with LP-ITP-$(\gamma,\sigma,\delta)$ when $\mu=\frac{1}{2}$. Thus, the matching-based approach becomes relevant only when $\mu<\frac{1}{2}$.
However, in the simple combinations above, we only set $\mu=\frac{1}{3}$ and $\gamma=0$ when considering $\mu<\frac{1}{2}$.
One main reason is that, as shown by \citet{journal:friggstad:MOR:25}, when $\mu\geq \frac{1}{3}$, a minimum-cost 2-tour set serving all customers with demands greater than $\mu\cdot k$, denoted by $\M^{>\mu}$, has cost at most $\OPT$, where $\OPT$ denotes the cost of an optimal solution to the unsplittable CVRP.
Consequently, to fully exploit this lower bound on $\OPT$, it is natural to set $\mu=\frac{1}{3}$.
Moreover, we only have the crude bound $\EE{c(\T^{>\mu}_{\mathrm{left}})}\leq c(\M^{>\mu})\leq\OPT$, where $\T^{>\mu}_{\mathrm{left}}$ denotes the minimum-cost 2-tour set restricted to the unserved customers with demands greater than $\mu\cdot k$ computed after the LP-rounding step.
Since this bound is weak, we have to use $\gamma=0$ when setting $\mu=\frac{1}{3}$.
If stronger structural bounds related to $c(\M^{>\mu})$ and $\EE{c(\T^{>\mu}_{\mathrm{left}})}$ could be established for other values of $\mu$, further improvements could be achieved by optimizing over different values of $\mu$ and $\gamma$.
This is exactly the analytical framework we develop next.

We obtain novel bounds related to $c(\M^{>\mu})$ and $\EE{c(\T^{>\mu}_{\mathrm{left}})}$ based on a \emph{discretization of customer demands}. 
For instance, we show how to derive new bounds on $c(\M^{>\mu})$ for $\mu\in\lrc{\frac{1}{4},\frac{1}{3}}$.

All customers with demands belonging to the interval $(\frac{1}{5-i}k,\frac{1}{4-i}k]$ for some $i\in[3]\coloneq\lrc{1,2,3}$ are categorized as type-$i$ customers. Thus, customer demands greater than $\frac{1}{4}k$ are discretized into three classes. 
Fix a set of tours, denoted by $\T_{\mathrm{large}}$, obtained by shortcutting all customers with demands at most $\frac{1}{4}k$ from an optimal solution to the unsplittable CVRP.
Clearly, $\T_{\mathrm{large}}$ serves all type-$i$ customers for $i\in[3]$, and by the triangle inequality, $c(\T_{\mathrm{large}})\leq \OPT$.
Since each type-$i$ customer has demand greater than $\frac{1}{4}k$, each tour in $\T_{\mathrm{large}}$ serves at most three such customers.

We define a \emph{tour pattern} as a sequence of the form $\lrd{i_1,i_2,\ldots,i_\ell}$, where $i_1,\ldots,i_\ell\in[3]$.
This represents a single tour serving one type-$i_1$ customer, one type-$i_2$ customer, and so on. 
A tour pattern is feasible if it allows a tour of this form whose total served demand is at most $k$.  
Since each type-$i$ customer may have demand arbitrarily close to, but strictly greater than, $d_i\coloneq \frac{1}{5-i}k$, feasibility is equivalent to $\sum_{j=1}^{\ell} d_{i_j} < k$.
For example, it can be verified that $\lrd{1,1,2}$ is feasible, whereas $\lrd{1,1,1,1}$ is infeasible.
Let $\P_{\mathrm{pat}}$ denote the set of all such patterns, and let $N_{\mathrm{pat}}=\size{\P_{\mathrm{pat}}}$. It can be verified that $N_{\mathrm{pat}}=18$ and
{\small\begin{align*}
\P_{\mathrm{pat}}
&=\{\lrd{1},\lrd{2},\lrd{3},\lrd{1,1},\lrd{1,2},\lrd{1,3},\lrd{2,1},\lrd{2,2},\lrd{2,3},\ldots,\lrd{2,2,1}\}.
\end{align*}
We use the above order of all patterns in $\P_{\mathrm{pat}}$, i.e., the first one in $\P_{\mathrm{pat}}$ is $\lrd{1}$, so $\T_{\mathrm{large}}$ can be partitioned into $N_{\mathrm{pat}}$ classes: $\Q_1,\ldots,\Q_{N_{\mathrm{pat}}}$,
where $\Q_i$ consists of all tours in $\T_{\mathrm{large}}$ whose pattern is the $i$-th one in $\P_{\mathrm{pat}}$.
Let
\begin{align*}
&
\Q_1=\{(v^{z}_{0},v^{z}_{1},v^{z}_{0})\}_{z\in[\size{\Q_1}]},\;
\Q_2=\{(v^{z}_{0},v^{z}_{2},v^{z}_{0})\}_{z\in[\size{\Q_2}]},\;
\Q_3=\{(v^{z}_{0},v^{z}_{3},v^{z}_{0})\}_{z\in[\size{\Q_3}]},\\
&
\Q_4=\{(v^{z}_{0},v^{z}_{4},v^{z}_{5},v^{z}_{0})\}_{z\in[\size{\Q_4}]},\;
\Q_5=\{(v^{z}_{0},v^{z}_{6},v^{z}_{7},v^{z}_{0})\}_{z\in[\size{\Q_5}]},\;
\ldots,\;
\Q_{18}=\{(v^{z}_{0},v^{z}_{{38}},v^{z}_{{39}},v^{z}_{{40}},v^{z}_{0})\}_{z\in[\size{\Q_{18}}]},
\end{align*}}\!\!
where $v^{z}_{0}=r$, for any $1\leq z\leq \max_{i\in[N_{\mathrm{pat}}]}\size{\Q_i}$, denotes the depot.

For each $i\in[N_{\mathrm{pat}}]$, we build a complete graph $\mathbb{M}_i=(\mathbb{V}_i,\mathbb{E}_i)$. Each $\mathbb{V}_i$ consists of the vertex $0$ corresponding to the depot and one vertex for each occurrence of a customer type in the $i$-th pattern in $\P_{\mathrm{pat}}$. In particular, 
{\small\begin{align*}
&\mathbb{V}_1=\{0,1\},\; 
\mathbb{V}_2=\{0,2\},\; 
\mathbb{V}_3=\{0,3\},\;
\mathbb{V}_4=\{0,4,5\},\;
\mathbb{V}_5=\{0,6,7\},\;
\ldots,\;
\mathbb{V}_{18}=\{0,38,39,40\}.
\end{align*}}

These $N_{\mathrm{pat}}$ complete graphs share the common vertex $0$, and their union forms a (multi-)graph, denoted by $\mathbb{G}=(\mathbb{V}_1\cup\cdots\cup\mathbb{V}_{N_{\mathrm{pat}}}, \mathbb{E}_1\cup\cdots\cup\mathbb{E}_{N_{\mathrm{pat}}})$.
We specify a normalized edge-cost function $\chi$ on $\mathbb{G}$ as follows:
{\small\begin{align*}
\chi(i,j) \coloneq \frac{\sum_{z\in[\size{\Q_\ell}]} c(v^{z}_{i},v^{z}_{j})}{\OPT},\quad\forall\ i,j\in \mathbb{V}_\ell,\quad\forall\ \ell\in[N_{\mathrm{pat}}].
\end{align*}}

It can be verified that each $\mathbb{M}_i$ equipped with $\chi$ is an asymmetric graph, yielding a series of bounds on $\chi$.
Moreover, for each $i\in[N_{\mathrm{pat}}]$, define a \emph{canonical tour} $\mathbb{T}_i$ in $\mathbb{M}_i$ corresponding to the $i$-th pattern in $\P_{\mathrm{pat}}$:
{\small\begin{align*}
&\mathbb{T}_1=(0,1,0),\; 
\mathbb{T}_2=(0,2,0),\; 
\mathbb{T}_3=(0,3,0),\;
\mathbb{T}_4=(0,4,5,0),\; 
\mathbb{T}_5=(0,6,7,0),\; 
\ldots,\; 
\mathbb{T}_{18}=(0,38,39,40,0).
\end{align*}}

The above method maps all tours in $\Q_i$ to $\mathbb{T}_i$ in $\mathbb{M}_i$. 
We have $\chi(\mathbb{T}_i) = \frac{c(\Q_i)}{\OPT}$, and thus $\sum_{i\in[N_{\mathrm{pat}}]}\chi(\mathbb{T}_i)\leq 1$.
Consequently, it can be verified that
\begin{itemize}
  \item the $\chi$-cost of any 2-tour set in $\mathbb{G}$ covering all vertices in $\mathbb{V}_1\cup\cdots\cup\mathbb{V}_{N_{\mathrm{pat}}}$ that correspond to type-1, type-2, or type-3 customers establishes an upper bound on $c(\M^{>\mu})$ for $\mu=\frac{1}{4}$;
  \item the $\chi$-cost of any 2-tour set in $\mathbb{G}$ covering all vertices in $\mathbb{V}_1\cup\cdots\cup\mathbb{V}_{N_{\mathrm{pat}}}$ that correspond to type-2 or type-3 customers establishes an upper bound on $c(\M^{>\mu})$ for $\mu=\frac{1}{3}$.
\end{itemize}

Thus, by constructing appropriate 2-tour sets in $\mathbb{G}$, we obtain upper bounds on $c(\M^{>\mu})$ for $\mu\in\lrc{\frac{1}{4},\frac{1}{3}}$.
Moreover, by integrating this discretization approach with a careful structural analysis of the randomized LP-rounding phase, we can obtain upper bounds on $\EE{c(\T^{>\mu}_{\mathrm{left}})}$ that are stronger than those on $c(\M^{>\mu})$.

We remark that the analysis above discretizes customer demands in $(\frac{1}{4}k,k]$ into merely three classes: $(\frac{1}{4}k,\frac{1}{3}k], (\frac{1}{3}k,\frac{1}{2}k], (\frac{1}{2}k,k]$, which mainly yields novel upper bounds on $\EE{c(\T^{>\mu}_{\mathrm{left}})}$ for $\mu\in\lrc{\frac{1}{4},\frac{1}{3}}$.
In the full analysis, we further discretize customer demands into finer classes.
This finer discretization not only strengthens the resulting upper bounds, but also yields novel bounds for a broader range of values of $\mu$.

Naturally, our final randomized algorithm for the unsplittable CVRP simply runs UITP-$(\gamma,\sigma,\delta,\mu)$ under various parameter settings and returns the best solution obtained.
Specifically, the parameters are enumerated in a brute-force manner. 
For instance, in the case of general vehicle capacity, we set $\sigma=\varepsilon$, and enumerate $\gamma\in\lrc{0,\frac{1}{M},\frac{2}{M},\ldots,1}$, $\delta\in\lrc{\varepsilon,\frac{1}{N},\frac{2}{N},\ldots,\frac{1}{2}}$, and $\mu\in\lrc{\frac{1}{N},\frac{2}{N},\ldots,\frac{1}{2}}$ satisfying $\mu\geq \delta$, for large integers $M$ and $N$ (e.g., $M=1000$ and $N=240$).
Using the derived upper bounds on $\EE{c(\T^{>\mu}_{\mathrm{left}})}$, the approximation ratio can be bounded by expressions involving $\chi$.
Combining various structural bounds on $\chi$, we obtain an upper bound on the expected approximation ratio by formulating and solving a \emph{factor-revealing} LP.

Due to the large size of the resulting factor-revealing LP, we restrict the demand discretization to hundreds of classes in order to ensure computational efficiency and numerical stability.
Through this analytical framework, we obtain expected approximation ratios of about $\alpha+1.647$ for the unsplittable CVRP with general vehicle capacity and $\alpha+1.548$ for the unsplittable CVRP with fixed vehicle capacity.

\paragraph{\textbf{Derandomization via A Priori 2-Tour Sets.}}

Finally, we address the derandomization of UITP-$(\gamma,\sigma,\delta,\mu)$.
As noted by \citet{journal:friggstad:MOR:25}, pure LP-rounding schemes such as LP-ITP-$(\gamma,\sigma,\delta)$ can be derandomized using the classical method of conditional expectations~\citep{book:williamson:BOOK:11}.
However, when the matching-based approach is applied after the LP-rounding step, the conditional expectation of $c(\T^{>\mu}_{\mathrm{left}})$ may no longer be efficiently computable once some tours have already been selected or rejected. As a result, standard derandomization techniques do not appear to apply directly to the UITP-$(\gamma,\sigma,\delta,\mu)$ framework.

To overcome this issue, we derandomize UITP-$(\gamma,\sigma,\delta,\mu)$ by restructuring the matching phase via a priori $2$-tour sets. 
Specifically, we first obtain a randomized variant of UITP-$(\gamma,\sigma,\delta,\mu)$.
Before the LP-rounding phase, we compute a $2$-tour set $\Z^{>\mu}$ covering all customers with demands greater than $\mu\cdot k$, chosen to minimize the expected cost of the residual tour set $\Z^{>\mu}_{\mathrm{left}}$ obtained by shortcutting all customers subsequently served by the LP-rounding phase from $\Z^{>\mu}$.
Through a reduction to the MW$2$SCP based on the optimal fractional solution $\{x^*_T\}_{T\in\T}$, we show that the tour set $\Z^{>\mu}$ can be computed in polynomial time.
Then, after the LP-rounding phase, instead of computing the minimum-cost 2-tour set $\T^{>\mu}_{\mathrm{left}}$, we use $\Z^{>\mu}_{\mathrm{left}}$ to serve the remaining customers with demands greater than $\mu\cdot k$.

By an analysis similar to that used for $\EE{c(\T^{>\mu}_{\mathrm{left}})}$, we show that this modification leads to only slightly weaker approximation ratios: about $\alpha+1.648$ for the general vehicle capacity case and $\alpha+1.548$ for the fixed vehicle capacity case.
More importantly, since $\Z^{>\mu}$ is computed in advance, both the probability that a customer is covered by the rounded tour set and its expected contribution to the upper bound can be computed in polynomial time when conditioned on some tours having already been selected or rejected.
Therefore, this enables a derandomization via the method of conditional expectations, yielding a deterministic polynomial-time algorithm with the same approximation guarantees.

\subsection{Paper Organization}
Section~\ref{sec2} introduces the basic notation.
Section~\ref{sec3} presents our unified algorithm, which reviews existing ITP algorithms and introduces ITP-$\delta$ (Section~\ref{sec3.1}), details the LP-based LP-ITP-$(\gamma,\sigma,\delta)$ algorithm (Section~\ref{sec3.2}), and develops the unified UITP-$(\gamma,\sigma,\delta,\mu)$ algorithm (Section~\ref{sec3.3}) with improved approximation ratios (Section~\ref{sec3.4}).
Section~\ref{sec4} proposes further improvements, introducing our final randomized algorithm (Section~\ref{sec4.1}) via both a simplified framework for intuition (Section~\ref{sec4.2}) and a full framework for the complete analysis (Section~\ref{sec4.3}).
Subsequently, Section~\ref{sec5} addresses the derandomization of our algorithm to obtain deterministic variants.
Finally, Section~\ref{sec6} concludes the paper and outlines several directions for future work.

Due to space constraints, proofs of lemmas and theorems marked with ``*'' are deferred to the appendix.

\section{Preliminaries}\label{sec2}
In the CVRP, an input instance is represented as $I=(G,c,d,k)$, where $G=(V\cup\lrc{r}, E)$ is an undirected complete graph with vertex set $V\cup \lrc{r}$ and edge set $E$. 
Here, $V=\{v_1,\ldots,v_n\}$ denotes the customers, and $r$ represents the depot, where a vehicle of capacity $k\in\mathbb{Z}_{\geq1}$ is initially located.
Each edge has a non-negative cost given by $c: E\to \mathbb{R}_{\geq0}$.
The demand of each customer is specified by $d:V\to\mathbb{Z}_{\geq1}$, where $d_v$ denotes the demand of customer $v$.
The cost function $c$ is a \emph{metric} if it satisfies $c(x,x)=0$ (reflexivity), $c(x,y)=c(y,x)$ (symmetry), and $c(x,y)\leq c(x,z)+c(z,y)$ (triangle inequality) for all $x,y,z\in V\cup\lrc{r}$.
If the symmetry condition does not necessarily hold, $c$ is called an \emph{asymmetric metric}.
Unless otherwise specified, we consider the asymmetric case, since our results do not exhibit notable differences between the metric and asymmetric settings, except for the TSP approximation ratio $\alpha$, which satisfies $\alpha = 1.5 - 10^{-36}$ in the metric case~\citep{conference:karlin:IPCO:23,journal:karlin:OR:24} and $\alpha = 14.7$ in the asymmetric case~\citep{journal:vygen:arXiv:26}.

A \emph{walk} $W$ in $G$, denoted by $(v_1,v_2,\ldots, v_\ell)$, is a sequence of vertices in $V\cup\lrc{r}$, where vertices may be repeated, and each consecutive pair of vertices is connected by an edge in $E$. 
Let $E(W)$ denote the (multi-)set of edges it contains, and let $V(W)$ denote the set of customers visited by $W$. 
The cost of $W$ is defined as $c(W)=\sum_{e\in E(W)}c(e)$.
A \emph{path} is a walk in which no vertex appears more than once. 
The first and last vertices of a path are referred to as its \emph{terminals}.
A \emph{closed} walk is a walk in which the first and the last vertices coincide, and
a \emph{cycle} is a closed walk in which only the first and the last vertices are the same.
A \emph{tour} is a cycle of the form $(r, v_1,v_2,\ldots, v_\ell, r)$, where each $v_i$ is a customer.
Such a tour is also called an $\ell$-tour. 
A tour visiting a single customer, e.g., $(r,v_1,r)$, is called \emph{trivial}.
An $\ell$-tour set is a set of tours in which each tour is an $i$-tour for some $i\leq \ell$.
A TSP tour on $V'\subseteq V\cup\lrc{r}$ is a cycle in the induced subgraph $G[V']$ that visits every vertex in $V'$.
Given a closed walk, by the triangle inequality, one can skip repeated vertices along the walk to obtain a cycle of non-increasing cost.
Such an operation is called \emph{shortcutting}.
For any cost function $c: X\to \mathbb{R}_{\geq0}$, we extend it to subsets of $X$ by defining $c(Y) = \sum_{x\in Y} c(x)$ for any $Y\subseteq X$.

The CVRP can be described as follows.

\begin{definition}[\textbf{The CVRP}]
Given a CVRP instance $I=(G,c,d,k)$, the objective is to construct a set of tours $\T$ with a demand assignment $\lambda: V\times\T\to \mathbb{Z}_{\geq 1}$ such that
\begin{enumerate}
\item each tour serves total demand at most $k$, i.e., $\sum_{v\in V(T)}\lambda_{v,T}\leq k$ for every $T\in \T$,
\item each tour serves only customers on the tour, i.e., $\sum_{v\in V\setminus V(T)}\lambda_{v,T}=0$ for every $T\in \T$,
\item the demands of all customers are fully satisfied, i.e., $\sum_{T\in\T}\lambda_{v,T}=d_v$ for every $v\in V$,
\end{enumerate}
and $c(\T)=\sum_{T\in\T}c(T)$ is minimized.
\end{definition}

By the triangle inequality, we assume without loss of generality that $\lambda_{v,T}=0$ if and only if $v\in V\setminus V(T)$.
According to the properties of the demand assignment, we distinguish the following variants of the problem.
If each customer's demand must be served entirely by a single tour, the CVRP is called the \emph{unsplittable} CVRP.
If a customer's demand may be served by multiple tours, the CVRP is called the \emph{splittable} CVRP.
If each customer has unit demand, the CVRP is called the \emph{unit-demand} CVRP.

We mainly consider the unsplittable CVRP.
By definition, we have $d_v \leq k$ for all $v\in V$.
By scaling down the vehicle capacity and customer demands by a factor of $k$, we may assume without loss of generality that the vehicle capacity is $1$, and that each customer has a demand of at most $1$.
This normalization is used only for analytical convenience; all running-time bounds are stated in terms of the original vehicle capacity $k$.
Moreover, we use $\OPT$ to denote the cost of an optimal solution to the unsplittable CVRP.

Let $\A^*$ be an optimal TSP tour on $V\cup \lrc{r}$, and define $\TSP\coloneq c(\A^*)$ and $\RLB\coloneq \sum_{v\in V}d_v\cdot \lrA{c(r,v)+c(v,r)}$.

We recall the following classical lower bounds on the optimal cost of the unsplittable CVRP.

\begin{lemma}[\citet{journal:haimovich:MOR:85}]\label{lowerbound}
It holds that $\OPT \geq \max\lrC{\TSP, \RLB}$.
\end{lemma}

The quantity $\RLB$ is commonly referred to as the \emph{radial lower bound} for the CVRP in the literature.
We assume that $\RLB>0$ to exclude the trivial case where $\RLB=0$.

For any positive integer $\ell$, let $[\ell]\coloneq\{1,2,\ldots,\ell\}$. 
For any real numbers $\ell\leq r$, the intervals $(\ell,r]$ and $[\ell,r]$ denote $\{x\in\mathbb{R}\mid \ell<x\leq r\}$ and $\{x\in\mathbb{R}\mid \ell\leq x\leq r\}$, respectively, and let $\V_{\ell}^{r}\coloneq\{v\in V\mid \ell<d_v\leq r\}$.
Moreover, for any integers $\ell\leq r$, we define $[\ell;r]\coloneq\{\ell,\ell+1,\ldots,r\}$.
For any $t\in\lrc{0,1}$ and any $r\in[0,1]$, we define 
{\small\begin{equation}\label{int}
\df{t}{r}\coloneqq \frac{\sum_{v\in \V^r_0}d^t_v\cdot \lrA{c(r,v)+c(v,r)}}{\sum_{v\in V}d_v\cdot \lrA{c(r,v)+c(v,r)}}=\frac{\sum_{v\in \V^r_0}d^t_v\cdot \lrA{c(r,v)+c(v,r)}}{\RLB}.
\end{equation}}

By definition, we have the following properties.
\begin{lemma}\label{intineq}
It holds that 
\begin{itemize}
  \item $\df{0}{1}=\df{1}{0}=0$ and $\df{1}{1}=1$;
  \item $0\leq \ell\cdot \lrA{\df{0}{r}-\df{0}{\ell}}<\df{1}{r}-\df{1}{\ell}\leq r\cdot \lrA{\df{0}{r}-\df{0}{\ell}}$ for any $0\leq \ell\leq r\leq 1$.
\end{itemize}
\end{lemma}

\section{A Unified Framework for the Unsplittable CVRP}\label{sec3}
In this section, we introduce our unified algorithm UITP-$(\gamma,\sigma,\delta,\mu)$ and demonstrate how simple parameter combinations yield improved approximation ratios for the unsplittable CVRP.

We begin by revisiting the basic ITP algorithms that underlie our framework. 

\subsection{The Basic Iterated Tour Partition Algorithms}\label{sec3.1}
\citet{journal:haimovich:MOR:85} proposed the ITP algorithm for the (unit-demand) CVRP, and \citet{journal:altinkemer:ORL:87} later adapted ITP to handle the unsplittable CVRP.
Given a TSP tour $\A$ on $V\cup \lrc{r}$, ITP first partitions $\A$ into paths, each serving a total demand of at most $k$, and then connects the terminals of each path to the depot to form a tour.
The ITP algorithm can be implemented in $\O(n^2)$ time.

\begin{lemma}[\citet{journal:altinkemer:ORL:87}]\label{itp}
Given a TSP tour $\A$ on $V\cup \lrc{r}$, the ITP algorithm computes, in $\O(n^2)$ time, a solution to the unsplittable CVRP with cost at most $c(\A) + \sum_{v\in V}2d_v\cdot\lrA{c(r,v)+c(v,r)}$. 
\end{lemma}

Fix a constant $\delta\in[0,\frac{1}{2}]$.
By classifying customers according to whether $d_v \leq \delta$ or $d_v > \delta$, \citet{journal:friggstad:MOR:25} proposed a variant of the ITP algorithm for the unsplittable CVRP.

\begin{lemma}[\citet{journal:friggstad:MOR:25}]\label{ditp}
Given a TSP tour $\A$ on $V\cup \lrc{r}$, there exists an algorithm that computes, in $\O(n^2)$ time, a solution to the unsplittable CVRP with cost at most $
c(\A) + \frac{1}{1-\delta}\sum_{v\in \V^\delta_0}d_v\cdot\lra{c(r,v)+c(v,r)} + \frac{2}{1-\delta} \sum_{v\in \V_\delta^1}d_v\cdot\lrA{c(r,v)+c(v,r)} - \frac{\delta}{1-\delta}\sum_{v\in \V_\delta^1}\lrA{c(r,v)+c(v,r)}$.
\end{lemma}

By Lemma~\ref{ditp}, when $\delta=0$, the algorithm of \citet{journal:friggstad:MOR:25} becomes the ITP algorithm in Lemma~\ref{itp}.
Moreover, its performance can be slightly improved by \emph{pre-processing} customers with demands greater than $\frac{1}{2}$.
Specifically, for each customer $v \in \V^1_{1/2}$, we first assign a trivial tour $(r,v,r)$ to serve its demand.
We then use the algorithm of \citet{journal:friggstad:MOR:25} to serve all customers in $\V^{1/2}_0$, using the TSP tour on $\V^{1/2}_0\cup\lrc{r}$ obtained from $\A$ by shortcutting the customers in $\V_{1/2}^1$.
We denote this refined algorithm by ITP-$\delta$.
Since shortcutting does not increase the tour cost by the triangle inequality, we obtain the following result.

\begin{lemma}\label{nditp}
\emergencystretch=1em
Given a TSP tour $\A$ on $V\cup \lrc{r}$ and a constant $\delta\in[0,\frac{1}{2}]$, the ITP-$\delta$ algorithm computes, in $\O(n^2)$ time, a solution to the unsplittable CVRP with cost at most $c(\A) + \frac{1}{1-\delta} \sum_{v\in \V^\delta_0}d_v\cdot\lrA{c(r,v)+c(v,r)} + \frac{2}{1-\delta} \sum_{v\in \V_\delta^{1/2}}d_v\cdot\lrA{c(r,v)+c(v,r)} - \frac{\delta}{1-\delta}\sum_{v\in \V_\delta^{1/2}}\lrA{c(r,v)+c(v,r)}+\sum_{v\in \V_{1/2}^{1}}\lrA{c(r,v)+c(v,r)}$, where each customer in $\V_{1/2}^{1}$ is served using a trivial tour.
Moreover, by \eqref{int}, this upper bound equals
{\small\begin{align*}
c(\A) + \RLB\cdot\lrB{\frac{1}{1-\delta}\df{1}{\delta} + \frac{2}{1-\delta}\lrA{\df{1}{\frac{1}{2}}-\df{1}{\delta}}-\frac{\delta}{1-\delta}\lrA{\df{0}{\frac{1}{2}}-\df{0}{\delta}} + \lrA{\df{0}{1}-\df{0}{\frac{1}{2}}}}.
\end{align*}}
\end{lemma}

\subsection{The LP-Based Iterated Tour Partition Algorithm}\label{sec3.2}
As shown by \citet{journal:friggstad:MOR:25}, the ITP-$\delta$ algorithm in Lemma~\ref{nditp} can be further improved by combining it with an LP-rounding technique.
Next, we briefly introduce their LP-rounding technique.

Fix a constant $\sigma\in[0,\delta]$.
A customer $v\in V$ is called \emph{small} if $d_v\leq\sigma$, and \emph{big} otherwise.
Let $V_{\mathrm{small}} \coloneq \V^\sigma_0$ denote the set of small customers, and $V_{\mathrm{big}} \coloneq V\setminus V_{\mathrm{small}}$ the set of big customers.

Since each big customer has demand greater than $\sigma$, each tour can serve at most $\Ceil{\frac{1}{\sigma}}-1$ big customers.
Since there are $n$ customers in total, there are at most $n^{\O(1/\sigma)}$ possible tours consisting of big customers.
Let $\T$ denote this set of tours.
For each tour $T\in\T$, we introduce a variable $x_T$.
The LP-rounding technique formulates and solves the following LP:
{\small\begin{alignat}{3}
\min & \quad & \sum_{T\in\T} c(T)\cdot x_T \tag{\textbf{A}}\label{reflp}\\
\text{s.t.} && \sum_{\substack{T\in \T: v\in V(T)}}x_T \geq\   & 1, \quad && \forall\  v \in V_{\mathrm{big}}, \notag\\
&& x_T \geq\   & 0, \quad && \forall\  T \in \T. \notag
\end{alignat}}
An optimal solution $\{x^*_T\}_{T\in\T}$ to \eqref{reflp} can be computed in $n^{\O(1/\sigma)}$ time for $\sigma>0$ and in $n^{\O(k)}$ time for $\sigma=0$.

By shortcutting all customers in $V_{\mathrm{small}}$ from an optimal solution to the unsplittable CVRP, one obtains a set of tours $\T'\subseteq \T$ serving all customers in $V_{\mathrm{big}}$. By the triangle inequality, $c(\T')\leq\OPT$.
Hence, the cost of the optimal fractional solution to \eqref{reflp} provides a valid lower bound on $\OPT$, i.e., $\sum_{T\in\T}c(T)\cdot x^*_T\leq\OPT$.

Then, the LP-based ITP algorithm applies the randomized LP-rounding method of \citep{journal:friggstad:MOR:25}.

Specifically, each tour $T\in\T$ is selected independently with probability $\min\{1,\gamma\cdot x^*_T\}$ for some constant $\gamma\geq 0$, yielding a set of tours $\T_{\mathrm{round}}$. 
One may assume w.l.o.g. that each customer is covered at most once in $\T_{\mathrm{round}}$, since if a customer is visited multiple times, we can shortcut those tours to eliminate duplicate visits without increasing the total cost by the triangle inequality.
Due to randomness, some customers may not be covered by any tour in $\T_{\mathrm{round}}$; let $V_{\mathrm{left}}\subseteq V$ denote this set of uncovered customers.

The randomized rounding satisfies the following properties.

\begin{lemma}[\citet{journal:friggstad:MOR:25}]\label{prob}
It holds that
\begin{itemize}
  \item $\EEE{c(\T_\mathrm{round})} \leq \gamma\sum_{T\in\T}c(T)\cdot x^*_T\leq\gamma\cdot \OPT$.
  \item $\PPP{v\in V_{\mathrm{left}}} \leq p_\gamma\coloneq e^{-\gamma}$ for each $v\in V_{\mathrm{big}}$, and $\PPP{v\in V_{\mathrm{left}}} = 1$ for each $v\in V_{\mathrm{small}}$.
\end{itemize}
\end{lemma}

Finally, the ITP-$\delta$ algorithm in Lemma~\ref{nditp} is then applied to serve all remaining customers in $V_{\mathrm{left}}$, using the TSP tour on $V_{\mathrm{left}}\cup\lrc{r}$ obtained by shortcutting an $\alpha$-approximate TSP tour on $V\cup\lrc{r}$.

\begin{algorithm}[t]
\caption{The LP-ITP-$(\gamma,\sigma,\delta)$ Algorithm for the Unsplittable CVRP}
\label{lpalg}
\small
\begin{algorithmic}[1]
\Require An instance $G=(V\cup\lrc{r},E)$ of the unsplittable CVRP with an $\alpha$-approximate TSP tour on $V\cup\lrc{r}$.

\Ensure A feasible solution to the unsplittable CVRP.

\State Initialize $\T_{\mathrm{f}} \gets\emptyset$.\label{alg1s1}

\State Solve the LP~\eqref{reflp} to obtain an optimal fractional solution $\{x^*_T\}_{T\in\T}$.\label{alg1s2}

\State For each $T\in\T$, select $T$ independently with probability $\min\{1,\gamma\cdot x^*_T\}$ to obtain a set of tours $\T_{\mathrm{round}}$.\label{alg1s3}

\State Let $V_{\mathrm{left}}$ denote the set of customers not covered by any tour in $\T_{\mathrm{round}}$.\label{alg1s4}

\State Let $\A$ be the TSP tour on $V_{\mathrm{left}}\cup\lrc{r}$ obtained by shortcutting the given $\alpha$-approximate TSP tour on $V\cup\lrc{r}$.\label{alg1s5}

\State Apply the ITP-$\delta$ algorithm to $\A$ to obtain a set of tours $\T_{\mathrm{ITP}}$ serving all customers in $V_{\mathrm{left}}$.\label{alg1s6}

\State \Return $\T_{\mathrm{f}} \gets \T_{\mathrm{round}}\cup \T_{\mathrm{ITP}}$.\label{alg1s7}

\end{algorithmic}
\end{algorithm}

We denote the above LP-based ITP algorithm by LP-ITP-$(\gamma,\sigma,\delta)$, as described in Algorithm~\ref{lpalg}.
Clearly, when $\gamma=0$, the choice of $\sigma$ is irrelevant, and LP-ITP-$(\gamma,\sigma,\delta)$ coincides with ITP-$\delta$.

\begin{lemma}[*]\label{lpitp}
Given an $\alpha$-approximate TSP tour on $V\cup \lrc{r}$ and constants $\gamma,\sigma,\delta\geq 0$ with $\sigma\leq\delta\leq\frac{1}{2}$, the LP-ITP-$(\gamma,\sigma,\delta)$ algorithm computes a solution to the unsplittable CVRP with expected cost at most
{\footnotesize\begin{align*}
&\alpha\cdot \TSP + \RLB\cdot\frac{1}{1-\delta}\df{1}{\sigma}\\
&\quad+\gamma\cdot\OPT\\
&\quad+p_\gamma\cdot\RLB\cdot\lrB{\frac{1}{1-\delta}\lrA{\df{1}{\delta}-\df{1}{\sigma}}+\frac{2}{1-\delta}\lrA{\df{1}{\frac{1}{2}}-\df{1}{\delta}}-\frac{\delta}{1-\delta}\lrA{\df{0}{\frac{1}{2}}-\df{0}{\delta}} + \lrA{\df{0}{1}-\df{0}{\frac{1}{2}}}}.
\end{align*}}

Moreover, the running time is $n^{\O(1/\sigma)}$ when $\sigma>0$, and $n^{\O(k)}$ when $\sigma=0$.
\end{lemma}

\begin{remark}\label{rmk1}
The LP-ITP-$(\gamma,\sigma,\delta)$ algorithm is randomized due to the randomized LP-rounding step in line~\ref{alg1s3}.
As noted by \citet{journal:friggstad:MOR:25}, this step can be derandomized using the classical method of conditional expectations~\citep{book:williamson:BOOK:11}, since both the probability that a customer is covered by the rounded tour set and its expected contribution to the upper bound of ITP-$\delta$ can be computed efficiently when conditioned on some tours having already been selected or rejected.
Therefore, LP-ITP-$(\gamma,\sigma,\delta)$ can be derandomized, yielding a deterministic algorithm with the same upper bound and running time as those stated in Lemma~\ref{lpitp} for the randomized algorithm.

Moreover, as shown by \citet{journal:friggstad:MOR:25}, by setting $\gamma = \ln 2$ and $\sigma = \delta = \varepsilon$, the LP-ITP-$(\gamma,\sigma,\delta)$ algorithm achieves an approximation ratio of $\alpha + 1 + \ln 2 + \O(\varepsilon)$ with a running time of $n^{\O(1/\varepsilon)}$.
\end{remark}

\subsection{The Unified Iterated Tour Partition Algorithm}\label{sec3.3}
Fix another constant $\mu\in[\delta,\frac{1}{2}]$.
We are now ready to introduce the unified ITP algorithm, UITP-$(\gamma,\sigma,\delta,\mu)$, which generalizes LP-ITP-$(\gamma,\sigma,\delta)$ by incorporating matching-based techniques in~\citep{journal:friggstad:MOR:25}. 

Given an $\alpha$-approximate TSP tour on $V \cup \lrc{r}$, the algorithm first applies the LP-rounding procedure used in Algorithm~\ref{lpalg} to obtain a set of tours $\T_{\mathrm{round}}$.
Let $V_{\mathrm{left}}$ denote the set of customers not covered by any tour in $\T_{\mathrm{round}}$.
Moreover, let $V^{\leq\mu}_{\mathrm{left}}$ denote the set of customers in $V_{\mathrm{left}}$ with demands at most $\mu$, and let $V^{>\mu}_{\mathrm{left}}$ denote the set of all customers in $V_{\mathrm{left}}$ with demands greater than $\mu$.
Hence, $V_{\mathrm{left}} = V^{\leq\mu}_{\mathrm{left}}\cup V^{>\mu}_{\mathrm{left}}$.

Next, the algorithm serves all customers in $V^{>\mu}_{\mathrm{left}}$ using a matching-based approach. 
This step produces another set of tours, denoted by $\T^{>\mu}_{\mathrm{left}}$. 
In particular, $\T^{>\mu}_{\mathrm{left}}$ is a minimum-cost 2-tour set serving all customers in $V^{>\mu}_{\mathrm{left}}$. (In the following, by the triangle inequality, whenever we ask for a 2-tour set to serve a customer set $V'$, we restrict attention to 1-tours and 2-tours that contain customers only from $V'$.)

The algorithm then obtains the TSP tour on $V^{\leq\mu}_{\mathrm{left}}\cup\lrc{r}$ by shortcutting the given $\alpha$-approximate TSP tour on $V\cup\lrc{r}$.
Finally, it applies the ITP-$\delta$ algorithm in Lemma~\ref{nditp} to serve the remaining customers.

The UITP-$(\gamma,\sigma,\delta,\mu)$ algorithm is described in Algorithm~\ref{malg}.

\begin{algorithm}[t]
\caption{The UITP-$(\gamma,\sigma,\delta,\mu)$ Algorithm for the Unsplittable CVRP}
\label{malg}
\small
\begin{algorithmic}[1]
\Require An instance $G=(V\cup\lrc{r},E)$ of the unsplittable CVRP with an $\alpha$-approximate TSP tour on $V\cup\lrc{r}$.

\Ensure A feasible solution to the unsplittable CVRP.

\State Initialize $\T_{\mathrm{f}} \gets\emptyset$.\label{alg2s1}

\State Solve the LP~\eqref{reflp} to obtain an optimal fractional solution $\{x^*_T\}_{T\in\T}$.\label{alg2s2}

\State For each $T\in\T$, select $T$ independently with probability $\min\{1,\gamma\cdot x^*_T\}$ to obtain a set of tours $\T_{\mathrm{round}}$.\label{alg2s3}

\State Let $V_{\mathrm{left}}$ denote the set of customers not covered by any tour in $\T_{\mathrm{round}}$.
Moreover, let $V^{\leq\mu}_{\mathrm{left}}$ denote the set of customers in $V_{\mathrm{left}}$ with demands at most $\mu$, and let $V^{>\mu}_{\mathrm{left}}$ denote the set of all customers in $V_{\mathrm{left}}$ with demands greater than $\mu$.

\State Compute a minimum-cost 2-tour set $\T^{>\mu}_{\mathrm{left}}$ that serves all customers in $V^{>\mu}_{\mathrm{left}}$.

\State Let $\A$ be the TSP tour on $V^{\leq\mu}_{\mathrm{left}}\cup\lrc{r}$ obtained by shortcutting the given $\alpha$-approximate TSP tour on $V\cup\lrc{r}$.

\State Apply the ITP-$\delta$ algorithm to $\A$ to obtain a set of tours $\T_{\mathrm{ITP}}$ serving all customers in $V^{\leq\mu}_{\mathrm{left}}$.

\State \Return $\T_{\mathrm{f}} \gets \T_{\mathrm{round}}\cup\T^{>\mu}_{\mathrm{left}}\cup \T_{\mathrm{ITP}}$.

\end{algorithmic}
\end{algorithm}

A minimum-cost 2-tour set serving all customers in $V^{>\mu}_{\mathrm{left}}$ can be obtained via a reduction to the MW$2$SCP, as described below. 
Each feasible 1-tour or 2-tour is regarded as a set consisting of the customers visited by that tour, with cost equal to the tour cost. 
Hence, finding a minimum-cost collection of such sets that covers all customers in $V^{>\mu}_{\mathrm{left}}$ is exactly an instance of the MW$2$SCP. 
Since there are at most $n$ customers in $V^{>\mu}_{\mathrm{left}}$, the number of feasible 1-tours and 2-tours is at most $\O(n^2)$.
The reduction to the MW$2$SCP can then be carried out in polynomial time. 
Moreover, the MW$2$SCP can be solved in polynomial time via a reduction to the minimum-cost matching problem~\citep{conference:gupta:SOSA:23}.
Hence, by solving the MW$2$SCP, we can obtain, in polynomial time, a minimum-cost 2-tour set $\T^{>\mu}_{\mathrm{left}}$ that serves all customers in $V^{>\mu}_{\mathrm{left}}$.

Note that when $\gamma=\sigma=0$ and $\delta=\mu=\frac{1}{3}$, UITP-$(\gamma,\sigma,\delta,\mu)$ coincides with the matching-based algorithm in \citep{journal:friggstad:MOR:25}, where it was shown that the computation of $\T^{>\mu}_{\mathrm{left}}$ can be implemented in $\O(n^3)$ time. 
It can be verified that the same implementation extends to any $\mu\in[0,1]$.
Therefore, by an argument analogous to that in the proof of Lemma~\ref{lpitp}, we obtain the following result.

\begin{lemma}\label{mitp}
Given an $\alpha$-approximate TSP tour on $V\cup \lrc{r}$ and constants $\gamma,\sigma,\delta,\mu\geq 0$ with $\sigma\leq\delta\leq\mu\leq\frac{1}{2}$, the UITP-$(\gamma,\sigma,\delta,\mu)$ algorithm computes a solution to the unsplittable CVRP with expected cost at most
{\footnotesize\begin{align*}
&\alpha\cdot \TSP + \RLB\cdot\frac{1}{1-\delta}\df{1}{\sigma}\\
&\quad+\gamma\cdot\OPT+p_\gamma\cdot\RLB\cdot\lrB{\frac{1}{1-\delta}\lrA{\df{1}{\delta}-\df{1}{\sigma}} + \frac{2}{1-\delta}\lrA{\df{1}{\mu}-\df{1}{\delta}} - \frac{\delta}{1-\delta}\lrA{\df{0}{\mu}-\df{0}{\delta}}}+\EEE{c(\T^{>\mu}_{\mathrm{left}})}.
\end{align*}}\!\!
Moreover, the running time is $n^{\O(1/\sigma)}$ when $\sigma>0$, and $n^{\O(k)}$ when $\sigma=0$.
\end{lemma}

Since $\EE{c(\T^{>\mu}_{\mathrm{left}})}$ may not be efficiently computable when conditioned on some tours having already been selected or rejected, UITP-$(\gamma,\sigma,\delta,\mu)$ may no longer be efficiently derandomized using the method of conditional expectations while preserving the expected bound in Lemma~\ref{mitp}.

When $\mu=\frac{1}{2}$, we obtain the following properties.

\begin{lemma}[*]\label{fact1}
When $\mu=\frac{1}{2}$, it holds that $\EE{c(\T^{>\mu}_{\mathrm{left}})}\leq p_\gamma\cdot\RLB\cdot\lra{\df{0}{1}-\df{0}{\mu}}$, and the UITP-$(\gamma,\sigma,\delta,\mu)$ algorithm coincides with the LP-ITP-$(\gamma,\sigma,\delta)$ algorithm.
\end{lemma}

Recall that LP-ITP-$(\gamma,\sigma,\delta)$ reduces to ITP-$\delta$ when $\gamma=0$. 
Thus, by Lemma~\ref{fact1}, UITP-$(\gamma,\sigma,\delta,\mu)$ generalizes both LP-ITP-$(\gamma,\sigma,\delta)$ and ITP-$\delta$.
In the following, we apply UITP-$(\gamma,\sigma,\delta,\mu)$ with different choices of $\gamma,\sigma,\delta,\mu$ to obtain improved approximation ratios for the unsplittable CVRP.

\subsection{Improvements via Simple Parameter Combinations}\label{sec3.4}

\citet{journal:friggstad:MOR:25} proposed a matching-based $\lra{\alpha+1.75}$-approximation algorithm and an LP-based $\lra{\alpha+1+\ln2+\O(\varepsilon)}$-approximation algorithm, which correspond to a trade-off between UITP-$(0,-,\frac{1}{3},\frac{1}{3})$ and UITP-$(0,-,\frac{1}{3},\frac{1}{2})$, and to UITP-$(\ln 2,\varepsilon,\varepsilon,\frac{1}{2})$, respectively.
We show that improved approximation ratios can already be obtained from simple combinations of parameter settings in UITP-$(\gamma,\sigma,\delta,\mu)$.
In particular, our algorithms run UITP-$(\gamma,\sigma,\delta,\mu)$ with different parameters and return the best resulting solution.

We first establish a structural bound on $c(\T^{>\mu}_{\mathrm{left}})$.
Let $\M^{>\mu}$ denote a minimum-cost 2-tour set serving all customers in $V$ with demands greater than $\mu$.
The following lemma bounds both $c(\T^{>\mu}_{\mathrm{left}})$ and $c(\M^{>\mu})$.

\begin{lemma}[*]\label{inimr}
When $\mu \in[\frac{1}{3},\frac{1}{2}]$, it holds that $c(\T^{>\mu}_{\mathrm{left}}) \leq c(\M^{>\mu})\leq \OPT$.
\end{lemma}

Note that \citet{journal:friggstad:MOR:25} previously used the bound $c(\M^{>\mu})\leq \OPT$ for $\mu = \frac{1}{3}$ in Lemma~\ref{inimr}. 
Using this structural bound together with different parameter combinations of UITP-$(\gamma,\sigma,\delta,\mu)$, we next derive improved approximation algorithms for the unsplittable CVRP.

We first consider the case of fixed vehicle capacity. 
The improvement follows from a simple combination of UITP-$(0,-,\frac{1}{3},\frac{1}{3})$ and UITP-$(\gamma,0,\frac{1}{3},\frac{1}{2})$ with a carefully chosen parameter $\gamma$.

\begin{theorem}[*]\label{thm1}
Given an $\alpha$-approximate TSP tour on $V\cup\lrc{r}$, for the unsplittable CVRP with fixed vehicle capacity, there exists an $n^{\O(k)}$-time algorithm that achieves an approximation ratio of $\alpha+1.5897$.
\end{theorem}

As mentioned, we consider the asymmetric case by default, so Theorem~\ref{thm1} applies both to the metric case, where $\alpha=1.5-10^{-36}$~\citep{conference:karlin:IPCO:23,journal:karlin:OR:24}, and to the asymmetric case, where $\alpha=14.7$~\citep{journal:vygen:arXiv:26}. 

We next consider the case of general vehicle capacity. 
The improvement follows from combining UITP-$(0,-,\frac{1}{3},\frac{1}{3})$, UITP-$(\gamma_1,\varepsilon,\frac{1}{3},\frac{1}{2})$, and UITP-$(\gamma_2,\varepsilon,\varepsilon,\frac{1}{2})$ with carefully chosen parameters $\gamma_1$ and $\gamma_2$.

\begin{theorem}[*]\label{thm2}
Given an $\alpha$-approximate TSP tour on $V\cup\lrc{r}$, for the unsplittable CVRP with general vehicle capacity, there exists an $n^{\O(1)}$-time algorithm that achieves an approximation ratio of $\alpha+1.6759$.
\end{theorem}

\section{Further Improvements via Novel Analysis}\label{sec4}

The improvements established in the previous section rely on the structural bounds in Lemmas~\ref{fact1} and \ref{inimr}. 
To the best of our knowledge, no other bounds on $c(\M^{>\mu})$ and $\EE{c(\T^{>\mu}_{\mathrm{left}})}$ are known in the literature for the case where $\mu < \frac{1}{2}$.
We introduce an analytical framework that yields new bounds on these quantities, leading to improved approximation ratios via a more refined analysis of UITP-$(\gamma,\sigma,\delta,\mu)$.
Before presenting our framework, we first outline our final randomized algorithm and establish the necessary notation.

\subsection{The Randomized Algorithm and Notation}\label{sec4.1}

Our final randomized algorithm runs UITP-$(\gamma,\sigma,\delta,\mu)$ over a discrete grid of parameter settings and returns the best solution obtained. 
The parameters are enumerated via the following grid-search procedure.

We fix large positive integers $M$ and $N$, and a sufficiently small constant $\varepsilon>0$, say $M=1000$, $N=240$, and $\varepsilon = 10^{-99}$.
If the vehicle capacity $k$ is general, we set $\sigma = \varepsilon$; otherwise, when $k$ is fixed, we set $\sigma = 0$.
The algorithm then enumerates $\gamma$, $\delta$, and $\mu$ over the discrete ranges specified in Algorithm~\ref{ralg}.
Consequently, since $M$ and $N$ are fixed constants, the number of parameter combinations is $\O(MN^2)=\O(1)$.

\begin{algorithm}[t]
\caption{The Randomized Algorithm for the Unsplittable CVRP}
\label{ralg}
\small
\begin{algorithmic}[1]
\Require An instance $G=(V\cup\lrc{r},E)$ of the unsplittable CVRP with an $\alpha$-approximate TSP tour on $V\cup\lrc{r}$.

\Ensure A feasible solution to the unsplittable CVRP.

\State Initialize $\mathfrak{T} \gets\emptyset$.\label{alg3l1}

\State Fix positive integers $M$ and $N$, and a sufficiently small constant $\varepsilon>0$.\label{alg3l2}
\Comment{Say, $M=1000$, $N=240$, and $\varepsilon = 10^{-99}$.}

\State Set $\sigma=\varepsilon$ if the vehicle capacity $k$ is general, and $\sigma=0$ otherwise.\label{alg3l3}

\For{each $\gamma\in\lrC{0,\frac{1}{M},\frac{2}{M},\ldots,1}$, each $\delta\in\lrC{\varepsilon,\frac{1}{N},\frac{2}{N},\ldots, \frac{1}{2}}$, and each $\mu\in\lrC{\frac{1}{N},\frac{2}{N},\ldots, \frac{1}{2}}$ with $\mu\geq \delta$}\label{alg3l4}

\State Run UITP-$(\gamma,\sigma, \delta,\mu)$ and add its solution to $\mathfrak{T}$.\label{alg3l8}

\EndFor
\State\Return $\arg\min_{\T\in\mathfrak{T}} c(\T)$.
\end{algorithmic}
\end{algorithm}

By Lemma~\ref{mitp}, we obtain the following bound on the running time of the algorithm.

\begin{lemma}\label{runtime}
Given an $\alpha$-approximate TSP tour on $V\cup \lrc{r}$ and fixed integers $M$ and $N$, Algorithm~\ref{ralg} runs in $\O(M N^2)\cdot n^{\O(1/\varepsilon)}=n^{\O(1)}$ time when the vehicle capacity $k$ is general, and in $\O(M N^2)\cdot n^{\O(k)}= n^{\O(k)}$ time when the vehicle capacity $k$ is fixed.
\end{lemma}

For the sake of analysis, we introduce the following notation.
For any $t\in\{0,1\}$, let 
{\small\begin{equation}\label{defxi}
\xi_t(i) = \df{t}{\frac{i}{N}}\quad\forall\ i\in\lrc{0,\varepsilon\cdot N}\cup [N].
\end{equation}}

Then, by Lemma~\ref{intineq} and \eqref{defxi}, we have the following structural properties:
{\small\begin{equation}\label{intxi}
\begin{cases}
\xi_1(\varepsilon\cdot N) \leq \varepsilon\cdot \xi_0(\varepsilon\cdot N),\\
\xi_0(1)-\xi_0(\varepsilon\cdot N) \geq 0,\\
\xi_1(1)-\xi_1(\varepsilon\cdot N) \geq \varepsilon\lrA{\xi_0(1)-\xi_0(\varepsilon\cdot N)},\\
\xi_1(1)-\xi_1(\varepsilon\cdot N) \leq \frac{1}{N}\lrA{\xi_0(1)-\xi_0(\varepsilon\cdot N)},\\
\xi_0(i) - \xi_0(i-1) \geq 0, \quad \forall\  i\in[N],\\
\xi_1(i) - \xi_1(i-1) \geq \frac{i-1}{N} \lrA{\xi_0(i) - \xi_0(i-1)}, \quad \forall\  i\in[N],\\
\xi_1(i) - \xi_1(i-1) \leq \frac{i}{N} \lrA{\xi_0(i) - \xi_0(i-1)}, \quad \forall\ i\in[N],\\
\xi_1(N) \leq 1,\\
\xi_0(i), \xi_1(i) \geq 0, \quad \forall\ i\in[N].\\
\end{cases}
\end{equation}}

Let the chosen parameters be parameterized as $\gamma = I_\gamma/M$, $\delta = I_\delta/N$, and $\mu = I_\mu/N$, where $I_\gamma\in\{0\}\cup [M]$, $I_\delta\in\{\varepsilon\cdot N\}\cup [N/2]$, and $I_\mu\in[N/2]$ denote the corresponding scaled integer parameters.
We define the normalized expected cost $\phi_\gamma(I_\mu)\coloneq \frac{\EE{c(\T^{>\mu}_{\mathrm{left}})}}{\OPT}$, and let
{\small\begin{equation}\label{ratio1}
\begin{split}
\zeta_1(I_\gamma,I_\delta,I_\mu)
&\coloneq\gamma + p_\gamma\cdot\lrB{\frac{1}{1-\delta}\xi_1(I_\delta)+ \frac{2}{1-\delta}\lrA{\xi_1(I_\mu)-\xi_1(I_\delta)}- \frac{\delta}{1-\delta}\lrA{\xi_0(I_\mu)-\xi_0(I_\delta)}}+ \phi_\gamma(I_\mu).
\end{split}
\end{equation}}

Furthermore, we define the auxiliary function $\zeta_2$ based on $I_\delta$. When $I_\delta\in[N/2]$, let
{\small\begin{equation}\label{ratio21}
\begin{split}
\zeta_2(I_\gamma,I_\delta,I_\mu)
&\coloneq \gamma + \frac{1}{1-\delta}\xi_1(\varepsilon\cdot N) + p_\gamma\cdot\lrBl{\frac{1}{1-\delta}\lrA{\xi_1(I_\delta) -  \xi_1(\varepsilon\cdot N)}+ \frac{2}{1-\delta}\lrA{\xi_1(I_\mu)-\xi_1(I_\delta)}}\\
&\quad\quad\lrBr{-\frac{\delta}{1-\delta}\lrA{\xi_0(I_\mu)-\xi_0(I_\delta)}} + \phi_\gamma(I_\mu).
\end{split}
\end{equation}}\!\!
Alternatively, when $I_\delta= \varepsilon\cdot N$, let
{\small\begin{equation}\label{ratio22}
\begin{split}
\zeta_2(I_\gamma,I_\delta,I_\mu)
&\coloneq\gamma + \xi_1(\varepsilon\cdot N) + p_\gamma\cdot2\lrA{\xi_1(I_\mu)-\xi_1(\varepsilon\cdot N)}+ \phi_\gamma(I_\mu).
\end{split}
\end{equation}}

We establish the following performance bounds.

\begin{lemma}[*]\label{ratio}
Let $\varepsilon>0$ be a sufficiently small constant, $\gamma\cdot M= I_\gamma\in \{0\}\cup [M]$, $\delta\cdot N=I_\delta\in \{\varepsilon\cdot N\}\cup [N/2]$, and $\mu\cdot N= I_\mu\in[N/2]$ with $I_\mu\geq I_\delta$. 
The expected approximation ratio of UITP-$(\gamma,\sigma,\delta,\mu)$ is at most $\zeta_1(I_\gamma,I_\delta,I_\mu)+\alpha$ for the case of fixed vehicle capacity (by setting $\sigma=0$), and at most $\zeta_2(I_\gamma,I_\delta,I_\mu)+\alpha+3\varepsilon$ for the case of general vehicle capacity (by setting $\sigma=\varepsilon$).
\end{lemma}

By Lemmas~\ref{fact1} and \ref{inimr}, together with  \eqref{defxi}, we have 
{\small\begin{equation}\label{weakmu}
\phi_\gamma(I_\mu)=\frac{\EEE{c(\T^{>\mu}_{\mathrm{left}})}}{\OPT}\leq
\begin{cases}
\frac{c(\M^{>\mu})}{\OPT}\leq 1, & \mu\in[\frac{1}{3},\frac{1}{2}],\\
p_\gamma\cdot (\xi_0(N)-\xi_0(N/2)), & \mu=\frac{1}{2}.
\end{cases}
\end{equation}}

By treating $\xi_t(j)$ as a variable for each $t\in\{0,1\}$ and $j\in\{\varepsilon\cdot N\}\cup [N]$, then all conditions in \eqref{intxi} act as linear constraints (L-constraints) over these variables.
The formulation in \eqref{intxi}--\eqref{weakmu} enables the construction of a factor-revealing LP to numerically evaluate the approximation ratio of our final randomized algorithm.

However, as the values of $M$ and $N$ increase, this method may yield approximation ratios for both the general and fixed vehicle capacity cases that merely approach---rather than improve upon---the approximation ratios in the previous section.
The main limitation is that the bounds on $\phi_\gamma(I_\mu)$ in \eqref{weakmu} inherit the conservative nature of Lemmas~\ref{fact1} and \ref{inimr}. 
To overcome this limitation, we establish structurally stronger bounds on $\phi_\gamma(I_\mu)$.

To this end, we introduce a novel analytical framework. Since its full version is technically involved, we first present a simplified framework as a warm-up to illustrate the main ideas and underlying intuition.

\paragraph{\textbf{Computational Setup.}}
The subsequent frameworks rely on factor-revealing LPs to derive upper bounds on the approximation ratios.
All such LPs in this paper are implemented in Python 3.8.10 and solved using the Gurobi Optimizer 9.5.1~\citep{misc:gurobi:MANUAL:26}. 
All experiments were conducted on a server with an Intel Xeon Gold 6226R CPU @ 2.90 GHz and 512 GB RAM, running Ubuntu Server 20.04 LTS. 
The Python source code used for our numerical computations is provided at \url{https://github.com/JingyangZhao/unsplittableCVRP}.

\subsection{Warm-Up: A Simplified Framework}\label{sec4.2}

To illustrate our main ideas, this simplified framework focuses on the general vehicle capacity case.
Moreover, we only derive bounds on $c(\M^{>\mu})$ for $\mu\in\lrc{\frac{1}{4},\frac{1}{3}}$.
Nevertheless, we will show that analyzing the performance of Algorithm~\ref{ralg} within this framework already yields an expected approximation ratio of at most $\alpha+1.6654$, improving upon the approximation ratio of $\alpha + 1.6759$ established in Theorem~\ref{thm2}.

\subsubsection{The Framework}
We begin with the following discretization of customer demands.

\paragraph{\textbf{Discretize Demands.}}
We discretize the customer demands in $(\frac{1}{4},1]$ into three classes: $(\frac{1}{4},\frac{1}{3}],\allowbreak (\frac{1}{3},\frac{1}{2}],\allowbreak (\frac{1}{2},1]$.

Let $V_0$ denote the set of customers in $V$ with demands at most $\frac{1}{4}$.
For each customer $v\in V\setminus V_0$, we say that $v$ is a type-$i$ customer if
$d_v\in (\frac{1}{5-i}, \frac{1}{4-i}]$ for some $i\in[3]$.
For each $i\in[3]$, let $V_i$ denote the set of all type-$i$ customers in $V\setminus V_0$, and let $d_i=\frac{1}{5-i}$, which is the best lower bound on the demands of type-$i$ customers.

Fix a tour set, denoted by $\T_{\mathrm{large}}$, obtained by shortcutting all customers in $V_0$ from an optimal solution to the unsplittable CVRP.
Clearly, $\T_{\mathrm{large}}$ serves exactly the customers in $V\setminus V_0$, and by the triangle inequality, we have $c(\T_{\mathrm{large}})\leq \OPT$.
Since every customer in $V\setminus V_0$ has demand greater than $\frac{1}{4}$, each tour in $\T_{\mathrm{large}}$ serves at most three such customers.

\paragraph{\textbf{Tour Patterns.}}
We define a \emph{tour pattern} as a sequence of the form $\lrd{i_1,i_2,\ldots,i_\ell}$, where $i_1,\ldots,i_\ell\in[3]$.
It represents a tour serving one type-$i_1$ customer, one type-$i_2$ customer, and so on.

A tour pattern is feasible if it allows a tour of this form whose total served demand is at most $1$.  
Since each type-$i$ customer may have demand arbitrarily close to, but strictly greater than, $d_i$, feasibility is equivalent to
{\small\begin{equation}\label{feaspatt}
\sum_{j=1}^{\ell} d_{i_j} < 1.
\end{equation}}

For example, the pattern $\lrd{1,1,2}$ is feasible because there may exist a tour serving two type-1 customers and one type-2 customer, where the two type-1 customers each have demand $\frac{1}{4}+\varepsilon$ and the type-2 customer has demand $\frac{1}{3}+\varepsilon$, for sufficiently small $\varepsilon>0$.
In contrast, the pattern $\lrd{1,1,1,1}$ is infeasible, since any tour serving four type-1 customers must serve a total demand strictly greater than $1$.

In the following, we restrict our attention to feasible tour patterns.

\paragraph{\textbf{Enumeration of Patterns.}}
Since there are only three customer types, we can enumerate all feasible tour patterns in $\T_{\mathrm{large}}$.  
Let $\P_{\mathrm{pat}}$ denote the set of all such patterns.
It can be verified that
{\small\begin{equation}\label{patcons}
\begin{split}
\P_{\mathrm{pat}}
&=\{
\lrd{1},\lrd{2},\lrd{3},
\lrd{1,1},\lrd{1,2},\lrd{1,3},
\lrd{2,1},\lrd{2,2},\lrd{2,3},
\lrd{3,1},\lrd{3,2},\\
&\quad\quad
\lrd{1,1,1},\lrd{1,1,2},\lrd{1,2,1},\lrd{1,2,2},
\lrd{2,1,1},\lrd{2,1,2},\lrd{2,2,1}
\}.
\end{split}
\end{equation}}\!\!
We fix the above ordering of $\P_{\mathrm{pat}}$; that is, the $i$-th pattern in $\P_{\mathrm{pat}}$ is the $i$-th pattern in the above list for each $i\in[N_{\mathrm{pat}}]$, where $N_{\mathrm{pat}}=\size{\P_{\mathrm{pat}}}$.
We then partition the tours in $\T_{\mathrm{large}}$ into $N_{\mathrm{pat}}$ classes: $\Q_1,\Q_2,\ldots,\Q_{N_{\mathrm{pat}}}$,
where $\Q_i$ consists of all tours in $\T_{\mathrm{large}}$ whose pattern is the $i$-th pattern in $\P_{\mathrm{pat}}$.
Note that we may have $\Q_i=\emptyset$.

Let
{\footnotesize\begin{equation}\label{tourcons}
\begin{split}
&
\Q_1=\{(v^{z}_{0},v^{z}_{1},v^{z}_{0})\}_{z\in[\size{\Q_1}]},\;
\Q_2=\{(v^{z}_{0},v^{z}_{2},v^{z}_{0})\}_{z\in[\size{\Q_2}]},\;
\Q_3=\{(v^{z}_{0},v^{z}_{3},v^{z}_{0})\}_{z\in[\size{\Q_3}]},
\\
&
\Q_4=\{(v^{z}_{0},v^{z}_{4},v^{z}_{5},v^{z}_{0})\}_{z\in[\size{\Q_4}]},\;
\Q_5=\{(v^{z}_{0},v^{z}_{6},v^{z}_{7},v^{z}_{0})\}_{z\in[\size{\Q_5}]},\;
\Q_6=\{(v^{z}_{0},v^{z}_{8},v^{z}_{9},v^{z}_{0})\}_{z\in[\size{\Q_6}]},
\\
&
\Q_7=\{(v^{z}_{0},v^{z}_{{10}},v^{z}_{{11}},v^{z}_{0})\}_{z\in[\size{\Q_7}]},\;
\Q_8=\{(v^{z}_{0},v^{z}_{{12}},v^{z}_{{13}},v^{z}_{0})\}_{z\in[\size{\Q_8}]},\;
\Q_9=\{(v^{z}_{0},v^{z}_{{14}},v^{z}_{{15}},v^{z}_{0})\}_{z\in[\size{\Q_9}]},
\\
&
\Q_{10}=\{(v^{z}_{0},v^{z}_{{16}},v^{z}_{{17}},v^{z}_{0})\}_{z\in[\size{\Q_{10}}]},\;
\Q_{11}=\{(v^{z}_{0},v^{z}_{{18}},v^{z}_{{19}},v^{z}_{0})\}_{z\in[\size{\Q_{11}}]},\;
\Q_{12}=\{(v^{z}_{0},v^{z}_{{20}},v^{z}_{{21}},v^{z}_{{22}},v^{z}_{0})\}_{z\in[\size{\Q_{12}}]},
\\
&
\Q_{13}=\{(v^{z}_{0},v^{z}_{{23}},v^{z}_{{24}},v^{z}_{{25}},v^{z}_{0})\}_{z\in[\size{\Q_{13}}]},\;
\Q_{14}=\{(v^{z}_{0},v^{z}_{{26}},v^{z}_{{27}},v^{z}_{{28}},v^{z}_{0})\}_{z\in[\size{\Q_{14}}]},\;
\Q_{15}=\{(v^{z}_{0},v^{z}_{{29}},v^{z}_{{30}},v^{z}_{{31}},v^{z}_{0})\}_{z\in[\size{\Q_{15}}]},
\\
&
\Q_{16}=\{(v^{z}_{0},v^{z}_{{32}},v^{z}_{{33}},v^{z}_{{34}},v^{z}_{0})\}_{z\in[\size{\Q_{16}}]},\;
\Q_{17}=\{(v^{z}_{0},v^{z}_{{35}},v^{z}_{{36}},v^{z}_{{37}},v^{z}_{0})\}_{z\in[\size{\Q_{17}}]},\;
\Q_{18}=\{(v^{z}_{0},v^{z}_{{38}},v^{z}_{{39}},v^{z}_{{40}},v^{z}_{0})\}_{z\in[\size{\Q_{18}}]},
\end{split}
\end{equation}}\!\!
where $v^{z}_{0}=r$, for any $1\leq z\leq \max_{i\in[N_{\mathrm{pat}}]}\size{\Q_i}$, denotes the depot.

\paragraph{\textbf{Associated Asymmetric Graphs.}}
For each $i\in[N_{\mathrm{pat}}]$, we construct a complete graph $\mathbb{M}_i=(\mathbb{V}_i,\mathbb{E}_i)$. Each vertex set $\mathbb{V}_i$ consists of the vertex $0$ corresponding to the depot and one vertex for each occurrence of a customer type in the $i$-th tour pattern in $\P_{\mathrm{pat}}$. In particular, 
{\footnotesize\begin{equation}\label{vsetcons}
\begin{split}
&
\mathbb{V}_1=\{0,1\},\; \mathbb{V}_2=\{0,2\},\; \mathbb{V}_3=\{0,3\},\;
\mathbb{V}_4=\{0,4,5\},\; \mathbb{V}_5=\{0,6,7\},\; \mathbb{V}_6=\{0,8,9\},\; \mathbb{V}_7=\{0,10,11\},
\\
&
\mathbb{V}_8=\{0,12,13\},\; \mathbb{V}_9=\{0,14,15\},\; \mathbb{V}_{10}=\{0,16,17\},\; \mathbb{V}_{11}=\{0,18,19\},\;
\mathbb{V}_{12}=\{0,20,21,22\},\; \mathbb{V}_{13}=\{0,23,24,25\},
\\
&
\mathbb{V}_{14}=\{0,26,27,28\},\; \mathbb{V}_{15}=\{0,29,30,31\},\; \mathbb{V}_{16}=\{0,32,33,34\},\; \mathbb{V}_{17}=\{0,35,36,37\},\; \mathbb{V}_{18}=\{0,38,39,40\}.
\end{split}
\end{equation}}

These $N_{\mathrm{pat}}$ complete graphs share the common vertex $0$, and their union forms a (multi-)graph, denoted by $\mathbb{G}=(\mathbb{V}_1\cup\mathbb{V}_2\cup\cdots\cup\mathbb{V}_{N_{\mathrm{pat}}}, \mathbb{E}_1\cup\mathbb{E}_2\cup\cdots\cup\mathbb{E}_{N_{\mathrm{pat}}})$.
To define edge costs for these $N_{\mathrm{pat}}$ complete graphs, it suffices to specify an edge-cost function $\chi$ on $\mathbb{G}$ as follows:
{\small\begin{equation}\label{funcons}
\chi(i,j) \coloneq \frac{\sum_{z\in[\size{\Q_\ell}]} c(v^{z}_{i},v^{z}_{j})}{\OPT},\quad\forall\ i,j\in \mathbb{V}_\ell,\quad\forall\ \ell\in[N_{\mathrm{pat}}].
\end{equation}}

Each $\mathbb{M}_i$ equipped with $\chi$ is an asymmetric metric space, since $\chi$ inherits reflexivity and the triangle inequality from the non-negative cost function $c$. 

Moreover, let 
{\small\begin{equation}\label{jsetcons}
\begin{split}
&\mathbb{J}_1=\lrc{1,4,5,6,8,11,17,20,21,22,23,24,26,28,29,33,34,36,40},\\
&\mathbb{J}_2=\lrc{2,7,10,12,13,14,19,25,27,30,31,32,35,37,38,39},\quad \mathbb{J}_3=\lrc{3,9,15,16,18}.
\end{split}
\end{equation}}\!\!
Note that $\mathbb{J}_1$, $\mathbb{J}_2$, and $\mathbb{J}_3$ are exactly the sets of indices of vertices corresponding to type-1, type-2, and type-3 customers, respectively, in the graph $\mathbb{G}$.

\paragraph{\textbf{Pattern Aggregation.}}
For each $i\in[N_{\mathrm{pat}}]$, define a \emph{canonical tour} $\mathbb{T}_i$ in $\mathbb{M}_i$ corresponding to the $i$-th tour pattern in $\P_{\mathrm{pat}}$: 
{\small\begin{equation}\label{contourcons}
\begin{split}
&\mathbb{T}_1=(0,1,0),\; \mathbb{T}_2=(0,2,0),\; \mathbb{T}_3=(0,3,0),\;\\
&\mathbb{T}_4=(0,4,5,0),\; \mathbb{T}_5=(0,6,7,0),\; \mathbb{T}_6=(0,8,9,0),\;\\
&\mathbb{T}_7=(0,10,11,0), \mathbb{T}_8=(0,12,13,0),\; \mathbb{T}_9=(0,14,15,0),\;\\
&\mathbb{T}_{10}=(0,16,17,0),\; \mathbb{T}_{11}=(0,18,19,0),\; \mathbb{T}_{12}=(0,20,21,22,0),\;\\
&\mathbb{T}_{13}=(0,23,24,25,0), \mathbb{T}_{14}=(0,26,27,28,0),\; \mathbb{T}_{15}=(0,29,30,31,0),\;\\
&\mathbb{T}_{16}=(0,32,33,34,0),\; \mathbb{T}_{17}=(0,35,36,37,0),\; \mathbb{T}_{18}=(0,38,39,40,0).
\end{split}
\end{equation}}

Hence, the above construction maps all tours in $\Q_i$ to the canonical tour $\mathbb{T}_i$ in $\mathbb{M}_i$, and we have 
{\small\begin{equation}\label{toureq}
\chi(\mathbb{T}_i) = \frac{c(\Q_i)}{\OPT},\quad\forall\ i\in[N_{\mathrm{pat}}].
\end{equation}}

\subsubsection{The Analysis}
The function $\chi$ satisfies the following properties.

\begin{lemma}[*]\label{zwp}
It holds that 
\begin{enumerate}
  \item[(a)] $\sum_{i=1}^{N_{\mathrm{pat}}}\chi(\mathbb{T}_i)\leq 1$;
  \item[(b)] $0\leq \chi(i,j) \leq \chi(i,\ell) + \chi(\ell, j)$, for any $\mathbb{V}\in\{\mathbb{V}_1,\mathbb{V}_2,\ldots,\mathbb{V}_{N_{\mathrm{pat}}}\}$ and any $i,j,\ell\in \mathbb{V}$;
  \item[(c)] $0=\chi(i,i)$ for any $i\in \mathbb{V}_1\cup\mathbb{V}_2\cup\cdots\cup\mathbb{V}_{N_{\mathrm{pat}}}$.
\end{enumerate}
\end{lemma}

Note that, under our definition, $\mathbb{G}$ is not necessarily a complete graph, and $\size{\Q_i}=\size{\Q_j}$ does not necessarily hold for $i\neq j$.
Consequently, the bounds in part (b) of Lemma~\ref{zwp} hold for any $\mathbb{V}\in\{\mathbb{V}_1,\mathbb{V}_2,\ldots,\mathbb{V}_{N_{\mathrm{pat}}}\}$ and any $i,j,\ell\in \mathbb{V}$, rather than for arbitrary $i,j,\ell\in\mathbb{V}_1\cup\mathbb{V}_2\cup\cdots\cup\mathbb{V}_{N_{\mathrm{pat}}}$.

We then show how to apply the above framework to derive bounds on $c(\M^{>\mu})$ for $\mu\in\{\frac{1}{4}, \frac{1}{3}\}$.

\begin{lemma}[*]\label{match3}
When $\mu=\frac{1}{3}$, it holds that $\frac{c(\M^{>\mu})}{\OPT}\leq \sum_{i=1}^{N_{\mathrm{pat}}}\Gamma_{\mu,i}$,
where 
$\Gamma_{\mu,1} = 0$, 
$\Gamma_{\mu,2} = \chi(\mathbb{T}_2)$, 
$\Gamma_{\mu,3} = \chi(\mathbb{T}_3)$,
$\Gamma_{\mu,4} = 0$,
$\Gamma_{\mu,5} = \chi(0,7)+\chi(7,0)$, 
$\Gamma_{\mu,6} = \chi(0,9)+\chi(9,0)$, 
$\Gamma_{\mu,7} = \chi(0,10)+\chi(10,0)$, 
$\Gamma_{\mu,8} = \chi(\mathbb{T}_8)$,
$\Gamma_{\mu,9} = \chi(\mathbb{T}_9)$,
$\Gamma_{\mu,10} = \chi(0,16)+\chi(16,0)$,
$\Gamma_{\mu,11} = \chi(\mathbb{T}_{11})$,
$\Gamma_{\mu,12} = 0$,
$\Gamma_{\mu,13} = \chi(0,25)+\chi(25,0)$,
$\Gamma_{\mu,14} = \chi(0,27)+\chi(27,0)$,
$\Gamma_{\mu,15} = \chi(0,30)+\chi(30,31)+\chi(31,0)$,
$\Gamma_{\mu,16} = \chi(0,32)+\chi(32,0)$,
$\Gamma_{\mu,17} = \chi(0,35)+\chi(35,37)+\chi(37,0)$,
and $\Gamma_{\mu,18} = \chi(0,38)+\chi(38,39)+\chi(39,0)$.
\end{lemma}

\begin{lemma}[*]\label{match4}
\emergencystretch=2em
When $\mu=\frac{1}{4}$, it holds that $\frac{c(\M^{>\mu})}{\OPT}\leq \sum_{i=1}^{N_{\mathrm{pat}}}\Gamma_{\mu,i}$,
where
\begin{itemize}
  \item $\Gamma_{\mu,i} = \chi(\mathbb{T}_i)$ for each $i\in[11]$;
  \item $\Gamma_{\mu,i}= \min_{(p_1,p_2,p_3)\in \{(p_{i,1},p_{i,2},p_{i,3}),(p_{i,2},p_{i,1},p_{i,3}),(p_{i,3},p_{i,1},p_{i,2})\}}\lra{\chi(0,p_1)+\chi(p_1,0)+\chi(0,p_2)+\chi(p_2,p_3)+\chi(p_3,0)}$,
    where $\mathbb{V}_i=\{0,p_{i,1},p_{i,2},p_{i,3}\}$ and $p_{i,1}<p_{i,2}<p_{i,3}$, for each $i\in[12;18]$.
\end{itemize} 
\end{lemma}

Since $\phi_\gamma(I_\mu)=\frac{\EE{c(\T^{>\mu}_{\mathrm{left}})}}{\OPT}\leq\frac{c(\M^{>\mu})}{\OPT}$ by Lemma~\ref{inimr}, Lemmas~\ref{match3} and \ref{match4} provide upper bounds on $\phi_\gamma(I_\mu)$ for $\mu=\frac{1}{3}$ and $\mu=\frac{1}{4}$, respectively, in terms of the function $\chi$.
To analyze the approximation ratios in \eqref{ratio1}--\eqref{ratio22}, we connect the function $\chi$ with the previous function $\xi_0$.

\begin{lemma}[*]\label{bridge}
It holds that 
\begin{itemize}
  \item $\xi_0(N/3) - \xi_0(N/4) \geq \kappa(N/4)$, where $\kappa(N/4)\coloneq \sum_{j\in\mathbb{J}_1}\lrA{\chi(0,j)+\chi(j,0)}$;
  \item $\xi_0(N/2) - \xi_0(N/3) \geq \kappa(N/3)$, where $\kappa(N/3)\coloneq \sum_{j\in\mathbb{J}_2}\lrA{\chi(0,j)+\chi(j,0)}$;
  \item $\xi_0(N) - \xi_0(N/2) \geq \kappa(N/2)$, where $\kappa(N/2)\coloneq \sum_{j\in\mathbb{J}_3}\lrA{\chi(0,j)+\chi(j,0)}$.
\end{itemize}
\end{lemma}

We are now ready to analyze an upper bound on the expected approximation ratio of Algorithm~\ref{ralg}.

\begin{theorem}\label{thm3}
When $M=1000$ and $N=12$, Algorithm~\ref{ralg} achieves an expected approximation ratio of at most $\alpha+1.6654$ for the unsplittable CVRP with general vehicle capacity.
\end{theorem}
\begin{proof}
Since Algorithm~\ref{ralg} enumerates all parameter combinations within the prescribed grid and returns the minimum-cost solution, its expected approximation ratio is bounded by the approximation ratio achieved under any valid subset of parameter choices.
To establish the claimed bound, it suffices to analyze a restricted configuration of the parameters used in Algorithm~\ref{ralg}. 
In particular, we consider the subset of iterations in which the parameters are restricted to $\delta \in \lrc{\varepsilon, \frac{1}{4}, \frac{1}{3}}$ and $\mu \in \lrc{\frac{1}{4}, \frac{1}{3}, \frac{1}{2}}$ with $\mu \geq \delta$.

For any $\gamma$ with $I_\gamma\in \{0\}\cup [M]$, $I_\delta\in\lrC{\varepsilon\cdot N, N/4, N/3}$, and $I_\mu\in\lrC{N/4, N/3, N/2}$, we introduce variables corresponding to $\xi_0(I_\delta)$, $\xi_1(I_\delta)$, and $\phi_\gamma(I_\mu)$.
For any $i,j\in \mathbb{V}_\ell$ with $\ell\in[N_{\mathrm{pat}}]$, we also introduce a variable corresponding to $\chi(i,j)$.
Hence, the conditions in \eqref{intxi}, together with Lemmas~\ref{zwp}, \ref{match3}, and \ref{bridge}, impose L-constraints on these variables.
Moreover, the conditions in Lemma~\ref{match4} involve the $\min$ operator, which can be linearized into L-constraints using standard techniques.
For example, a constraint of the form $\mathrm{LHS} \leq \mathrm{RHS} + \min\lrc{x_{1,1},\ldots,x_{1,j_1}}+\cdots + \min\lrc{x_{i,1},\ldots,x_{i,j_i}}$ can be linearized by introducing auxiliary variables $\tilde{x}_1,\ldots,\tilde{x}_i$ together with the constraints $\mathrm{LHS} \leq \mathrm{RHS} + \sum_{i'\in[i]} \tilde{x}_{i'}$ and $\tilde{x}_{i'} \leq x_{i',j'}$ for all $i'\in[i]$ and $j'\in[j_{i'}]$.

Consequently, by \eqref{intxi}--\eqref{weakmu} and Lemmas~\ref{ratio}--\ref{bridge}, the expected approximation ratio achieved by Algorithm~\ref{ralg} under this restricted configuration is at most $\pi + \alpha + 3\varepsilon$, where $\pi$ is the optimal value of the following LP:
{\small\setlength{\jot}{2pt}\begin{align*}
\max\quad &\pi \\
\text{s.t.}\quad 
&\pi\leq \zeta_2(I_\gamma,I_\delta, I_\mu),
&&\forall\ I_\gamma\in \{0\}\cup [M],\ \forall\ I_\delta\in\lrC{\varepsilon\cdot N, N/4, N/3},\\
& &&\forall\ I_\mu\in\lrC{N/4, N/3, N/2}\ \text{with}\ I_\mu\geq I_\delta,\\
&\phi_\gamma(I_\mu)\leq \sum_{i=1}^{N_{\mathrm{pat}}}\Gamma_{\mu,i},
&&\forall\ I_\gamma\in \{0\}\cup [M],\ \forall\ I_\mu\in\lrC{N/4, N/3},\\
&\phi_\gamma(I_{1/2})\leq p_\gamma\cdot (\xi_0(N)-\xi_0(N/2)),
&&\forall\ I_\gamma\in \{0\}\cup [M],\\
&\xi_0(N/(4-i))-\xi_0(N/(5-i)) \geq \kappa(N/(5-i)),
&&\forall\ i\in[3],\\
&\rlap{\text{the L-constraints on $\xi_0$ and $\xi_1$ in \eqref{intxi} and on $\chi$ in Lemma~\ref{zwp},}}\\
&\rlap{\text{the L-constraints on $\Gamma_{\mu,i}$ derived by linearizing the constraints in Lemmas~\ref{match3} and \ref{match4}}.}
\end{align*}}

The above LP involves the parameter $\varepsilon$ because the first and the third constraints in \eqref{intxi}, i.e., $\xi_1(\varepsilon\cdot N) \leq \varepsilon\cdot \xi_0(\varepsilon\cdot N)$ and $\xi_1(1)-\xi_1(\varepsilon\cdot N) \geq \varepsilon\lra{\xi_0(1)-\xi_0(\varepsilon\cdot N)}$, depend on the value of $\varepsilon$.
Since $\varepsilon$ is a sufficiently small constant, we choose not to use the first constraint and relax the third constraint to $\xi_1(1)-\xi_1(\varepsilon\cdot N) \geq 0$. 
In this way, the above LP is relaxed to become independent of $\varepsilon$,\footnote{Note that both $\xi_0(\varepsilon\cdot N)$ and $\xi_1(\varepsilon\cdot N)$ are treated as variables in the LP and thus their values are optimized independently of $\varepsilon$.} which can only worsen the resulting upper bound. (In fact, numerical computations show that the optimal value of the LP barely changes.)

The resulting LP serves as a factor-revealing LP whose optimal value provides an upper bound on the parameter $\pi$.
When $M=1000$ and $N=12$, solving the LP yields an optimal value of at most $1.665323$, with a running time of less than one second.\footnote{Even when $M=1500$ and $N=240$, the optimal value remains at most $1.665323$, and the running time is still less than one second.} 
Given the modest size of the LP and its negligible running time, numerical inaccuracies are unlikely to affect the reported bound significantly.
Hence, by choosing $\varepsilon$ to be a sufficiently small constant, Algorithm~\ref{ralg} achieves an expected approximation ratio of at most $\alpha+1.6654$.
\end{proof}

We remark that, using an argument analogous to that in the proof of Theorem~\ref{thm3}, one can obtain an expected approximation ratio of at most $\alpha+1.589694$ for the case of fixed vehicle capacity, which does not improve upon the result of Theorem~\ref{thm2}.
To obtain improvements, we next develop a more comprehensive framework.

\subsection{A More Refined Framework}\label{sec4.3}
We now extend the simplified framework to a more refined analytical framework.

Previously, we discretized the customer demands in $(\frac{1}{4},1]$ into only three classes: $(\frac{1}{4},\frac{1}{3}],\allowbreak (\frac{1}{3},\frac{1}{2}],\allowbreak (\frac{1}{2},1]$.
Naturally, a more refined framework can be obtained by considering more classes.
Moreover, using $\EE{c(\T^{>\mu}_{\mathrm{left}})}\leq c(\M^{>\mu})$, we derived bounds on $\EE{c(\T^{>\mu}_{\mathrm{left}})}$ by simply applying the upper bounds on $c(\M^{>\mu})$. 
In the refined framework, we show that better bounds on $\EE{c(\T^{>\mu}_{\mathrm{left}})}$ can be obtained through a tighter analysis.

\subsubsection{The Framework}

\paragraph{\textbf{Discretize Demands.}}
Fix a small integer $\NM$ (in our final analysis, we consider $\NM=4$), and assume that 
$N$ is a common multiple of $2,3,\ldots,\NM$.
% $N$ is a multiple of $\NM$.
We then discretize the customer demands in $(\frac{1}{\NM},1]$ into $N/2-N/\NM+1$ classes: 
$
(\frac{1}{\NM},\frac{N/\NM+1}{N}],\allowbreak (\frac{N/\NM+1}{N},\frac{N/\NM+2}{N}],\allowbreak\ldots,\allowbreak (\frac{N/2-1}{N}, \frac{1}{2}],\allowbreak (\frac{1}{2},1].
$

The customer demands in $(\frac{1}{2},1]$ could be further discretized into $N/2$ classes: $(\frac{1}{2},\frac{N/2+1}{N}],\allowbreak\ldots,\allowbreak (\frac{N-1}{N},1]$.
However, this appears unnecessary, as numerical computations show that such a refinement does not improve the approximation ratio.
Moreover, since this additional classification increases the size of the factor-revealing LP, we do not adopt it in order to keep the factor-revealing LP solvable for large integers $M$ and $N$.

Let $V_0$ denote the set of customers in $V$ with demands at most $\frac{1}{\NM}$.
For each customer $v\in V\setminus V_0$, we say that $v$ is a type-$i$ customer if
$d_v\in (\frac{N/\NM-1+i}{N}, \frac{N/\NM+i}{N}]$ for some $i\in[N/2-N/\NM]$, or 
$d_v\in (\frac{1}{2}, 1]$ when $i=N/2-N/\NM+1$.
For each $i\in[N/2-N/\NM+1]$, let $V_i$ denote the set of all type-$i$ customers in $V\setminus V_0$, and let $d_i=\frac{N/\NM-1+i}{N}$, which serves as the best lower bound on the demands of type-$i$ customers.

Fix a tour set, denoted by $\T_{\mathrm{large}}$, obtained by shortcutting all customers in $V_0$ from an optimal solution to the unsplittable CVRP.
By the triangle inequality, $c(\T_{\mathrm{large}})\leq \OPT$.
Since every customer in $V\setminus V_0$ has demand greater than $\frac{1}{\NM}$, each tour in $\T_{\mathrm{large}}$ serves at most $\NM-1$ such customers.

\paragraph{\textbf{Tour Patterns.}}
Similarly, we define a \emph{tour pattern} as a sequence $\lrd{i_1,i_2,\ldots,i_\ell}$, where $i_1,\ldots,i_\ell\in[N/2-N/\NM+1]$, and we refer to $\ell$ as its \emph{size}.
We focus on feasible tour patterns, i.e., those satisfying $\sum_{j=1}^{\ell} d_{i_j} < 1$ by \eqref{feaspatt}. 
Since we consider the asymmetric case, if $\lrd{i_1,i_2,\ldots,i_\ell}\neq \lrd{i_\ell,i_{\ell-1},\ldots,i_1}$, then $\lrd{i_1,i_2,\ldots,i_\ell}$ and $\lrd{i_\ell,i_{\ell-1},\ldots,i_1}$ are treated as different patterns.
Since there are exactly $N/2-N/\NM+1$ types of customers in $V\setminus V_0$, and each has demand greater than $\frac{1}{\NM}$, there are at most $(N/2-N/\NM+1)^{\NM-1}$ possible tour patterns in $\T_{\mathrm{large}}$.
These patterns can be enumerated by a computer. 
We denote this set of patterns by $\P_{\mathrm{pat}}$, and let $N_{\mathrm{pat}} = \size{\P_{\mathrm{pat}}}$.
Notably, 
when $N=240$, $N_{\mathrm{pat}} = 41601$ for $\NM=4$, and $N_{\mathrm{pat}} = 400497$ for $\NM=5$.

By definition, each pattern in $\P_{\mathrm{pat}}$ has size at most $\NM-1$.
According to pattern sizes, we then partition $\P_{\mathrm{pat}}$ into $\NM-1$ subsets: $\P^1_{\mathrm{pat}}, \ldots, \P^{\NM-1}_{\mathrm{pat}}$, where $\P^{i}_{\mathrm{pat}}$ denotes the set of patterns in $\P_{\mathrm{pat}}$ of size $i$.
We let $N^i_{\mathrm{pat}} = \size{\P^i_{\mathrm{pat}}}$ for each $i\in[\NM-1]$.
We may fix an ordering of $\P_{\mathrm{pat}}$ such that the patterns in $\P^1_{\mathrm{pat}}$ appear first, followed by those in $\P^2_{\mathrm{pat}}$, and so on.
Then, we partition the tours in $\T_{\mathrm{large}}$ into $N_{\mathrm{pat}}$ classes: $\Q_1,\Q_2,\ldots,\Q_{N_{\mathrm{pat}}}$,
where $\Q_i$ consists of all tours in $\T_{\mathrm{large}}$ whose pattern is the $i$-th pattern in $\P_{\mathrm{pat}}$.

As in \eqref{tourcons}, we define
{\footnotesize\begin{equation}\label{tourcons+}
\begin{aligned}
\Q_i &= \lrC{\lrA{v^{z}_{0},v^{z}_{i},v^{z}_{0}}}_{z\in[\size{\Q_i}]}, 
&\forall\ i\in[N^1_{\mathrm{pat}}],\\
\Q_{N^1_{\mathrm{pat}}+i} 
&= \lrC{\lrA{v^{z}_{0},v^{z}_{{N^1_{\mathrm{pat}}+2i-1}},v^{z}_{{N^1_{\mathrm{pat}}+2i}},v^{z}_{0}}}_{z\in[\size{\Q_{N^1_{\mathrm{pat}}+i}}]},
&\forall\ i\in[N^2_{\mathrm{pat}}],\\
\Q_{N^1_{\mathrm{pat}}+N^2_{\mathrm{pat}}+i} 
&= \lrC{\lrA{v^{z}_{0},v^{z}_{{N^1_{\mathrm{pat}}+2N^2_{\mathrm{pat}}+3i-2}},v^{z}_{{N^1_{\mathrm{pat}}+2N^2_{\mathrm{pat}}+3i-1}},v^{z}_{{N^1_{\mathrm{pat}}+2N^2_{\mathrm{pat}}+3i}},v^{z}_{0}}}_{z\in[\size{\Q_{N^1_{\mathrm{pat}}+N^2_{\mathrm{pat}}+i}}]},
&\forall\ i\in[N^3_{\mathrm{pat}}],\\
&\ldots
\end{aligned}
\end{equation}}

\paragraph{\textbf{Associated Asymmetric Graphs.}}
For each $i\in[N_{\mathrm{pat}}]$, we construct a complete graph $\mathbb{M}_i=(\mathbb{V}_i,\mathbb{E}_i)$. Each vertex set $\mathbb{V}_i$ consists of the vertex $0$ corresponding to the depot and one vertex for each occurrence of a customer type in the $i$-th tour pattern in $\P_{\mathrm{pat}}$. In particular, 
{\small\begin{equation}\label{vsetcons+}
\begin{aligned}
\mathbb{V}_i &= \lrC{0,i}, 
&\forall\ i\in[N^1_{\mathrm{pat}}],\\
\mathbb{V}_{N^1_{\mathrm{pat}}+i} 
&= \lrC{0,N^1_{\mathrm{pat}}+2i-1,N^1_{\mathrm{pat}}+2i},
&\forall\ i\in[N^2_{\mathrm{pat}}],\\
\mathbb{V}_{N^1_{\mathrm{pat}}+N^2_{\mathrm{pat}}+i} 
&= \lrC{0,N^1_{\mathrm{pat}}+2N^2_{\mathrm{pat}}+3i-2,N^1_{\mathrm{pat}}+2N^2_{\mathrm{pat}}+3i-1,N^1_{\mathrm{pat}}+2N^2_{\mathrm{pat}}+3i},
&\forall\ i\in[N^3_{\mathrm{pat}}],\\
&\ldots
\end{aligned}
\end{equation}}

These $N_{\mathrm{pat}}$ complete graphs share the common vertex $0$, and their union forms a (multi-)graph, denoted by $\mathbb{G}=(\mathbb{V}_1\cup\mathbb{V}_2\cup\cdots\cup\mathbb{V}_{N_{\mathrm{pat}}}, \mathbb{E}_1\cup\mathbb{E}_2\cup\cdots\cup\mathbb{E}_{N_{\mathrm{pat}}})$.
To define edge costs for these $N_{\mathrm{pat}}$ complete graphs, we define the edge-cost function $\chi$ on $\mathbb{G}$ in the same way as in \eqref{funcons}, i.e.,
{\small\begin{equation}\label{funcons+}
\chi(i,j) \coloneq \frac{\sum_{z\in[\size{\Q_\ell}]} c(v^{z}_{i},v^{z}_{j})}{\OPT},\quad\forall\ i,j\in \mathbb{V}_\ell,\quad\forall\ \ell\in[N_{\mathrm{pat}}].
\end{equation}}

Each $\mathbb{M}_i$ equipped with $\chi$ remains an asymmetric metric space.

\paragraph{\textbf{Pattern Aggregation.}}
For each $i\in[N_{\mathrm{pat}}]$, define a \emph{canonical tour} $\mathbb{T}_i$ in $\mathbb{M}_i$ corresponding to the $i$-th tour pattern in $\P_{\mathrm{pat}}$: 
{\small\begin{equation}\label{contourcons+}
\begin{aligned}
\mathbb{T}_i &= \lrA{0,i,0},
&\forall\ i\in[N^1_{\mathrm{pat}}],\\
\mathbb{T}_{N^1_{\mathrm{pat}}+i} 
&= \lrA{0,N^1_{\mathrm{pat}}+2i-1,N^1_{\mathrm{pat}}+2i,0},
&\forall\ i\in[N^2_{\mathrm{pat}}],\\
\mathbb{T}_{N^1_{\mathrm{pat}}+N^2_{\mathrm{pat}}+i} 
&= \lrA{0,N^1_{\mathrm{pat}}+2N^2_{\mathrm{pat}}+3i-2,N^1_{\mathrm{pat}}+2N^2_{\mathrm{pat}}+3i-1,N^1_{\mathrm{pat}}+2N^2_{\mathrm{pat}}+3i,0},
&\forall\ i\in[N^3_{\mathrm{pat}}],\\
&\ldots
\end{aligned}
\end{equation}}

\subsubsection{The Analysis}
By an argument analogous to that in the proof of Lemma~\ref{zwp}, the function $\chi$ defined in this subsection satisfies the following properties.

\begin{lemma}\label{zwp+}
It holds that
\begin{enumerate}
  \item[(a)] $\sum_{i=1}^{N_{\mathrm{pat}}}\chi(\mathbb{T}_i)\leq 1$;
  \item[(b)] $0\leq \chi(i,j)\leq\chi(i,\ell) + \chi(\ell, j)$, for any $\mathbb{V}\in\{\mathbb{V}_1,\mathbb{V}_2,\ldots,\mathbb{V}_{N_{\mathrm{pat}}}\}$ and any $i,j,\ell\in \mathbb{V}$;
  \item[(c)] $0=\chi(i,i)$ for any $i\in \mathbb{V}_1\cup\mathbb{V}_2\cup\cdots\cup\mathbb{V}_{N_{\mathrm{pat}}}$.
\end{enumerate}
\end{lemma}

Next, we derive novel upper bounds on $\phi_\gamma(I_\mu)=\frac{\EE{c(\T^{>\mu}_{\mathrm{left}})}}{\OPT}$ for $I_\mu\in[N/\NM; N/2]$ with $\NM=4$.

We introduce some notation.
For each $i\in[N_{\mathrm{pat}}]$, assume that $\mathbb{V}_i=\{0,p_{i,1},\ldots, p_{i,\size{\mathbb{V}_i}-1}\}$, where $p_{i,1}<\cdots<p_{i,\size{\mathbb{V}_i}-1}$.
By definition, each $p_{i,j}$ corresponds to a set of customers $\{v^{z}_{{p_{i,j}}}\}_{z\in[\size{\Q_i}]}$ of the same type.

Let $\mathbb{V}^{>\mu}_{i}$ denote the set consisting of the vertex $0$ and all vertices in 
$\mathbb{V}_i\setminus\{0\}$ whose mapped customer sets contain only customers with demands greater than $\mu$. 
Define 
{\small\begin{equation}\label{newVi}
\mathbb{V}^{>\mu}_{i}=\lrC{0,\tilde{p}^{>\mu}_{i,1},\ldots,\tilde{p}^{>\mu}_{i,\ell_i}}
\quad\text{and}\quad
\mathbb{T}^{>\mu}_{i}=\lrA{0,\tilde{p}^{>\mu}_{i,1},\ldots,\tilde{p}^{>\mu}_{i,\ell_i},0},
\end{equation}}\!\!
where $\tilde{p}^{>\mu}_{i,1}<\cdots<\tilde{p}^{>\mu}_{i,\ell_i}$. Note that $\mathbb{T}^{>\mu}_{i}$ is obtained by shortcutting all vertices in $\mathbb{V}_{i}\setminus \mathbb{V}^{>\mu}_{i}$ from $\mathbb{T}_{i}$.

Moreover, define
{\small\begin{equation}\label{newQi}
\Q^{>\mu}_i = \lrC{\lrA{v^{z}_{0},v^{z}_{{\tilde{p}^{>\mu}_{i,1}}},\ldots,v^{z}_{{\tilde{p}^{>\mu}_{i,\ell_i}}},v^{z}_{0}}}_{z\in[\size{\Q_i}]},
\end{equation}}\!\!
which is the set of tours obtained by shortcutting all customers in $V\setminus V_0$ with demands at most $\mu$ from $\Q_i$.
Hence, $\size{\Q_i}=\size{\Q^{>\mu}_i}$ and all tours in $\Q^{>\mu}_i$ are mapped into the tour $\mathbb{T}^{>\mu}_i$ in $\mathbb{M}_i$ by definition.

Then, $\Q^{>\mu}\coloneq \bigcup_{i\in[N_{\mathrm{pat}}]}\Q^{>\mu}_i$ denotes the set of tours obtained by shortcutting all customers in $V\setminus V_0$ with demands at most $\mu$ from $\T_{\mathrm{large}}$.
We now construct a 2-tour set $\Q^{>\mu}_{\mathrm{left}}$ serving all customers in $V^{>\mu}_{\mathrm{left}}$.

\paragraph{\textbf{Construction of $\Q^{>\mu}_{\mathrm{left}}$.}}

Note that we consider $\NM=4$. 
For any tour $Q^{>\mu}\in\Q^{>\mu}$, since $\size{V(Q^{>\mu})}\leq 3$, we obtain a 2-tour set $Q^{>\mu}_{\mathrm{left}}$ serving all customers in $V(Q^{>\mu})\cap V^{>\mu}_{\mathrm{left}}$ from $Q^{>\mu}$ as follows.

\textbf{Case~1: $\size{V(Q^{>\mu})\cap V^{>\mu}_{\mathrm{left}}}\leq 2$.}
We define $Q^{>\mu}_{\mathrm{left}}$ as a tour set consisting only of the tour obtained from $Q^{>\mu}$ by shortcutting all customers in $V^{>\mu}\setminus V^{>\mu}_{\mathrm{left}}$.

\textbf{Case~2: $\size{V(Q^{>\mu})\cap V^{>\mu}_{\mathrm{left}}}=3$.} Suppose $Q^{>\mu}=(r,v_1,v_2,v_3,r)$ since $\size{V(Q^{>\mu})}\leq 3$. 
We define $Q^{>\mu}_{\mathrm{left}}$ as a $2$-tour set consisting of two tours. Specifically, $Q^{>\mu}_{\mathrm{left}}=\arg\min\lrc{c(Q^{>\mu}_1), c(Q^{>\mu}_2), c(Q^{>\mu}_3)}$, where  
{\small\begin{equation}\label{shortdef}
Q^{>\mu}_1=\lrC{(r,v_1,r), (r,v_2,v_3,r)},\ Q^{>\mu}_2=\lrC{(r,v_2,r),(r,v_1,v_3,r)},\ Q^{>\mu}_3=\lrC{(r,v_3,r), (r,v_1,v_2,r)}.
\end{equation}}\!\!
Hence, this construction follows the same idea as the one used in the proof of Lemma~\ref{match4}.

Clearly, $\Q^{>\mu}_{\mathrm{left}}\coloneq\bigcup_{Q^{>\mu}\in\Q^{>\mu}}Q^{>\mu}_{\mathrm{left}}$ is a valid 2-tour set serving all customers in $V^{>\mu}_{\mathrm{left}}$.
Recall that $\T^{>\mu}_{\mathrm{left}}$ is a minimum-cost 2-tour set serving all customers in $V^{>\mu}_{\mathrm{left}}$. We obtain 
{\small\begin{equation}\label{relc}
\EEE{c(\T^{>\mu}_{\mathrm{left}})}\leq \EEE{c(\Q^{>\mu}_{\mathrm{left}})}.
\end{equation}}

To bound the expected cost $\EE{c(\Q^{>\mu}_{\mathrm{left}})}$, we fix an arbitrary tour $Q^{>\mu}\in\Q^{>\mu}$ and analyze $\EE{c(Q^{>\mu}_{\mathrm{left}})}$.
Since each tour $Q^{>\mu}$ contains at most three customers when $\NM=4$, we consider these three cases separately.

\begin{lemma}\label{expt1}
If $\size{V(Q^{>\mu})}=1$, it holds that $\EE{c(Q^{>\mu}_{\mathrm{left}})}\leq p_\gamma\sum_{v\in V(Q^{>\mu})}\lra{c(r,v)+c(v,r)}$.
\end{lemma}
\begin{proof}
Suppose $Q^{>\mu}=\lra{r,v_1,r}$.
Since $\PP{v_1\in V^{>\mu}_{\mathrm{left}}}\leq p_\gamma$ by Lemma~\ref{prob}, we have $\EE{c(Q^{>\mu}_{\mathrm{left}})} = \PP{v_1\in V^{>\mu}_{\mathrm{left}}}\cdot c(Q^{>\mu}) \leq p_\gamma\cdot c(Q^{>\mu})=p_\gamma\sum_{v\in V(Q^{>\mu})}\lrA{c(r,v)+c(v,r)}$.
\end{proof}

For any tour $Q=\lrA{r, v_1,v_2,\ldots, v_\ell,r}$ with $\ell\geq 2$, define 
{\small\begin{equation}\label{defDelta}
\begin{cases}
\Delta_{\max}(Q)\coloneq \max_{i<j;\, i,j\in[\ell]}\lrA{c(v_i, r)+c(r, v_j) - c(v_i,v_j)},\\
\Delta_{\min}(Q)\coloneq \min_{i<j;\, i,j\in[\ell]}\lrA{c(v_i, r)+c(r, v_j) - c(v_i,v_j)},\\
\Delta_{\mathrm{sum}}(Q)\coloneq \sum_{i<j;\, i,j\in[\ell]}\lrA{c(v_i, r)+c(r, v_j) - c(v_i,v_j)}.
\end{cases}
\end{equation}}

By the triangle inequality, we have
{\small\begin{equation}\label{eqDelta}
\Delta_{\min}(Q)\geq 0.
\end{equation}}

The above notation has the following useful properties.

When $\size{V(Q^{>\mu})}=2$, we have $c(Q^{>\mu})=\sum_{v\in V(Q^{>\mu})}\lra{c(r,v_i)+c(v_i,r)}-\Delta_{\max}(Q^{>\mu})$.
Moreover, when $\size{V(Q^{>\mu})}=\size{V(Q^{>\mu})\cap V^{>\mu}_{\mathrm{left}}}=3$, we have $c(Q^{>\mu}_{\mathrm{left}})=\sum_{v\in V(Q^{>\mu})}\lra{c(r,v_i)+c(v_i,r)}-\Delta_{\max}(Q^{>\mu})$ by \eqref{shortdef}.

\begin{lemma}[*]\label{expt2}
If $\size{V(Q^{>\mu})}=2$, it holds that $\EE{c(Q^{>\mu}_{\mathrm{left}})}\leq p_\gamma\sum_{v\in V(Q^{>\mu})}\lra{c(r,v)+c(v,r)} - \max\lrc{0,2p_\gamma-1}\cdot \Delta_{\max}(Q^{>\mu})$.
\end{lemma}

\begin{lemma}[*]\label{expt3}
If $\size{V(Q^{>\mu})}=3$, it holds that
{\small\begin{align*}
\EE{c(Q^{>\mu}_{\mathrm{left}})}&\leq p_\gamma\sum_{v\in V(Q^{>\mu})}\lrA{c(r,v)+c(v,r)}\\
&\quad\quad - 
\begin{cases}
0, &p_\gamma\in[0,\frac{1}{3}],\\
\lrA{3p_\gamma-1}\cdot\min\lrC{\Delta_{\min}(Q^{>\mu}), \frac{1}{2}\Delta_{\max}(Q^{>\mu})}, &p_\gamma\in[\frac{1}{3},\frac{1}{2}],\\
\min\lrC{(1-p_\gamma)\Delta_{\min}(Q^{>\mu})+(2p_\gamma-1)\Delta_{\max}(Q^{>\mu}), \frac{3p_\gamma-1}{2}\Delta_{\max}(Q^{>\mu})}, &p_\gamma\in[\frac{1}{2},1].
\end{cases}
\end{align*}}
\end{lemma}

Note that the factor-revealing LP used to analyze the approximation ratio of Algorithm~\ref{ralg} is formulated as a maximization problem. 
In this setting, a constraint of the form $\mathrm{LHS} \leq \mathrm{RHS} - \max\lrc{x_{1,1},\ldots,x_{1,j_1}}-\cdots - \max\lrc{x_{i,1},\ldots,x_{i,j_i}}$ can be linearized by introducing auxiliary variables together with additional constraints, similar to the proof of Theorem~\ref{thm3}. 
However, this transformation does not directly apply to inequalities of the form $\mathrm{LHS} \leq \mathrm{RHS} - \min\lrc{\ldots}$, such as the one established in Lemma~\ref{expt3}. 
Indeed, in a maximization problem, an auxiliary variable introduced to represent a minimum term subtracted from the right-hand side would be driven to its smallest feasible value (e.g., zero), thereby rendering the resulting bound ineffective. 
To address this issue, we derive a collection of lower bounds for the $\min\lrc{\ldots}$ expressions in Lemma~\ref{expt3}.

\begin{lemma}[*]\label{expt3+}
If $\size{V(Q^{>\mu})}=3$, it holds that 
{\footnotesize\begin{align*}
\EE{c(Q^{>\mu}_{\mathrm{left}})}&\leq p_\gamma\sum_{v\in V(Q^{>\mu})}\lrA{c(r,v)+c(v,r)}\\
&- 
\begin{cases}
0, &p_\gamma\in[0,\frac{1}{2}],\\
\max\lrC{\lrA{2p_\gamma-1}\Delta_{\max}(Q^{>\mu}), \frac{2p_\gamma-1}{2}\Delta_{\mathrm{sum}}(Q^{>\mu})}, &p_\gamma\in[\frac{1}{2},\frac{2}{3}],\\
\max\lrC{\begin{aligned}
&\lrA{2p_\gamma-1}\Delta_{\max}(Q^{>\mu}),\\
&\lrA{3p_\gamma-2}\Delta_{\max}(Q^{>\mu})+\frac{1-p_\gamma}{2}\Delta_{\mathrm{sum}}(Q^{>\mu}),\frac{3p_\gamma-1}{6}\Delta_{\mathrm{sum}}(Q^{>\mu})
\end{aligned}}, &p_\gamma\in[\frac{2}{3},1].
\end{cases}
\end{align*}}
\end{lemma}

Now, we are ready to analyze $\EE{c(\Q^{>\mu}_{\mathrm{left}})}$ using Lemmas~\ref{expt1}, \ref{expt2}, and \ref{expt3+}, and thus $\EE{c(\T^{>\mu}_{\mathrm{left}})}$ by \eqref{relc}.

\begin{lemma}[*]\label{tmu}
When $\NM=4$ and $I_\mu\in[N/{\NM};N/{2}]$, it holds that 
{\small\begin{align*}
\phi_\gamma(I_\mu)&\leq L(p_\gamma,\psi_1(I_\mu),\psi_2(I_\mu),\psi_3(I_\mu),\psi_4(I_\mu))\\
&\coloneq p_\gamma\psi_1(I_\mu) - \max\lrC{0,2p_\gamma-1}\cdot\psi_2(I_\mu)\\
&\quad\quad-\begin{cases}
0, &p_\gamma\in[0,\frac{1}{2}],\\
\max\lrC{\lrA{2p_\gamma-1}\psi_3(I_\mu), \frac{2p_\gamma-1}{2}\psi_4(I_\mu)}, &p_\gamma\in[\frac{1}{2},\frac{2}{3}],\\
\max\lrC{\lrA{2p_\gamma-1}\psi_3(I_\mu), \lrA{3p_\gamma-2}\psi_3(I_\mu)+\frac{1-p_\gamma}{2}\psi_4(I_\mu),\frac{3p_\gamma-1}{6}\psi_4(I_\mu)}, &p_\gamma\in[\frac{2}{3},1],
\end{cases}
\end{align*}}\!\!
where the functions $\psi_1(I_\mu)$, $\psi_2(I_\mu)$, $\psi_3(I_\mu)$, and $\psi_4(I_\mu)$ satisfy:
{\footnotesize\begin{align*}
&\psi_1(I_\mu)\leq \sum_{i\in[N_{\mathrm{pat}}]}\sum_{j\in[\ell_i]}\lrA{\chi(0, \tilde{p}^{>\mu}_{i,j})+\chi(\tilde{p}^{>\mu}_{i,j},0)},\\
&\psi_2(I_\mu)\geq\sum_{i\in[N_{\mathrm{pat}}]:\ell_i = 2} \lrA{\chi(\tilde{p}^{>\mu}_{i,1},0)+\chi(0, \tilde{p}^{>\mu}_{i,2}) - \chi(\tilde{p}^{>\mu}_{i,1}, \tilde{p}^{>\mu}_{i,2})},\\
&\psi_3(I_\mu)\geq \sum_{i\in[N_{\mathrm{pat}}]:\ell_i = 3}\max
\left\{
\begin{array}{l}
\chi(\tilde{p}^{>\mu}_{i,1},0)+\chi(0, \tilde{p}^{>\mu}_{i,2}) - \chi(\tilde{p}^{>\mu}_{i,1}, \tilde{p}^{>\mu}_{i,2}),\\
\chi(\tilde{p}^{>\mu}_{i,1},0)+\chi(0, \tilde{p}^{>\mu}_{i,3}) - \chi(\tilde{p}^{>\mu}_{i,1}, \tilde{p}^{>\mu}_{i,3}),\\
\chi(\tilde{p}^{>\mu}_{i,2},0)+\chi(0, \tilde{p}^{>\mu}_{i,3}) - \chi(\tilde{p}^{>\mu}_{i,2}, \tilde{p}^{>\mu}_{i,3})
\end{array}
\right\},
\\
&\psi_4(I_\mu)\geq\sum_{i\in[N_{\mathrm{pat}}]:\ell_i = 3} \lrA{2\chi(\tilde{p}^{>\mu}_{i,1},0)+\chi(0, \tilde{p}^{>\mu}_{i,2})+\chi(\tilde{p}^{>\mu}_{i,2},0)+2\chi(\tilde{p}^{>\mu}_{i,3},0) - \chi(\tilde{p}^{>\mu}_{i,1}, \tilde{p}^{>\mu}_{i,2})-\chi(\tilde{p}^{>\mu}_{i,1}, \tilde{p}^{>\mu}_{i,3})-\chi(\tilde{p}^{>\mu}_{i,2}, \tilde{p}^{>\mu}_{i,3})}.
\end{align*}}
\end{lemma}

All constraints on $\psi_i(I_\mu)$ for $i\in[4]$ in Lemma~\ref{tmu} can be linearized into L-constraints.

We next establish the connection lemma, analogous to Lemma~\ref{bridge}. 

Recall that each vertex $i\in\mathbb{V}_1\cup\cdots\cup\mathbb{V}_{N_{\mathrm{pat}}}\setminus\{0\}$ with $i\in\mathbb{V}_j$ corresponds to the set of identical-type customers $\{v^{z}_{i}\}_{z\in[\size{\Q_j}]}$.
For each $i\in[N/2-N/\NM+1]$, let $\mathbb{J}_{i}$ denote the set of vertices in $\mathbb{V}_1\cup\cdots\cup\mathbb{V}_{N_{\mathrm{pat}}}\setminus\{0\}$ whose mapped customer sets consist only of type-$i$ customers.
Consequently, the set $\lrc{\mathbb{J}_1,\ldots,\mathbb{J}_{N/2-N/\NM+1}}$ forms a partition of $\mathbb{V}_1\cup\cdots\cup\mathbb{V}_{N_{\mathrm{pat}}}\setminus\{0\}$.
This leads to the following result.

\begin{lemma}[*]\label{bridge+}
It holds that 
\begin{itemize}
  \item $\xi_0(N) - \xi_0(N/2) \geq \kappa(N/2)$, where $\kappa(N/2)\coloneq \sum_{j\in \mathbb{J}_{N/2-N/\NM+1}}\lrA{\chi(0,j)+\chi(j,0)}$;
  \item $\xi_0(N/\NM+i) - \xi_0(N/\NM+i-1) \geq \kappa(N/\NM+i-1)$, for $i\in[N/2-N/\NM]$, where $\kappa(N/\NM+i-1)\coloneq \sum_{j\in \mathbb{J}_{i}}\lrA{\chi(0,j)+\chi(j,0)}$.
\end{itemize}
\end{lemma}

\begin{theorem}\label{thm4}
When $M=1000$ and $N=240$, Algorithm~\ref{ralg} achieves an expected approximation ratio of about $\alpha+1.647$ for the unsplittable CVRP with general vehicle capacity, and about $\alpha+1.548$ for the unsplittable CVRP with fixed vehicle capacity.
\end{theorem}
\begin{proof}
We analyze the cases of general vehicle capacity and fixed vehicle capacity, respectively.

\paragraph{\textbf{General Vehicle Capacity.}}
Recall that Algorithm~\ref{ralg} runs UITP-$(\gamma,\sigma,\delta,\mu)$ with $\sigma=\varepsilon$ by enumerating the parameters $\gamma$, $\delta$, and $\mu$.
We consider a special case, where $\NM=4$, $\gamma$ is enumerated only over $\lrc{0,\frac{1}{M},\ldots,1}$, $\delta$ only over $\lrc{\varepsilon}\cup\lrc{\frac{1}{\NM},\frac{N/\NM+1}{N},\ldots,\frac{1}{3}}$, and $\mu$ only over $\lrc{\frac{1}{\NM},\frac{N/\NM+1}{N},\ldots,\frac{1}{2}}$ with $\mu\geq\delta$.
We evaluate the expected approximation ratio of this restricted algorithm, denoted by \textsc{R-General-Alg}, which naturally provides an upper bound on the expected approximation ratio of the full Algorithm~\ref{ralg}.
Moreover, similar to the proof of Theorem~\ref{thm3}, to avoid introducing the parameter $\varepsilon$ into the factor-revealing LP, we omit the first constraint in \eqref{intxi} and relax the third constraint therein to $\xi_1(1)-\xi_1(\varepsilon\cdot N) \geq 0$.

By \eqref{ratio21}, \eqref{ratio22}, and Lemmas~\ref{zwp+}, \ref{tmu}, and \ref{bridge+}, the expected approximation ratio of \textsc{R-General-Alg} is at most $\pi + \alpha + 3\varepsilon$, where $\pi$ is the optimal value of the following LP:
{\small\setlength{\jot}{2pt}\begin{align*}
\max\quad &\pi \\
\text{s.t.}\quad 
&\pi\leq \zeta_2(I_\gamma,I_\delta, I_\mu),
    &&\forall\ I_\gamma\in \{0\}\cup [M],\ \forall\ I_\delta\in\lrC{\varepsilon\cdot N}\cup[N/\NM; N/3],\\
& &&\forall\ I_\mu\in[\max\lrc{I_\delta, N/\NM}; N/2],\\
&\phi_\gamma(I_\mu)\leq L(p_\gamma,\psi_1(I_\mu),\psi_2(I_\mu),\psi_3(I_\mu),\psi_4(I_\mu)),
&&\forall\ I_\gamma\in \{0\}\cup [M],\ \forall\ I_\mu\in[N/\NM; N/2],\\
&\xi_0(N)-\xi_0(N/2) \geq \kappa(N/2),\\
&\xi_0(i+1)-\xi_0(i) \geq \kappa(i),
  &&\forall\ i\in[N/\NM; N/2-1],\\
&\rlap{\text{the modified L-constraints on $\xi_0$ and $\xi_1$ in \eqref{intxi} and the L-constraints on $\chi$ in Lemma~\ref{zwp+}},}\\
&\rlap{\text{the L-constraints on $\psi_i(I_\mu)$ derived by linearizing the constraints in Lemma~\ref{tmu}}.}
\end{align*}}

\paragraph{\textbf{Fixed Vehicle Capacity.}}
In this case, $\sigma=0$.
We consider a special case, where $\gamma$ is enumerated only over $\lrc{0,\frac{1}{M},\ldots,1}$, $\delta$ only over $\lrc{\frac{1}{\NM},\allowbreak\frac{N/\NM+1}{N},\allowbreak\ldots,\allowbreak\frac{1}{3}}$, and $\mu$ only over $\lrc{\frac{1}{\NM},\allowbreak\frac{N/\NM+1}{N},\allowbreak\ldots,\allowbreak\frac{1}{2}}$ with $\mu\geq\delta$.
As before, the expected approximation ratio of this restricted algorithm, denoted by \textsc{R-Fixed-Alg}, also provides an upper bound on the expected approximation ratio of Algorithm~\ref{ralg}.

The expected approximation ratio of \textsc{R-Fixed-Alg} is at most $\pi + \alpha$, where $\pi$ is the optimal value of a factor-revealing LP that is almost identical to the one for \textsc{R-General-Alg}.
The differences are twofold: (i) the upper bound constraint on $\pi$ is replaced by $\pi\leq \zeta_1(I_\gamma,I_\delta, I_\mu)$; and (ii) because the parameters $\gamma$, $\delta$, and $\mu$ do not depend on $\varepsilon$ in this case, we simply omit the first four constraints in \eqref{intxi}.

\paragraph{\textbf{Numerical Results.}}

Setting $M=1000$, the optimal values and running times of the factor-revealing LPs for both \textsc{R-General-Alg} and \textsc{R-Fixed-Alg} under different values of $N$ are summarized in Table~\ref{res1}.
Note that we use $M=1000$, as the results barely change even when $M$ is increased to $1500$; see Appendix~\ref{additional}.

\begin{table}[t]
\centering
\small
\renewcommand{\arraystretch}{0.89}
\caption{
Values and running times of the factor-revealing LPs for \textsc{R-General-Alg} and \textsc{R-Fixed-Alg} with varying $N$.
}
\begin{tabular}{ccccc}
\toprule
\multirow{2}{*}{$N$} & \multicolumn{2}{c}{\textsc{R-General-Alg}} & \multicolumn{2}{c}{\textsc{R-Fixed-Alg}}\\
\cmidrule{2-5}
& Value & Runtime (s) & Value & Runtime (s)\\
\midrule
\midrule
$60$ & $1.650110$ & $46$ & $1.552077$ & $13$ \\
$120$ & $1.647908$ & $642$ & $1.549126$ & $286$ \\
$180$ & $1.647120$ & $4664$ & $1.548119$ & $1149$ \\
$240$ & $1.646793$ & $10111$ & $1.547613$ & $3993$ \\
$300$ & $1.646325$ & $32053$ & $1.547308$ & $7231$ \\
$360$ & $1.646418$ & $54449$ & $1.547104$ & $23290$ \\
\bottomrule
\end{tabular}
\label{res1}
\end{table}

As shown in Table~\ref{res1}, the upper bounds generally improve as the grid size $N$ increases. For the general capacity case, however, the objective value exhibits a slight increase when $N$ is increased from $300$ to $360$ (from $1.646325$ to $1.646418$). This slight fluctuation is likely due to numerical imprecision: as the size of the LP increases, the constraint matrix becomes increasingly ill-conditioned, approaching the limits of the solver's numerical precision despite the robust parameter settings used in our computational experiments.
(Detailed solver parameter settings are provided in the \texttt{README} file at \url{https://github.com/JingyangZhao/unsplittableCVRP}.)

To mitigate such numerical effects, we conservatively adopt the results obtained at $N=240$.

For the general vehicle capacity case, the factor-revealing LP at $N=240$ yields the objective value $1.646793$. Choosing $\varepsilon$ sufficiently small gives $\pi+3\varepsilon < 1.647$, suggesting an expected approximation ratio of about $\alpha+1.647$. Similarly, for the fixed vehicle capacity case, the factor-revealing LP at $N=240$ yields the objective value $1.547613< 1.548$, suggesting an expected approximation ratio of about $\alpha+1.548$.
\end{proof}

\section{Derandomization via A Priori 2-Tour Sets}\label{sec5}

In the previous analysis of Algorithm~\ref{ralg}, to analyze the expected cost $\EE{c(\T^{>\mu}_{\mathrm{left}})}$ of UITP-$(\gamma,\sigma,\delta,\mu)$, we first constructed a 3-tour set $\Q^{>\mu}$ using $\T_{\mathrm{large}}$, then a 2-tour set $\Q^{>\mu}_{\mathrm{left}}$ serving all customers in $V^{>\mu}_{\mathrm{left}}$ from $\Q^{>\mu}$, and finally obtained an upper bound of $\EE{c(\T^{>\mu}_{\mathrm{left}})}$ by using $\EE{c(\Q^{>\mu}_{\mathrm{left}})}$.

To derandomize UITP-$(\gamma,\sigma,\delta,\mu)$ while preserving the approximation ratio, a natural approach is to first compute a 3-tour set $\Z^{>\mu}$ serving all customers in $V^{>\mu}$ such that the 2-tour set $\Z^{>\mu}_{\mathrm{left}}$ serving all customers in $V^{>\mu}_{\mathrm{left}}$, obtained from $\Z^{>\mu}$ using the same method of constructing $\Q^{>\mu}_{\mathrm{left}}$, has the minimum expected cost.
However, it is unclear how to design a polynomial-time algorithm to compute such a 3-tour set.

To overcome this issue, we allow a slightly weaker approximation ratio and instead begin by computing a $2$-tour set $\Z^{>\mu}$ serving all customers in $V^{>\mu}$ such that the expected cost of the residual $2$-tour set $\Z^{>\mu}_{\mathrm{left}}$ is minimized, where $\Z^{>\mu}_{\mathrm{left}}$ is obtained from $\Z^{>\mu}$ by shortcutting all customers in $V^{>\mu}\setminus V^{>\mu}_{\mathrm{left}}$.

Such a 2-tour set can be computed efficiently via a reduction to the MW$2$SCP. 
In particular, for any 1-tour $(r,v,r)\in\Z^{>\mu}$, its expected contribution to $\EE{c(\Z^{>\mu}_{\mathrm{left}})}$ is $\PP{v\in V^{>\mu}_{\mathrm{left}}}\cdot \lrA{c(r,v)+c(v,r)}$, where $\PP{v\in V^{>\mu}_{\mathrm{left}}}=\prod_{T\in\T: v\in V(T)}(1-\min\{1,\gamma\cdot x^*_T\})$ by the LP-rounding procedure.
Then, we create a 1-set $\{v\}$ whose weight equals this contribution.
For any 2-tour $(r,v_1,v_2,r)\in\Z^{>\mu}$, its expected contribution is
{\small\begin{align*}
&\PP{v_1,v_2\in V^{>\mu}_{\mathrm{left}}}\cdot \lrA{c(r,v_1)+c(v_1,v_2)+c(v_2,r)}+\lrA{\PP{v_1\in V^{>\mu}_{\mathrm{left}}}-\PP{v_1,v_2\in V^{>\mu}_{\mathrm{left}}}}\cdot \lrA{c(r,v_1)+c(v_1,r)}\\
&\quad + \lrA{\PP{v_2\in V^{>\mu}_{\mathrm{left}}}-\PP{v_1,v_2\in V^{>\mu}_{\mathrm{left}}}}\cdot\lrA{c(r,v_2)+c(v_2,r)},
\end{align*}}\!\!
where $\PP{v_1,v_2\in V^{>\mu}_{\mathrm{left}}}=\prod_{T\in\T: \lrc{v_1,v_2}\cap V(T)\neq \emptyset}(1-\min\{1,\gamma\cdot x^*_T\})$, and the probabilities $\PP{v_1\in V^{>\mu}_{\mathrm{left}}}$ and $\PP{v_2\in V^{>\mu}_{\mathrm{left}}}$ can be computed similarly.
Then, we create a 2-set $\{v_1,v_2\}$ whose weight equals this contribution.
The resulting construction defines an instance of MW$2$SCP, and solving this instance yields the desired $2$-tour set $\Z^{>\mu}$ minimizing $\EE{c(\Z^{>\mu}_{\mathrm{left}})}$.

\begin{lemma}[*]\label{dmitp}
Given an $\alpha$-approximate TSP tour on $V\cup \lrc{r}$ and constants $\gamma,\sigma,\delta,\mu\geq 0$ with $\sigma\leq\delta\leq\mu\leq\frac{1}{2}$, there exists a deterministic algorithm, denoted by D-UITP-$(\gamma,\sigma,\delta,\mu)$, that computes a solution to the unsplittable CVRP with cost at most
{\footnotesize\begin{equation}\label{dup}
\begin{split}
&\alpha\cdot \TSP + \RLB\cdot\frac{1}{1-\delta}\df{1}{\sigma}\\
&\ +\gamma\cdot\OPT+p_\gamma\cdot\RLB\cdot\lrB{\frac{1}{1-\delta}\lrA{\df{1}{\delta}-\df{1}{\sigma}} + \frac{2}{1-\delta}\lrA{\df{1}{\mu}-\df{1}{\delta}} - \frac{\delta}{1-\delta}\lrA{\df{0}{\mu}-\df{0}{\delta}}}+\EEE{c(\Z^{>\mu}_{\mathrm{left}})}.
\end{split}
\end{equation}}

Moreover, the running time is $n^{\O(1/\sigma)}$ when $\sigma>0$, and $n^{\O(k)}$ when $\sigma=0$.
\end{lemma}

Finally, by replacing UITP-$(\gamma,\sigma,\delta,\mu)$ in Algorithm~\ref{ralg} with D-UITP-$(\gamma,\sigma,\delta,\mu)$, we obtain a deterministic algorithm for the unsplittable CVRP.

Define $\tilde{\phi}_\gamma(I_\mu)\coloneq \frac{\EE{c(\Z^{>\mu}_{\mathrm{left}})}}{\OPT}$ for any $I_\mu\in[N/2]$.
We obtain the following result analogous to Lemma~\ref{tmu}.

\begin{lemma}[*]\label{tmu+}
When $\NM=4$ and $I_\mu\in[N/{\NM};N/{2}]$, it holds that 
{\small\begin{align*}
\tilde{\phi}_\gamma(I_\mu) \leq \tilde{L}(p_\gamma,\psi_1(I_\mu),\psi_2(I_\mu),\psi_3(I_\mu))\coloneq p_\gamma\psi_1(I_\mu) - \max\lrC{0,2p_\gamma-1}\cdot\lrA{\psi_2(I_\mu)+\psi_3(I_\mu)},
\end{align*}}\!\!
and the functions $\psi_1(I_\mu)$, $\psi_2(I_\mu)$, and $\psi_3(I_\mu)$ satisfy the bounds in Lemma~\ref{tmu}.
\end{lemma}

\begin{theorem}\label{thm5}
When $M=1000$ and $N=240$, the deterministic algorithm achieves an approximation ratio of about $\alpha+1.648$ for the unsplittable CVRP with general vehicle capacity, and about $\alpha+1.548$ for the unsplittable CVRP with fixed vehicle capacity.
\end{theorem}
\begin{proof}
We analyze the cases of general vehicle capacity and fixed vehicle capacity, respectively.

\paragraph{\textbf{General Vehicle Capacity.}}
The deterministic algorithm runs D-UITP-$(\gamma,\sigma,\delta,\mu)$ with $\sigma=\varepsilon$ by enumerating $\gamma$, $\delta$, and $\mu$.
Similar to the proof of Theorem~\ref{thm4}, we consider a special case, where $\NM=4$, $\gamma$ is enumerated only over $\lrc{0,\frac{1}{M},\ldots,1}$, $\delta$ only over $\lrc{\varepsilon}\cup\lrc{\frac{1}{\NM},\frac{N/\NM+1}{N},\ldots,\frac{1}{3}}$, and $\mu$ only over $\lrc{\frac{1}{\NM},\frac{N/\NM+1}{N},\ldots,\frac{1}{2}}$ with $\mu\geq\delta$.
We evaluate the approximation ratio of this restricted algorithm, denoted by \textsc{D-General-Alg}, which naturally provides an upper bound on the approximation ratio of the full deterministic algorithm.

Recall that $\zeta_2(I_\gamma,I_\delta, I_\mu)$ is defined in \eqref{ratio21} and \eqref{ratio22}.
By replacing $\phi_\gamma(I_\mu)$ in the expression of $\zeta_2(I_\gamma,I_\delta, I_\mu)$ with $\tilde{\phi}_\gamma(I_\mu)$, we define $\tilde{\zeta}_2(I_\gamma,I_\delta, I_\mu)$. Then, by the proof of Lemma~\ref{ratio}, the approximation ratio of D-UITP-$(\gamma,\sigma,\delta,\mu)$ is at most $\tilde{\zeta}_2(I_\gamma,I_\delta,I_\mu)+\alpha+3\varepsilon$ for the case of general vehicle capacity (by setting $\sigma=\varepsilon$).

As mentioned in the proof of Theorem~\ref{thm4}, all relevant constraints on $\psi_i(I_\mu)$ in Lemma~\ref{tmu} can be transformed into L-constraints.
Moreover, to avoid introducing the parameter $\varepsilon$ into the factor-revealing LP, we omit the first constraint in \eqref{intxi} and relax the third constraint therein to $\xi_1(1)-\xi_1(\varepsilon\cdot N) \geq 0$.

Similarly, by Lemmas~\ref{zwp+}, \ref{bridge+}, and \ref{tmu+}, the approximation ratio of \textsc{D-General-Alg} is at most $\pi + \alpha + 3\varepsilon$, where $\pi$ is the optimal value of the following LP:
{\small\setlength{\jot}{2pt}\begin{align*}
\max\quad &\pi \\
\text{s.t.}\quad 
&\pi\leq \tilde{\zeta}_2(I_\gamma,I_\delta, I_\mu),
    &&\forall\ I_\gamma\in \{0\}\cup [M],\ \forall\ I_\delta\in\lrC{\varepsilon\cdot N}\cup[N/\NM; N/3],\\
& &&\forall\ I_\mu\in[\max\lrc{I_\delta, N/\NM}; N/2],\\
&\tilde{\phi}_\gamma(I_\mu)\leq \tilde{L}(p_\gamma,\psi_1(I_\mu),\psi_2(I_\mu),\psi_3(I_\mu)),
&&\forall\ I_\gamma\in \{0\}\cup [M],\ \forall\ I_\mu\in[N/\NM; N/2],\\
&\xi_0(N)-\xi_0(N/2) \geq \kappa(N/2),\\
&\xi_0(i+1)-\xi_0(i) \geq \kappa(i),
  &&\forall\ i\in[N/\NM; N/2-1],\\
&\rlap{\text{the modified L-constraints on $\xi_0$ and $\xi_1$ in \eqref{intxi} and the L-constraints on $\chi$ in Lemma~\ref{zwp+}},}\\
&\rlap{\text{the L-constraints on $\psi_i(I_\mu)$ derived by linearizing the constraints in Lemma~\ref{tmu}}.}
\end{align*}}

\paragraph{\textbf{Fixed Vehicle Capacity.}}
In this case, we set $\sigma=0$.
We consider a special case, where $\gamma$ is enumerated only over $\lrc{0,\frac{1}{M},\ldots,1}$, $\delta$ only over $\lrc{\frac{1}{\NM},\frac{N/\NM+1}{N},\ldots,\frac{1}{3}}$, and $\mu$ only over $\lrc{\frac{1}{\NM},\frac{N/\NM+1}{N},\ldots,\frac{1}{2}}$ with $\mu\geq\delta$.
As before, the approximation ratio of this restricted algorithm, denoted by \textsc{D-Fixed-Alg}, also provides an upper bound on the approximation ratio of the full deterministic algorithm.

Similarly, replacing $\phi_\gamma(I_\mu)$ in the expression of $\zeta_1(I_\gamma,I_\delta, I_\mu)$ in \eqref{ratio1} with $\tilde{\phi}_\gamma(I_\mu)$, define $\tilde{\zeta}_1(I_\gamma,I_\delta, I_\mu)$. Then, by the proof of Lemma~\ref{ratio}, the approximation ratio of D-UITP-$(\gamma,\sigma,\delta,\mu)$ is at most $\tilde{\zeta}_1(I_\gamma,I_\delta,I_\mu)+\alpha$ for the case of fixed vehicle capacity (by setting $\sigma=0$).

The approximation ratio of \textsc{D-Fixed-Alg} is at most $\pi + \alpha$, where $\pi$ is the optimal value of a factor-revealing LP that is almost identical to the one for \textsc{D-General-Alg}. The differences are twofold: (i) the upper bound constraint on $\pi$ is replaced by $\pi\leq \tilde{\zeta}_1(I_\gamma,I_\delta,I_\mu)$; and (ii) because the parameters $\gamma$, $\delta$, and $\mu$ do not depend on $\varepsilon$ in this case, we simply omit the first four constraints in \eqref{intxi}.

\paragraph{\textbf{Numerical Results.}}
Setting $M=1000$, the optimal values and running times of the factor-revealing LPs for both \textsc{D-General-Alg} and \textsc{D-Fixed-Alg} under different values of $N$ are summarized in Table~\ref{res2}.

\begin{table}[t]
\centering
\small
\renewcommand{\arraystretch}{0.89}
\caption{
Values and running times of the factor-revealing LPs for \textsc{D-General-Alg} and \textsc{D-Fixed-Alg} with varying $N$.
}
\label{res2}
\begin{tabular}{ccccc}
\toprule
\multirow{2}{*}{$N$} & \multicolumn{2}{c}{\textsc{D-General-Alg}} & \multicolumn{2}{c}{\textsc{D-Fixed-Alg}}\\
\cmidrule{2-5}
& Value & Runtime (s) & Value & Runtime (s)\\
\midrule
\midrule
$60$ & $1.651271$ & $12$ & $1.552437$ & $9$ \\
$120$ & $1.648878$ & $302$ & $1.549425$ & $246$ \\
$180$ & $1.648018$ & $1951$ & $1.548403$ & $1216$ \\
$240$ & $1.647607$ & $7701$ & $1.547889$ & $3368$ \\
$300$ & $1.647136$ & $30106$ & $1.547581$ & $9636$ \\
$360$ & $-1$ & $40586$ & $1.547375$ & $16651$ \\
\bottomrule
\end{tabular}
\end{table}

As shown in Table~\ref{res2}, the upper bounds generally improve as the grid size $N$ increases. For the general capacity case, however, the solver fails to return a valid optimal value at $N=360$ (indicated by $-1$), likely due to numerical instability. Consistent with the discussion in the proof of Theorem~\ref{thm4}, we therefore conservatively adopt the results obtained at $N=240$.

For the general vehicle capacity case, the factor-revealing LP at $N=240$ yields the objective value $1.647607$. Choosing $\varepsilon$ sufficiently small gives $\pi+3\varepsilon < 1.648$, suggesting an approximation ratio of about $\alpha+1.648$. Similarly, for the fixed vehicle capacity case, the factor-revealing LP at $N=240$ yields the objective value $1.547889<1.548$, suggesting an approximation ratio of about $\alpha+1.548$.
\end{proof}

We remark that several additional properties can be exploited in the metric case.
First, since a feasible tour pattern $\lrd{i_1,i_2,\ldots,i_\ell}$ and its reverse $\lrd{i_\ell,i_{\ell-1},\ldots,i_1}$ are equivalent by symmetry, the number of distinct patterns is reduced. 
For instance, when $N=240$, $N_{\mathrm{pat}}$ decreases from $41601$ to $21326$ for $\NM=4$, and from $400497$ to $201575$ for $\NM=5$, compared to the asymmetric case. 
Second, the bounds in (b) of Lemma~\ref{zwp+} can be strengthened to $\chi(i,j) = \chi(j,i) \leq \chi(i,\ell) + \chi(\ell, j)$ by symmetry. 
Somewhat surprisingly, however, numerical computations indicate that these two properties do not yield notable further improvements in the final approximation ratios (see Appendix~\ref{additional} for further discussion).

Moreover, \citet{journal:blauth:MP:23} proved that for metric CVRP instances with $\RLB\geq (1-\varepsilon)\cdot\OPT$, there exists a function $f: \mathbb{R}_{>0}\rightarrow\mathbb{R}_{>0}$ with $\lim_{{\varepsilon\rightarrow 0}}f(\varepsilon)=0$ and a polynomial-time algorithm that computes a TSP tour $\A_\varepsilon$ on $V\cup\lrc{r}$ such that $c(\A_\varepsilon)\leq (1+f(\varepsilon))\cdot\OPT$.
We believe that the approximation ratios in Theorems~\ref{thm4} and \ref{thm5} for the metric case can be further improved using their result: if $\RLB\geq (1-\varepsilon)\cdot\OPT$, the improvement can be achieved via $\A_\varepsilon$; otherwise, the condition $\RLB\leq (1-\varepsilon)\cdot\OPT$ strengthens the bounds in Lemma~\ref{bridge+}, which numerically leads to improved values of the factor-revealing LPs.

\section{Conclusion}\label{sec6}

For both the asymmetric and metric unsplittable CVRP, we improve the previous best approximation ratio from about $\alpha+1.694$ to about $\alpha+1.648$ (or $\alpha+1.647$ if randomness is allowed) for the general-capacity case, and to about $\alpha+1.548$ for the fixed-capacity case. Our results are achieved through three main ingredients: a unified algorithmic framework that generalizes and integrates existing LP-rounding and matching-based techniques; a novel analytical framework based on demand discretization and factor-revealing LPs; and a new derandomization technique utilizing a priori 2-tour sets and the method of conditional expectations.

A natural direction for future research is to further improve these approximation ratios. For instance, while our analysis currently focuses on $\NM=4$, considering larger values may yield stronger guarantees.
However, this introduces several technical challenges.
The matching-based approach relies on a 2-tour set to serve customers with demands greater than $\frac{1}{\NM}$.
When $\NM>4$, these customers may be served by 4-tours in an optimal solution, and therefore using a 2-tour set to serve them is likely to incur a relatively large cost, making the matching-based approach less effective.
Moreover, as shown in Lemma~\ref{expt3+}, the analysis is already quite involved when $\NM=4$, and deriving strong structural bounds becomes even more challenging for larger values such as $\NM=5$.
Furthermore, the size of the factor-revealing LP grows exponentially with $\NM$, which may limit practical computations to relatively small values of $N$ when $\NM>4$, and consequently make further improvements difficult to achieve.

To bypass these limitations, a promising direction is to develop techniques that allow the use of 3-tours to serve customers with large demands.
However, computing a minimum-cost 3-tour set is NP-hard~\citep{conference:gupta:SOSA:23}, and directly applying existing approximation algorithms appears unlikely to be effective.

Finally, our techniques may be extendable to other variants of the CVRP, such as the unsplittable capacitated arc routing problem, where demands are associated with arcs rather than vertices~\citep{journal:jansen:NETWORKS:93}.
It would be interesting to explore whether similar ideas can yield improved approximation ratios for these variants and to investigate the broader applicability of our methods to other routing problems.

\section*{Acknowledgement}
This version was submitted for publication in May 2026 and rejected in August 2026.
For the metric unsplittable CVRP with general capacity, \citet{journal:ebert:arXiv:26} recently improved the approximation ratio from approximately $3.176$, achieved in our initial version, to $3.159$.
The authors have decided to release this draft publicly.
The key ideas in this work were developed by the authors, not by AI.

\bibliographystyle{apalike} 
\bibliography{CVRP}

\newpage
\appendix
\normalsize

\section{Omitted Proofs}
\subsection{Proof of Lemma~\ref{lpitp}}\label{plpitp}
\begin{proof}
By line~\ref{alg1s7} of Algorithm~\ref{lpalg}, $\T_{\mathrm{round}}\subseteq \T_{\mathrm{f}}$, where $\T_{\mathrm{f}}$ denotes the set of tours returned by LP-ITP-$(\gamma,\sigma,\delta)$, and $\T_{\mathrm{round}}$ denotes the set of tours selected by the randomized LP-rounding procedure.

Let $\A$ denote the TSP tour on $V_{\mathrm{left}}\cup\lrc{r}$ obtained by shortcutting the given $\alpha$-approximate TSP tour on $V\cup\lrc{r}$.
We first analyze the expected cost of $\T_{\mathrm{f}}\setminus \T_{\mathrm{round}}$.
{\small\begin{align*}
&\EEE{c(\T_{\mathrm{f}}\setminus \T_{\mathrm{round}})}\\
&\leq c(\A) + \frac{1}{1-\delta}\sum_{v\in \V^\delta_0}\PPP{v\in V_{\mathrm{left}}}\cdot d_v\cdot \lrA{c(r,v)+c(v,r)}+ \frac{2}{1-\delta}\sum_{v\in \V^{1/2}_\delta}\PPP{v\in V_{\mathrm{left}}}\cdot d_v\cdot \lrA{c(r,v)+c(v,r)}\\
&\quad\quad - \frac{\delta}{1-\delta}\sum_{v\in \V^{1/2}_\delta}{\PPP{v\in V_{\mathrm{left}}}\cdot \lrA{c(r,v)+c(v,r)}}+\sum_{v\in \V^1_{1/2}}\PPP{v\in V_{\mathrm{left}}}\cdot \lrA{c(r,v)+c(v,r)}\\
&\leq c(\A) + \frac{1}{1-\delta}\sum_{v\in \V^\sigma_0} d_v\cdot \lrA{c(r,v)+c(v,r)}+ p_\gamma\cdot\lrAl{\frac{1}{1-\delta}\sum_{v\in \V^\delta_\sigma} d_v\cdot \lrA{c(r,v)+c(v,r)}}\\
&\quad\quad +\lrAr{\frac{2}{1-\delta}\sum_{v\in \V^{1/2}_\delta} d_v\cdot \lrA{c(r,v)+c(v,r)} - \frac{\delta}{1-\delta}\sum_{v\in \V^{1/2}_\delta}{ \lrA{c(r,v)+c(v,r)}}+\sum_{v\in \V^1_{1/2}} \lrA{c(r,v)+c(v,r)}}\\
&= c(\A) + \RLB\cdot\frac{\frac{1}{1-\delta}\sum_{v\in \V^\sigma_0} d_v\cdot \lrA{c(r,v)+c(v,r)}}{\RLB}+p_\gamma\cdot\RLB\cdot\lrAl{\frac{\frac{1}{1-\delta}\sum_{v\in \V^\delta_\sigma} d_v\cdot \lrA{c(r,v)+c(v,r)}}{\RLB}}\\
&\quad\quad+\lrAr{\frac{\frac{2}{1-\delta}\sum_{v\in \V^{1/2}_\delta} d_v\cdot \lrA{c(r,v)+c(v,r)}  - \frac{\delta}{1-\delta}\sum_{v \in \V^{1/2}_\delta}{ \lrA{c(r,v)+c(v,r)}}+\sum_{v \in \V^1_{1/2}} \lrA{c(r,v)+c(v,r)}}{\RLB}}\\
&\leq \alpha\cdot \TSP + \RLB\cdot\frac{1}{1-\delta}\df{1}{\sigma}+p_\gamma\cdot\RLB\cdot\lrBl{\frac{1}{1-\delta}\lrA{\df{1}{\delta}-\df{1}{\sigma}}}\\
&\quad\quad + \lrBr{\frac{2}{1-\delta}\lrA{\df{1}{\frac{1}{2}}-\df{1}{\delta}}-\frac{\delta}{1-\delta}\lrA{\df{0}{\frac{1}{2}}-\df{0}{\delta}} + \lrA{\df{0}{1}-\df{0}{\frac{1}{2}}}},
\end{align*}}\!\!
where the first inequality follows from Lemmas~\ref{int} and \ref{nditp}, using the fact that $\T_{\mathrm{f}}\setminus \T_{\mathrm{round}}$ is obtained by applying ITP-$\delta$ with $\A$ on $V_{\mathrm{left}}\cup\lrc{r}$;
the second inequality follows from Lemma~\ref{prob}; and the last inequality uses both the fact that $\A$ is obtained by shortcutting an $\alpha$-approximate TSP tour on $V\cup\lrc{r}$ and the definition in \eqref{int}.

Therefore, the claimed upper bound on the expected cost of the solution returned by LP-ITP-$(\gamma,\sigma,\delta)$ follows from the above upper bound on $\EE{c(\T_{\mathrm{f}}\setminus \T_{\mathrm{round}})}$ together with $\EE{c(\T_{\mathrm{round}})} \leq \gamma\cdot\OPT$ by Lemma~\ref{prob}.

Finally, the running time of LP-ITP-$(\gamma,\sigma,\delta)$ is dominated by the time required to solve the LP~\eqref{reflp}, which is $n^{\O(1/\sigma)}$ when $\sigma>0$ and $n^{\O(k)}$ when $\sigma=0$.
\end{proof}

\subsection{Proof of Lemma~\ref{fact1}}\label{pfact1}
\begin{proof}
When $\mu=\frac{1}{2}$, every customer in $V^{>\mu}_{\mathrm{left}}$ has demand greater than $\frac{1}{2}$.
Hence, $\T^{>\mu}_{\mathrm{left}}$ consists of trivial tours serving the customers in $V^{>\mu}_{\mathrm{left}}$.
Consequently, 
{\small\begin{align*}
\EEE{c(\T^{>\mu}_{\mathrm{left}})}=\sum_{v\in\V^1_\mu}\PPP{v\in V_{\mathrm{left}}}\cdot \lrA{c(r,v)+c(v,r)} \leq  \sum_{v\in\V^1_\mu}p_\gamma\cdot \lrA{c(r,v)+c(v,r)} = p_\gamma\cdot\RLB\cdot\lra{\df{0}{1}-\df{0}{\mu}},
\end{align*}}
where the inequality follows from Lemma~\ref{prob}, and the last equality follows from \eqref{int}.

Moreover, by Lemma~\ref{nditp}, ITP-$\delta$ serves each customer with demand greater than $\frac{1}{2}$ using a trivial tour.
Thus, by Algorithms~\ref{lpalg} and \ref{malg}, UITP-$(\gamma,\sigma,\delta,\mu)$ coincides with LP-ITP-$(\gamma,\sigma,\delta)$ when $\mu=\frac{1}{2}$.
\end{proof}

\subsection{Proof of Lemma~\ref{inimr}}\label{pinimr}
\begin{proof}
When $\mu \in[\frac{1}{3},\frac{1}{2}]$, by shortcutting all customers with demands at most $\mu$ from an optimal solution to the unsplittable CVRP, one obtains a 2-tour set that serves all customers with demands greater than $\mu$, since each tour can serve at most two customers whose demands are greater than $\mu$.
By the triangle inequality and the optimality of $\M^{>\mu}$, we have $c(\M^{>\mu})\leq \OPT$. Similarly, we have $c(\T^{>\mu}_{\mathrm{left}})\leq c(\M^{>\mu})$.
\end{proof}

\subsection{Proof of Theorem~\ref{thm1}}\label{pthm1}
\begin{proof}
Using the given $\alpha$-approximate TSP tour, the algorithm runs UITP-$(0,-,\frac{1}{3},\frac{1}{3})$ and UITP-$(\gamma,0,\frac{1}{3},\frac{1}{2})$ to obtain two solutions, where the parameter $\gamma\geq 0$ is defined later, and returns the one with the smaller cost.

We remark that when $\gamma=0$, the randomized LP-rounding step can be skipped, and thus UITP-$(0,-,\frac{1}{3},\frac{1}{3})$ is deterministic.
Moreover, by Lemma~\ref{fact1}, UITP-$(\gamma,\sigma,\delta,\mu)$ coincides with LP-ITP-$(\gamma,\sigma,\delta)$ when $\mu=\frac{1}{2}$. 
Since LP-ITP-$(\gamma,\sigma,\delta)$ can be derandomized by Remark~\ref{rmk1}, the overall algorithm can be implemented as a deterministic algorithm.
Furthermore, by Lemma~\ref{mitp}, the running time is $n^{\O(k)}$.

By Lemma~\ref{mitp}, the solution computed by UITP-$(0,-,\frac{1}{3},\frac{1}{3})$ has cost at most
\begin{equation}\label{case11}
\alpha\cdot \TSP + \RLB\cdot\frac{3}{2}\df{1}{\frac{1}{3}}+\EEE{c(\T^{>1/3}_{\mathrm{left}})}\leq \lrA{\alpha+1+\frac{3}{2}\df{1}{\frac{1}{3}}}\cdot\OPT,
\end{equation}
where the inequality follows from Lemmas~\ref{lowerbound} and \ref{inimr}.

By Lemmas~\ref{lpitp} and \ref{fact1}, the solution computed by UITP-$(\gamma,0,\frac{1}{3},\frac{1}{2})$ has cost at most 
{\footnotesize\begin{equation}\label{case12}
\begin{split}
&\alpha\cdot \TSP+\gamma\cdot\OPT+p_\gamma\cdot\RLB\cdot\lrB{\frac{3}{2}\df{1}{\frac{1}{3}} + 3\lrA{\df{1}{\frac{1}{2}}-\df{1}{\frac{1}{3}}}-\frac{1}{2}\lrA{\df{0}{\frac{1}{2}}-\df{0}{\frac{1}{3}}} + \lrA{\df{0}{1}-\df{0}{\frac{1}{2}}}}\\
&\leq\alpha\cdot \OPT+\gamma\cdot\OPT+p_\gamma\cdot\OPT\cdot\lrB{\frac{3}{2}\df{1}{\frac{1}{3}} + 3\lrA{\df{1}{\frac{1}{2}}-\df{1}{\frac{1}{3}}}-\lrA{\df{1}{\frac{1}{2}}-\df{1}{\frac{1}{3}}} + 2\lrA{\df{1}{1}-\df{1}{\frac{1}{2}}}}\\
&=\lrB{\alpha+\gamma+e^{-\gamma}\cdot\lrA{2-\frac{1}{2}\df{1}{\frac{1}{3}}}}\cdot\OPT,
\end{split}
\end{equation}}
where the first inequality follows from Lemma~\ref{lowerbound}, together with $\df{0}{\frac{1}{2}}-\df{0}{\frac{1}{3}}\geq 2\lra{\df{1}{\frac{1}{2}}-\df{1}{\frac{1}{3}}}$ and $\df{0}{1}-\df{0}{\frac{1}{2}}\leq2\lra{\df{1}{1}-\df{1}{\frac{1}{2}}}$ by Lemma~\ref{intineq}, and the equality follows from the fact that $\df{1}{1}=1$ by Lemma~\ref{lowerbound} and $p_\gamma=e^{-\gamma}$ by Lemma~\ref{prob}.

Since $\df{1}{\frac{1}{3}}\leq 1$ by Lemma~\ref{lowerbound}, by \eqref{case11} and \eqref{case12}, the approximation ratio of the algorithm is at most
\begin{equation}\label{case13}
\max_{\df{1}{\frac{1}{3}}\leq 1} \min\lrC{\alpha+1+\frac{3}{2}\df{1}{\frac{1}{3}}, \alpha+\gamma+e^{-\gamma}\cdot\lrA{2-\frac{1}{2}\df{1}{\frac{1}{3}}}}.
\end{equation}

Note that although $\df{1}{\frac{1}{3}}$ depends on the input instance, we treat it as a variable in \eqref{case13} and maximize over all feasible values satisfying $\df{1}{\frac{1}{3}} \leq 1$. 
In particular, the value of \eqref{case13} equals $\alpha + \pi$, where $\pi$ is the optimal objective value of the following factor-revealing LP:
\begin{align*}
\max\quad &\pi \\
\text{s.t.}\quad 
&\pi \leq 1+\frac{3}{2}\df{1}{\frac{1}{3}},\\
&\pi \leq \gamma+e^{-\gamma}\cdot\lrA{2-\frac{1}{2}\df{1}{\frac{1}{3}}},\\
&\df{1}{\frac{1}{3}} \leq 1.
\end{align*}

Solving the above LP yields the optimal value $\pi$, and thus an instance-independent upper bound of $\alpha+\pi$ on the approximation ratio.
We next derive an explicit bound on its optimal value.

% For convenience, let $y=\df{1}{\frac{1}{3}}$. By \eqref{case13}, for any fixed $\gamma\geq 0$, the first expression inside the minimum is strictly increasing in $y$, whereas the second is strictly decreasing in $y$. 
% It can be verified that for the optimal parameter $\gamma\geq 0$ minimizing \eqref{case13}, the worst-case value of the minimum of the two expressions is attained at their intersection, namely $1+\frac{3}{2}y = \gamma+e^{-\gamma}\cdot\lra{2-\frac{1}{2}y}$. Furthermore, this optimal parameter must satisfy the relation $\gamma=\ln\lra{2-\frac{1}{2}y}$. Combining this relation with the intersection condition yields the characteristic equation $\ln\lra{2-\frac{1}{2}y}=\frac{3}{2}y$. Let $y_0<0.39313$ denote the unique root of this equation, and set $\gamma=\ln\lra{2-\frac{1}{2}y_0}$. We next verify that this choice indeed yields an approximation ratio of at most $\alpha+1+\frac{3}{2}y_0$.

For convenience, let $y=\df{1}{\frac{1}{3}}$. By \eqref{case13}, the first expression inside the minimum is increasing in $y$, whereas the second is decreasing in $y$. 
It can be verified that for an optimal parameter $\gamma\geq 0$ minimizing \eqref{case13}, the worst-case value of the minimum of the two expressions is attained at their intersection, namely $1+\frac{3}{2}y = \gamma+e^{-\gamma}\cdot\lra{2-\frac{1}{2}y}$. Furthermore, this optimal parameter must satisfy the relation $\gamma=\ln\lra{2-\frac{1}{2}y}$.
Combining this relation with the intersection condition yields $\frac{3}{2}y=\ln\lra{2-\frac{1}{2}y}$. Let $y_0<0.39313$ denote the unique root of this equation. We set $\gamma=\ln\lra{2-\frac{1}{2}y_0}$. We now verify that this choice indeed yields an approximation ratio of at most $\alpha+1+\frac{3}{2}y_0$.

We analyze the expression in \eqref{case13} by considering the following two cases.

When $y\leq y_0$, we have $\alpha+1+\frac{3}{2}y\leq \alpha+1+\frac{3}{2}y_0$;
otherwise, we have
\begin{align*}
\alpha+\gamma+e^{-\gamma}\cdot\lrA{2-\frac{1}{2}y}=\alpha+\ln\lrA{2-\frac{1}{2}y_0}+\frac{2-\frac{1}{2}y}{2-\frac{1}{2}y_0}=\alpha+\frac{3}{2}y_0+\frac{2-\frac{1}{2}y}{2-\frac{1}{2}y_0}\leq \alpha+1+\frac{3}{2}y_0,
\end{align*}
where the equalities follow from the settings of $\gamma$ and $y_0$, and the inequality from the assumption $y > y_0$.

Therefore, by \eqref{case13}, the approximation ratio of the algorithm is at most $\alpha+1+\frac{3}{2}y_0<\alpha+1.5897$.
\end{proof}

\subsection{Proof of Theorem~\ref{thm2}}\label{pthm2}
\begin{proof}
Fix a small constant $\varepsilon$, say $\varepsilon=10^{-99}$.
Using the given $\alpha$-approximate TSP tour, the algorithm runs UITP-$(0,-,\frac{1}{3},\frac{1}{3})$, UITP-$(\gamma_1,\varepsilon,\frac{1}{3},\frac{1}{2})$, and UITP-$(\gamma_2,\varepsilon,\varepsilon,\frac{1}{2})$ to obtain three solutions, where the parameters $\gamma_1,\gamma_2\geq 0$ are defined later, and returns the one with the smallest cost.
Similar to the proof of Theorem~\ref{thm1}, the algorithm can be implemented as a deterministic algorithm that runs in $n^{\O(1)}$ time. 

By \eqref{case11}, the solution computed by UITP-$(0,-,\frac{1}{3},\frac{1}{3})$ has cost at most
\begin{equation}\label{case21}
\lrA{\alpha+1+\frac{3}{2}\df{1}{\frac{1}{3}}}\cdot\OPT,
\end{equation}
where the inequality follows from Lemmas~\ref{lowerbound} and \ref{inimr}.

By Lemma~\ref{fact1}, UITP-$(\gamma,\sigma,\delta,\mu)$ coincides with LP-ITP-$(\gamma,\sigma,\delta)$ when $\mu=\frac{1}{2}$.
Therefore, by Lemma~\ref{lpitp}, the solution computed by UITP-$(\gamma_1,\varepsilon,\frac{1}{3},\frac{1}{2})$ has cost at most 
{\footnotesize\begin{equation}\label{case22}
\begin{split}
&\alpha\cdot \TSP+\RLB\cdot\frac{3}{2}\df{1}{\varepsilon}+\gamma_1\cdot\OPT\\
&\quad+p_{\gamma_1}\cdot\RLB\cdot\lrB{\frac{3}{2}\lrA{\df{1}{\frac{1}{3}}-\df{1}{\varepsilon}} + 3\lrA{\df{1}{\frac{1}{2}}-\df{1}{\frac{1}{3}}}-\frac{1}{2}\lrA{\df{0}{\frac{1}{2}}-\df{0}{\frac{1}{3}}} + \lrA{\df{0}{1}-\df{0}{\frac{1}{2}}}}\\
&\leq\alpha\cdot \OPT+\OPT\cdot\frac{3}{2}\df{1}{\varepsilon}+\gamma_1\cdot\OPT+p_{\gamma_1}\cdot\OPT\cdot\lrB{2-2\df{1}{\varepsilon}-\frac{1}{2}\lrA{\df{1}{\frac{1}{3}}-\df{1}{\varepsilon}}}\\
&=\lrB{\alpha+\frac{3}{2}\df{1}{\varepsilon}+\gamma_1+e^{-\gamma_1}\cdot\lrA{2-\frac{3}{2}\df{1}{\varepsilon}-\frac{1}{2}\df{1}{\frac{1}{3}}}}\cdot\OPT,
\end{split}
\end{equation}}
where the first inequality follows from a similar analysis as in \eqref{case12}.

Moreover, the solution computed by UITP-$(\gamma_2,\varepsilon,\varepsilon,\frac{1}{2})$ has cost at most 
{\footnotesize\begin{equation}\label{case23}
\begin{split}
&\alpha\cdot \TSP+\RLB\cdot\frac{1}{1-\varepsilon}\df{1}{\varepsilon}+\gamma_2\cdot\OPT\\
&\quad+p_{\gamma_2}\cdot\RLB\cdot\lrB{\frac{2}{1-\varepsilon}\lrA{\df{1}{\frac{1}{2}}-\df{1}{\varepsilon}}-\frac{\varepsilon}{1-\varepsilon}\lrA{\df{0}{\frac{1}{2}}-\df{0}{\varepsilon}} + \lrA{\df{0}{1}-\df{0}{\frac{1}{2}}}}\\
&\leq\alpha\cdot \OPT+\OPT\cdot\frac{1}{1-\varepsilon}\df{1}{\varepsilon}+\gamma_2\cdot\OPT+p_{\gamma_2}\cdot\OPT\cdot\lrA{2-2\df{1}{\varepsilon}}\\
&\leq\lrB{\alpha+\df{1}{\varepsilon}+\gamma_2+e^{-\gamma_2}\cdot\lrA{2-2\df{1}{\varepsilon}}+2\varepsilon}\cdot\OPT,
\end{split}
\end{equation}}
where the first inequality follows from a similar analysis as in \eqref{case12}, and the second inequality follows from $\frac{1}{1-\varepsilon}\df{1}{\varepsilon}=\df{1}{\varepsilon}+\frac{\varepsilon}{1-\varepsilon}\df{1}{\varepsilon}\leq \df{1}{\varepsilon}+\frac{\varepsilon}{1-\varepsilon}\leq \df{1}{\varepsilon}+2\varepsilon$, using the fact that $\df{1}{1}=1$ by Lemma~\ref{lowerbound} and $\varepsilon$ is a sufficiently small constant.

By \eqref{case21}, \eqref{case22}, and \eqref{case23}, and since $\df{1}{\varepsilon}\leq\df{1}{\frac{1}{3}}\leq 1$ by Lemma~\ref{lowerbound}, the approximation ratio of the algorithm is at most $\alpha+\rho+2\varepsilon$, where $\rho$ is defined as
{\footnotesize\begin{equation}\label{case24}
\max_{\df{1}{\varepsilon}\leq\df{1}{\frac{1}{3}}\leq 1} \min\lrC{
1+\frac{3}{2}\df{1}{\frac{1}{3}},\frac{3}{2}\df{1}{\varepsilon}+\gamma_1+e^{-\gamma_1}\cdot\lrA{2-\frac{3}{2}\df{1}{\varepsilon}-\frac{1}{2}\df{1}{\frac{1}{3}}},
\df{1}{\varepsilon}+\gamma_2+e^{-\gamma_2}\cdot\lrA{2-2\df{1}{\varepsilon}}}.
\end{equation}}

Similarly, $\df{1}{\varepsilon}$ and $\df{1}{\frac{1}{3}}$ are treated as variables in \eqref{case24} and maximized over all feasible values satisfying $\df{1}{\varepsilon}\leq\df{1}{\frac{1}{3}}\leq 1$.
We next derive an explicit bound on its optimal value.

For convenience, let $x_1 = \df{1}{\varepsilon}$ and $x_2 = \df{1}{\frac{1}{3}}$.
Observe that the three expressions inside the minimum in \eqref{case24} have different monotonicities: the first is increasing in $x_2$, the second is increasing in $x_1$ and decreasing in $x_2$, and the third is decreasing in $x_1$.
It can be verified that for any optimal parameters $\gamma_1,\gamma_2\geq 0$ minimizing \eqref{case24}, the worst-case value of the minimum in \eqref{case24} is attained when the three expressions are equal, namely $1+\frac{3}{2}x_2 = \frac{3}{2}x_1+\gamma_1+e^{-\gamma_1}\cdot\lra{2-\frac{3}{2}x_1-\frac{1}{2}x_2} = x_1+\gamma_2+e^{-\gamma_2}\cdot\lra{2-2x_1}$.
Furthermore, these optimal parameters must satisfy the relations $\gamma_1 = \ln\lra{2-\frac{3}{2}x_1-\frac{1}{2}x_2}$ and $\gamma_2 = \ln\lra{2-2x_1}$.
Combining these relations with the above intersection conditions yields a system of equations. 
Let $y_1$ and $y_2$ denote the corresponding values of $x_1$ and $x_2$, where $y_1>0.17458$ is the unique root of the equation $\frac{1}{2}y+ 6 (1-y)\lra{1-e^{-\frac{1}{2} y}} =\ln\lra{2-2 y}$ and $y_2 = 4 (1-y_1)\lra{1-e^{-\frac{1}{2}y_1}} + y_1$.
We set $\gamma_1 = \ln\lra{2-\frac{3}{2}y_1-\frac{1}{2}y_2}$ and $\gamma_2=\ln\lra{2-2y_1}$.
We now verify that this choice indeed yields an approximation ratio of at most $\alpha+1+y_1+\ln\lra{2-2y_1}+2\varepsilon<\alpha+1.6759$.

We analyze the expression in \eqref{case24} by considering the following three cases.

When $x_1 \geq y_1$, we have
{\small\begin{align*}
x_1+\gamma_2+e^{-\gamma_2}\cdot\lrA{2-2x_1}=x_1+\ln\lrA{2-2y_1}+\frac{1-x_1}{1-y_1} \leq y_1+\ln\lrA{2-2y_1}+\frac{1-y_1}{1-y_1} = 1+y_1+\ln\lra{2-2y_1},
\end{align*}}
where the first equality follows from the definition of $\gamma_2$, and the inequality holds because the expression is decreasing in $x_1$ over the domain $x_1 \geq y_1$.

When $x_1 \leq y_1$ and $x_2 \leq y_2$, we have
\begin{align*}
1+\frac{3}{2}x_2 \leq 1+\frac{3}{2}y_2=1+y_1+\lrA{\frac{3}{2}y_2-y_1}=1+y_1+\ln\lra{2-2y_1},
\end{align*}
where the equalities follow from the definitions of $y_1$ and $y_2$, which imply $\frac{3}{2}y_2-y_1 = \ln\lra{2-2y_1}$.

Otherwise, we have $x_1 \leq y_1$ and $x_2 \geq y_2$. Then,
\begin{align*}
\frac{3}{2}x_1+\gamma_1+e^{-\gamma_1}\cdot\lrA{2-\frac{3}{2}x_1-\frac{1}{2}x_2} &= \frac{3}{2}x_1+\ln\lrA{2-\frac{3}{2}y_1-\frac{1}{2}y_2}+\frac{2-\frac{3}{2}x_1-\frac{1}{2}x_2}{2-\frac{3}{2}y_1-\frac{1}{2}y_2}\\
&\leq \frac{3}{2}y_1+\ln\lrA{2-\frac{3}{2}y_1-\frac{1}{2}y_2}+\frac{2-\frac{3}{2}y_1-\frac{1}{2}y_2}{2-\frac{3}{2}y_1-\frac{1}{2}y_2}\\
&= \frac{3}{2}y_1+\ln\lrA{2-\frac{3}{2}y_1-\frac{1}{2}y_2}+1\\
&= 1+y_1+\ln\lra{2-2y_1},
\end{align*}
where the inequality holds because the expression is increasing in $x_1$ and decreasing in $x_2$ over this domain, and the last equality follows from the fact that $2-\frac{3}{2}y_1-\frac{1}{2}y_2 = 2(1-y_1)e^{-\frac{1}{2}y_1}$, which implies that $\ln\lra{2-\frac{3}{2}y_1-\frac{1}{2}y_2} = \ln\lra{2-2y_1}-\frac{1}{2}y_1$.

Therefore, by \eqref{case24}, we conclude that $\rho \leq 1+y_1+\ln\lra{2-2y_1}$. The approximation ratio of the algorithm is at most $\alpha+1+y_1+\ln\lra{2-2y_1}+2\varepsilon<\alpha+1.6759$.
\end{proof}

\subsection{Proof of Lemma~\ref{ratio}}\label{pratio}
\begin{proof}
We analyze the two capacity settings separately.

\paragraph{\textbf{Fixed Vehicle Capacity.}} 
By Lemmas~\ref{lowerbound} and \ref{mitp}, together with \eqref{defxi} and \eqref{ratio1}, the expected approximation ratio of UITP-$(\gamma,\sigma,\delta,\mu)$ with $\sigma=0$ is at most $\zeta_1(I_\gamma,I_\delta,I_\mu)+\alpha$.

\paragraph{\textbf{General Vehicle Capacity.}} 
We further distinguish whether $I_\delta\in[N/2]$ or $I_\delta=\varepsilon\cdot N$.

\textbf{Case~1: $I_\delta\in[N/2]$.}
By Lemmas~\ref{lowerbound} and \ref{mitp}, together with \eqref{defxi} and \eqref{ratio21}, the expected approximation ratio of UITP-$(\gamma,\sigma,\delta,\mu)$ with $\sigma=\varepsilon$ is at most $\zeta_2(I_\gamma,I_\delta,I_\mu)+\alpha$.

\textbf{Case~2: $I_\delta = \varepsilon\cdot N$.}
Similarly, by Lemmas~\ref{lowerbound} and \ref{mitp}, together with \eqref{defxi}, the expected approximation ratio of UITP-$(\gamma,\sigma,\delta,\mu)$ with $\sigma=\varepsilon$ is at most
\begin{align*}
&\alpha + \frac{1}{1-\varepsilon}\xi_1(\varepsilon\cdot N) + \gamma + p_\gamma\cdot \frac{2}{1-\varepsilon}\lrA{\xi_1(I_\mu)-\xi_1(\varepsilon\cdot N)} + \phi_\gamma(I_\mu)\\
& = \zeta_2(I_\gamma,I_\delta,I_\mu) + \alpha + \frac{\varepsilon}{1-\varepsilon}\xi_1(\varepsilon\cdot N) + p_\gamma\cdot \frac{2\varepsilon}{1-\varepsilon}\lrA{\xi_1(I_\mu)-\xi_1(\varepsilon\cdot N)}\\
&\leq \zeta_2(I_\gamma,I_\delta,I_\mu) + \alpha + \frac{2\varepsilon}{1-\varepsilon}\xi_1(I_\mu)\\
&\leq \zeta_2(I_\gamma,I_\delta,I_\mu) + \alpha + 3\varepsilon,
\end{align*}
where the equality follows from \eqref{ratio22}, and the last inequality follows from $\xi_1(I_\mu)\leq 1$ by \eqref{intxi}, using the fact that $\varepsilon>0$ is chosen to be sufficiently small.
\end{proof}

\subsection{Proof of Lemma~\ref{zwp}}\label{pzwp}
\begin{proof}
For part (a), by \eqref{toureq}, we have $\frac{c(\Q_i)}{\OPT} = \chi(\mathbb{T}_i)$ for each $i\in[N_{\mathrm{pat}}]$. Hence,
\[
\sum_{i=1}^{N_{\mathrm{pat}}}\chi(\mathbb{T}_i)=\sum_{i=1}^{N_{\mathrm{pat}}}\frac{c(\Q_i)}{\OPT} = \frac{c(\T_{\mathrm{large}})}{\OPT} \leq 1.
\]
Parts (b) and (c) follow from the fact that each $\mathbb{M}_i$ is an asymmetric metric space together with \eqref{funcons}.
\end{proof}

\subsection{Proof of Lemma~\ref{match3}}\label{pmatch3}
\begin{proof}
When $\mu=\frac{1}{3}$, $\M^{>\mu}$ is a minimum-cost 2-tour set serving all type-2 and type-3 customers.
To prove the lemma, it suffices to show that there exists a 2-tour set serving all type-2 and type-3 customers whose cost is at most $\sum_{i=1}^{N_{\mathrm{pat}}}\Gamma_{\mu,i}\cdot\OPT$.
To this end, for each $i\in[N_{\mathrm{pat}}]$, we construct a 2-tour set $\Q'_i$ satisfying $\frac{c(\Q'_i)}{\OPT}=\Gamma_{\mu,i}$ that serves all type-2 and type-3 customers in $V(\Q_i)$.

By \eqref{patcons} and \eqref{tourcons}, it is feasible to define $\Q'_i$ as follows:
{\small\begin{equation}\label{Qtour3}
\begin{split}
&
\Q'_1=\emptyset,\;
\Q'_2=\{(v^{z}_{0},v^{z}_{2},v^{z}_{0})\}_{z\in[\size{\Q_2}]},\;
\Q'_3=\{(v^{z}_{0},v^{z}_{3},v^{z}_{0})\}_{z\in[\size{\Q_3}]},
\\
&
\Q'_4=\emptyset,\;
\Q'_5=\{(v^{z}_{0},v^{z}_{7},v^{z}_{0})\}_{z\in[\size{\Q_5}]},\;
\Q'_6=\{(v^{z}_{0},v^{z}_{9},v^{z}_{0})\}_{z\in[\size{\Q_6}]},\;
\Q'_7=\{(v^{z}_{0},v^{z}_{{10}},v^{z}_{0})\}_{z\in[\size{\Q_7}]},
\\
&
\Q'_8=\{(v^{z}_{0},v^{z}_{{12}},v^{z}_{{13}},v^{z}_{0})\}_{z\in[\size{\Q_8}]},\;
\Q'_9=\{(v^{z}_{0},v^{z}_{{14}},v^{z}_{{15}},v^{z}_{0})\}_{z\in[\size{\Q_9}]},\;
\Q'_{10}=\{(v^{z}_{0},v^{z}_{{16}},v^{z}_{0})\}_{z\in[\size{\Q_{10}}]},
\\
&
\Q'_{11}=\{(v^{z}_{0},v^{z}_{{18}},v^{z}_{{19}},v^{z}_{0})\}_{z\in[\size{\Q_{11}}]},
\\
&
\Q'_{12}=\emptyset,\;
\Q'_{13}=\{(v^{z}_{0},v^{z}_{{25}},v^{z}_{0})\}_{z\in[\size{\Q_{13}}]},\;
\Q'_{14}=\{(v^{z}_{0},v^{z}_{{27}},v^{z}_{0})\}_{z\in[\size{\Q_{14}}]},\;
\Q'_{15}=\{(v^{z}_{0},v^{z}_{{30}},v^{z}_{{31}},v^{z}_{0})\}_{z\in[\size{\Q_{15}}]},
\\
&
\Q'_{16}=\{(v^{z}_{0},v^{z}_{{32}},v^{z}_{0})\}_{z\in[\size{\Q_{16}}]},\;
\Q'_{17}=\{(v^{z}_{0},v^{z}_{{35}},v^{z}_{{37}},v^{z}_{0})\}_{z\in[\size{\Q_{17}}]},\;
\Q'_{18}=\{(v^{z}_{0},v^{z}_{{38}},v^{z}_{{39}},v^{z}_{0})\}_{z\in[\size{\Q_{18}}]}.
\end{split}
\end{equation}}

Note that the idea behind the above construction is that, if we regard the vertex $0$ in $\mathbb{G}$ as the depot and all other vertices in $\mathbb{V}_1\cup\mathbb{V}_2\cup\cdots\cup\mathbb{V}_{N_{\mathrm{pat}}}$ as customers, then any 2-tour set in $\mathbb{G}$ covering all vertices in $\mathbb{J}_2\cup \mathbb{J}_3$ naturally induces a 2-tour set covering all type-2 and type-3 customers in the original graph $G$.

By \eqref{funcons}, \eqref{contourcons}, and \eqref{Qtour3}, the bound 
$\frac{c(\Q'_i)}{\OPT}=\Gamma_{\mu,i}$ holds for each $i\in[N_{\mathrm{pat}}]$.

Since all type-2 and type-3 customers are served in $\T_{\mathrm{large}}=\bigcup_{i=1}^{N_{\mathrm{pat}}}\Q_i$, it follows that $\bigcup_{i=1}^{N_{\mathrm{pat}}}\Q'_i$ forms a 2-tour set serving these customers.
Hence, we obtain $\frac{c(\M^{>\mu})}{\OPT}\leq \sum_{i=1}^{N_{\mathrm{pat}}}\frac{c(\Q'_i)}{\OPT}=\sum_{i=1}^{N_{\mathrm{pat}}}\Gamma_{\mu,i}$.
\end{proof}

\subsection{Proof of Lemma~\ref{match4}}\label{pmatch4}
\begin{proof}
Similar to the proof of Lemma~\ref{match3}, it suffices to construct a 2-tour set serving all type-1, type-2, and type-3 customers with cost at most $\sum_{i=1}^{N_{\mathrm{pat}}}\Gamma_{\mu,i}\cdot\OPT$.
To this end, for each $i\in[N_{\mathrm{pat}}]$, we construct a 2-tour set $\Q'_i$ satisfying $\frac{c(\Q'_i)}{\OPT}=\Gamma_{\mu,i}$ that serves all type-1, type-2, and type-3 customers in $V(\Q_i)$. 

\textbf{Case~1: $i\in[11]$.}
By \eqref{tourcons}, $\Q_i$ is a 2-tour set. Hence, we define $\Q'_i=\Q_i$, which implies $\frac{c(\Q'_i)}{\OPT}= \chi(\mathbb{T}_i)=\Gamma_{\mu,i}$.

\textbf{Case~2: $i\in[12;18]$.}
By \eqref{tourcons}, $\Q_i$ is a 3-tour set in which every tour is a 3-tour.
Assume that 
\[
\Q_i = \{(v^{z}_{0},v^{z}_{{i_1}},v^{z}_{{i_2}},v^{z}_{{i_3}},v^{z}_{0})\}_{z\in[\size{\Q_i}]}.
\]
Define
\begin{align*}
\Q_{i,1}&= \{(v^{z}_{0},v^{z}_{{i_1}},v^{z}_{0}), (v^{z}_{0},v^{z}_{{i_2}},v^{z}_{{i_3}},v^{z}_{0})\}_{z\in[\size{\Q_i}]},\\
\Q_{i,2}&= \{(v^{z}_{0},v^{z}_{{i_2}},v^{z}_{0}), (v^{z}_{0},v^{z}_{{i_1}},v^{z}_{{i_3}},v^{z}_{0})\}_{z\in[\size{\Q_i}]},\\
\Q_{i,3}&= \{(v^{z}_{0},v^{z}_{{i_3}},v^{z}_{0}), (v^{z}_{0},v^{z}_{{i_1}},v^{z}_{{i_2}},v^{z}_{0})\}_{z\in[\size{\Q_i}]}.
\end{align*}
Each of $\Q_{i,1}$, $\Q_{i,2}$, and $\Q_{i,3}$ is a 2-tour set serving all customers in $V(\Q_i)$.
We define
\[
\Q'_i = \arg\min\lrC{c(\Q_{i,1}),c(\Q_{i,2}),c(\Q_{i,3})}.
\]
Let $\mathbb{V}_i=\{0, p_{i,1},p_{i,2},p_{i,3}\}$, where $p_{i,1}<p_{i,2}<p_{i,3}$. 
By \eqref{vsetcons}, \eqref{funcons}, and \eqref{contourcons}, we obtain
{\footnotesize\begin{align*}
\frac{c(\Q'_i)}{\OPT}
&= \min_{(p_1,p_2,p_3)\in \{(p_{i,1},p_{i,2},p_{i,3}),(p_{i,2},p_{i,1},p_{i,3}),(p_{i,3},p_{i,1},p_{i,2})\}}\lra{\chi(0,p_1)+\chi(p_1,0)+\chi(0,p_2)+\chi(p_2,p_3)+\chi(p_3,0)}.
\end{align*}}
Hence, we have $\frac{c(\Q'_i)}{\OPT}=\Gamma_{\mu,i}$.

By taking the union over all $i$, we obtain a 2-tour set $\bigcup_{i=1}^{N_{\mathrm{pat}}}\Q'_i$ that serves all type-1, type-2, and type-3 customers.
Hence, we obtain $\frac{c(\M^{>\mu})}{\OPT}\leq \sum_{i=1}^{N_{\mathrm{pat}}}\frac{c(\Q'_i)}{\OPT}=\sum_{i=1}^{N_{\mathrm{pat}}}\Gamma_{\mu,i}$.
\end{proof}

\subsection{Proof of Lemma~\ref{bridge}}\label{pbridge}
\begin{proof}
Recall that $V_i$ denotes the set of all type-$i$ customers in $V\setminus V_0$ for each $i\in[3]$.

On one hand, by \eqref{patcons}, \eqref{tourcons}, \eqref{vsetcons}, \eqref{funcons}, and \eqref{jsetcons}, we have
\begin{itemize}
  \item $\frac{\sum_{v\in V_1}\lrA{c(r, v)+c(v,r)}}{\OPT} = \sum_{j\in\mathbb{J}_1}\lrA{\chi(0,j)+\chi(j,0)}$;
  \item $\frac{\sum_{v\in V_2}\lrA{c(r, v)+c(v,r)}}{\OPT} = \sum_{j\in\mathbb{J}_2}\lrA{\chi(0,j)+\chi(j,0)}$;
  \item $\frac{\sum_{v\in V_3}\lrA{c(r, v)+c(v,r)}}{\OPT} = \sum_{j\in\mathbb{J}_3}\lrA{\chi(0,j)+\chi(j,0)}$.
\end{itemize}

On the other hand, by \eqref{int} and \eqref{defxi}, together with the definition of type-$i$ customers for each $i\in[3]$, we have
\begin{itemize}
  \item $\xi_0(N/3) - \xi_0(N/4) = \frac{\sum_{v\in V_1}\lrA{c(r, v)+c(v,r)}}{\RLB}$;
  \item $\xi_0(N/2) - \xi_0(N/3) = \frac{\sum_{v\in V_2}\lrA{c(r, v)+c(v,r)}}{\RLB}$;
  \item $\xi_0(N) - \xi_0(N/2) = \frac{\sum_{v\in V_3}\lrA{c(r, v)+c(v,r)}}{\RLB}$.
\end{itemize}

Since $\RLB\leq\OPT$ by Lemma~\ref{lowerbound}, the bounds in the lemma follow immediately.
\end{proof}

\subsection{Proof of Lemma~\ref{expt2}}\label{pexpt2}
\begin{proof}
Suppose $Q^{>\mu}=(r,v_1,v_2,r)$. 
Let $X_i \in \{0,1\}$ be the \emph{survival indicator} for customer $v_i$, i.e., $X_i=1$ if $v_i\in V^{>\mu}_{\mathrm{left}}$, and $X_i=0$ otherwise.
Let $p_i=\PP{X_i=1}$ and $\Psi_i = c(r,v_i)+c(v_i,r)$ for each $i\in[2]$, $p_{1,2} = \PP{X_1=X_2=1}$, and $\Delta_{1,2} = c(v_1, r)+c(r, v_2) - c(v_1,v_2)$. 
Note that $c(Q^{>\mu})=\Psi_1+\Psi_2-\Delta_{1,2}$ and $\Delta_{\max}(Q^{>\mu}) = \Delta_{1,2}$.
Since $\PP{X_1=X_2=0} = 1 - \PP{X_1=1} - \PP{X_2=1} + \PP{X_1=X_2=1} = 1- p_1 - p_2 + p_{1,2}\geq 0$, we have 
\begin{equation}\label{probound}
p_{1,2} \geq \max\lrC{0, p_1+p_2 - 1}.
\end{equation}

By the definition of $Q^{>\mu}_{\mathrm{left}}$, we have the following three cases.
\begin{itemize}
  \item When $X_1+X_2=0$, we have $Q^{>\mu}_{\mathrm{left}}=\emptyset$, and thus $c(Q^{>\mu}_{\mathrm{left}})=0$.
  \item When $X_1+X_2=1$, if $X_i=1$ for some $i\in[2]$, we have $Q^{>\mu}_{\mathrm{left}}=(r,v_i,r)$, and thus $c(Q^{>\mu}_{\mathrm{left}})=\Psi_i$.
  \item When $X_1+X_2=2$, we have $Q^{>\mu}_{\mathrm{left}}=(r,v_1,v_2,r)$, and thus $c(Q^{>\mu}_{\mathrm{left}})=\Psi_1+\Psi_2-\Delta_{1,2}$. 
\end{itemize}

Hence, we obtain
\begin{equation}
\begin{split}
\EE{c(Q^{>\mu}_{\mathrm{left}})} &= \lrA{p_1-  p_{1,2}} \Psi_1 + \lrA{p_2-  p_{1,2}} \Psi_2 + p_{1,2}  \lrA{\Psi_1+\Psi_2-\Delta_{1,2}}\\
&=p_1\Psi_1 + p_2\Psi_2 - p_{1,2} \Delta_{1,2}\\
&\leq p_1\Psi_1 + p_2\Psi_2 - \max\lrC{0, p_1 + p_2 -1 }\Delta_{1,2},
\end{split}
\end{equation}
where the inequality follows from \eqref{probound}.

Define $f(p_1,p_2)\coloneq p_1\cdot\Psi_1 + p_2\cdot\Psi_2 - \max\lrC{0, p_1 + p_2 -1 }\cdot\Delta_{1,2}$. 
When $p_1+p_2\leq 1$, we have $\frac{\partial f}{\partial p_1} = \Psi_1\geq 0$.
When $p_1+p_2\geq 1$, we have 
{\footnotesize\[
\frac{\partial f(p_1,p_2)}{\partial p_1} = \Psi_1 - \Delta_{1,2} = \lrA{c(r,v_1)+c(v_1,r)} - \lrA{c(v_1, r)+c(r, v_2) - c(v_1,v_2)} = c(r,v_1) + c(v_1,v_2) - c(r, v_2)\geq 0,
\]}
where the last inequality follows from the triangle inequality.

Hence, we have $\frac{\partial f(p_1,p_2)}{\partial p_1}\geq 0$, and similarly $\frac{\partial f(p_1,p_2)}{\partial p_2} \geq \Psi_2 - \Delta_{1,2} = c(v_2,r) + c(v_1,v_2) - c(v_1, r)\geq 0$.

Since $p_i\leq p_\gamma$ by Lemma~\ref{prob}, we have
\begin{align*}
\EE{c(Q^{>\mu}_{\mathrm{left}})} &\leq p_\gamma \lrA{\Psi_1 +\Psi_2} - \max\lrC{0, 2p_\gamma -1}\cdot \Delta_{1,2} \\
&= p_\gamma\sum_{v\in V(Q^{>\mu})}\lrA{c(r,v)+c(v,r)} - \max\lrC{0,2p_\gamma-1}\cdot \Delta_{\max}(Q^{>\mu}),
\end{align*}
which completes the proof.
\end{proof}

\subsection{Proof of Lemma~\ref{expt3}}\label{pexpt3}
\begin{proof}
Suppose $Q^{>\mu}=\lra{r,v_1,v_2,v_3,r}$. 
Similar to the proof of Lemma~\ref{expt2}, let $X_i \in \{0,1\}$ be the survival indicator for customer $v_i$.
Define $p_i=\PP{X_i=1}$ and $\Psi_i = c(r,v_i)+c(v_i,r)$ for each $i\in[3]$. 
For distinct $i,j,\ell \in [3]$, let $b_{i,j} = \PP{X_i=1,X_j=1,X_\ell=0}$ and $t = \PP{X_1=1,X_2=1,X_3=1}$. 
Moreover, for $1\leq i<j\leq 3$, let $\Delta_{i,j} = c(v_i, r)+c(r, v_j) - c(v_i,v_j)$.

By the definition of $Q^{>\mu}_{\mathrm{left}}$, we have the following four cases.
\begin{itemize}
  \item When $X_1+X_2+X_3=0$, we have $Q^{>\mu}_{\mathrm{left}}=\emptyset$, and thus $c(Q^{>\mu}_{\mathrm{left}})=0$.
  \item When $X_1+X_2+X_3=1$, if $X_i=1$ for some $i\in[3]$, we have $Q^{>\mu}_{\mathrm{left}}=(r,v_i,r)$, and thus $c(Q^{>\mu}_{\mathrm{left}})=\Psi_i$.
  \item When $X_1+X_2+X_3=2$, if $X_i=X_j=1$ for some $1\leq i<j\leq 3$, we have $Q^{>\mu}_{\mathrm{left}}=(r,v_i,v_j,r)$, and thus $c(Q^{>\mu}_{\mathrm{left}})=\Psi_i+\Psi_j-\Delta_{i,j}$.
  \item When $X_1+X_2+X_3=3$, we have $c(Q^{>\mu}_{\mathrm{left}})=\sum_{i=1}^3 \Psi_i - \Delta_{\max}(Q^{>\mu})$ since $\sum_{i=1}^3 \Psi_i - \Delta_{\max}(Q^{>\mu})$ is exactly the minimum cost among the three $2$-tour sets defined in~\eqref{shortdef}.
\end{itemize}

Hence, we obtain
\begin{equation}\label{eq:exp_cost_3}
\begin{split}
\EE{c(Q^{>\mu}_{\mathrm{left}})} & = \lrA{p_1-b_{1,2}-b_{1,3}-t}\Psi_1+\lrA{p_2-b_{1,2}-b_{2,3}-t}\Psi_2+\lrA{p_3-b_{1,3}-b_{2,3}-t}\Psi_3\\
&\quad\quad + b_{1,2}\lrA{\Psi_1+\Psi_2-\Delta_{1,2}}+ b_{1,3}\lrA{\Psi_1+\Psi_3-\Delta_{1,3}}+ b_{2,3}\lrA{\Psi_2+\Psi_3-\Delta_{2,3}}\\
&\quad\quad + t\lrA{\Psi_1+\Psi_2+\Psi_3-\Delta_{\max}(Q^{>\mu})}\\
&= \sum_{i=1}^3 p_i \Psi_i - \lrA{\sum_{1\leq i<j\leq 3} b_{i,j}\Delta_{i,j} + t\Delta_{\max}(Q^{>\mu})}.
\end{split}
\end{equation}

\begin{claim}\label{pip}
To maximize the expected cost $\EE{c(Q^{>\mu}_{\mathrm{left}})}$, without loss of optimality, we may assume that the marginal probabilities equal their upper bounds, i.e., $p_1=p_2=p_3=p_\gamma$.
\end{claim}
\begin{claimproof}
Suppose $\P^*$ is an optimal joint distribution maximizing the expected cost $\EE{c(Q^{>\mu}_{\mathrm{left}})}$, and let $p_1,p_2,p_3\leq p_\gamma$ be its marginal probabilities. 
Let $\bm{X}=(X_1,X_2,X_3)\in\{0,1\}^3$ denote the random vector of survival indicators drawn from $\P^*$.

By the triangle inequality, serving an additional surviving customer cannot decrease the cost of the constructed tour set $Q^{>\mu}_{\mathrm{left}}$. 
Hence, the function $f(\bm{X})=c(Q^{>\mu}_{\mathrm{left}})$ is monotone non-decreasing with respect to $\bm{X}$; that is, for any $\bm{X}, \bm{Y} \in \{0,1\}^3$, if $X_i \leq Y_i$ for all $i \in [3]$, then $f(\bm{X}) \leq f(\bm{Y})$.

We construct a new joint distribution $\P$ for a random vector $\bm{Y} = (Y_1, Y_2, Y_3)$ whose marginal probabilities are exactly $p_\gamma$.
Let $Z_1, Z_2, Z_3$ be mutually independent Bernoulli random variables that are also independent of $\bm{X}$, where, for each $i\in[3]$, $\PP{Z_i=1}=\frac{p_\gamma-p_i}{1-p_i}$ if $p_i<1$ and $\PP{Z_i=1}=0$ otherwise.

Define $Y_i = \max(X_i, Z_i)$ for each $i \in [3]$. By definition, we have $Y_i \geq X_i$ for all $i\in[3]$, which implies $f(\bm{Y}) \geq f(\bm{X})$ for every realization. Taking expectations, we obtain $\EE{f(\bm{Y})} \geq \EE{f(\bm{X})}$.

Moreover, the marginal probabilities of $\bm{Y}$ are exactly $p_\gamma$:
\begin{align*}
\PPP{Y_i=1}
&=1-\PPP{Y_i=0}\\
&=1-\PPP{X_i=0}\cdot \PPP{Z_i=0}
\quad \text{(by independence)}\\
&=1-(1-p_i)\lrA{1-\frac{p_\gamma-p_i}{1-p_i}}\\
&=1-(1-p_\gamma)\\
&=p_\gamma.
\end{align*}

Therefore, for any joint distribution with marginal probabilities at most $p_\gamma$, there exists another joint distribution with marginal probabilities exactly $p_\gamma$ whose expected cost is at least as large. 
Hence, without loss of optimality, we may assume that $p_1 = p_2 = p_3 = p_\gamma$.
\end{claimproof}

Let $S \coloneq \sum_{1\leq i<j\leq 3} b_{i,j}\Delta_{i,j} + t\Delta_{\max}(Q^{>\mu})$ denote the expected \emph{savings}. 
To maximize the expected cost, $S$ must be minimized subject to valid probability constraints. 
By Claim~\ref{pip}, we may assume without loss of generality that $p_i=p_\gamma$ for all $i\in[3]$. 
Thus, the joint probabilities must satisfy the marginal constraints
\begin{equation}\label{eq:prob_marginal}
b_{i,j} + b_{i,\ell} + t \leq p_\gamma \quad\forall\ i\in[3],
\end{equation}
and the non-negativity of the empty-set probability 
\[
\PPP{X_1=X_2=X_3=0}=1-\lrA{p_1+p_2+p_3}+\lrB{\lrA{b_{1,2}+t}+\lrA{b_{1,3}+t}+\lrA{b_{2,3}+t}}-t\geq 0,
\] 
which implies:
\begin{equation}\label{eq:prob_union}
b_{1,2} + b_{1,3} + b_{2,3} + 2t \geq 3p_\gamma - 1,
\end{equation}
together with $b_{i,j}\geq 0$ and $t\geq 0$. 

To analyze the minimum savings $S$, we consider the following LP:
\begin{equation*}
\begin{aligned}
\min \quad & S=\sum_{1\leq i<j\leq 3} b_{i,j}\Delta_{i,j} + t\Delta_{\max}(Q^{>\mu}) \\
\text{s.t.} \quad & \eqref{eq:prob_marginal},\ \eqref{eq:prob_union}, \text{ and } b_{i,j}, t \geq 0.
\end{aligned}
\end{equation*}

We analyze the minimum savings $S$ over this feasible region in three cases depending on $p_\gamma$.

\textbf{Case~1: $0 \leq p_\gamma \leq \frac{1}{3}$.} 
The right-hand side of \eqref{eq:prob_union} is non-positive. Hence, the minimum savings is trivially achieved by setting $b_{i,j}=0$ for all $1\leq i<j\leq 3$ and $t=0$, yielding $S \geq 0$.

For $p_\gamma \geq \frac{1}{3}$, we analyze the minimum savings based on its dual program. 
Let $y_1, y_2, y_3 \geq 0$ be the dual variables associated with the marginal constraints \eqref{eq:prob_marginal}, and let $z \geq 0$ be the dual variable associated with the union constraint \eqref{eq:prob_union}. 
The dual LP is given by
\begin{equation}\label{eq:dual_lp}
\begin{aligned}
\max \quad & W = z\lrA{3p_\gamma - 1} - p_\gamma\lrA{y_1 + y_2 + y_3} \\
\text{s.t.} \quad & z - y_i - y_j \leq \Delta_{i,j}, \quad \forall\ 1 \leq i < j \leq 3, \\
& 2z - \lrA{y_1 + y_2 + y_3} \leq \Delta_{\max}(Q^{>\mu}), \\
& z, y_1, y_2, y_3 \geq 0.
\end{aligned}
\end{equation}

Any feasible dual solution yields a lower bound $W \leq S$. If we find a feasible dual solution whose objective value matches that of a feasible primal solution, then both solutions are optimal and the minimum savings is certified.
Without loss of generality, we may define $\lrc{(i_1,j_1),(i_2,j_2),(i_3,j_3)}=\lrc{(1,2),(1,3),(2,3)}$ such that $\Delta_{\min}(Q^{>\mu}) = \Delta_{i_1,j_1}$ and $\Delta_{\max}(Q^{>\mu}) = \Delta_{i_3,j_3}$.

\textbf{Case~2: $\frac{1}{3} \leq p_\gamma \leq \frac{1}{2}$.}
In this case, $3p_\gamma - 1 \leq p_\gamma$. We further consider two subcases.

\textbf{Subcase~2a: $\Delta_{\min}(Q^{>\mu}) \leq \frac{1}{2}\Delta_{\max}(Q^{>\mu})$.} 
We set $z = \Delta_{\min}(Q^{>\mu})$ and $y_i = 0$ for all $i\in[3]$. This choice is dual feasible since $z \leq \Delta_{i,j}$ for all $1\leq i<j\leq 3$ and $2z \leq \Delta_{\max}(Q^{>\mu})$. The dual objective becomes $W = (3p_\gamma - 1)\Delta_{\min}(Q^{>\mu})$, which matches the primal value achieved by setting $b_{i_1,j_1} = 3p_\gamma - 1$ and $b_{i_2,j_2}=b_{i_3,j_3}= t = 0$.

\textbf{Subcase~2b: $\Delta_{\min}(Q^{>\mu}) > \frac{1}{2}\Delta_{\max}(Q^{>\mu})$.} 
We set $z = \frac{1}{2}\Delta_{\max}(Q^{>\mu})$ and $y_i = 0$ for all $i\in[3]$. This choice is dual feasible since $z < \Delta_{\min}(Q^{>\mu}) \leq \Delta_{i,j}$ for all $1\leq i<j\leq 3$. 
The dual objective becomes $W = \frac{3p_\gamma - 1}{2}\Delta_{\max}(Q^{>\mu})$, which matches the primal value achieved by setting $b_{i,j} = 0$ for all $1\leq i<j\leq 3$ and $t = \frac{3p_\gamma - 1}{2}$.

Combining these two subcases, we obtain 
\[
S \geq (3p_\gamma - 1)\cdot\min\lrC{\Delta_{\min}(Q^{>\mu}), \frac{1}{2}\Delta_{\max}(Q^{>\mu})}.
\]

\textbf{Case~3: $\frac{1}{2} \leq p_\gamma \leq 1$.}
In this case, $3p_\gamma - 1 > p_\gamma$, so positive values of $y_i$ may be required to satisfy the dual constraints when $z$ is large.
Similarly, we consider two subcases.

\textbf{Subcase~3a: $\Delta_{\min}(Q^{>\mu}) \leq \frac{1}{2}\Delta_{\max}(Q^{>\mu})$.} 
We set $z = \Delta_{\max}(Q^{>\mu}) - \Delta_{\min}(Q^{>\mu})$, $y_1 = \Delta_{\max}(Q^{>\mu}) - 2\Delta_{\min}(Q^{>\mu}) \geq 0$, and $y_2=y_3=0$. 
This choice is dual feasible since $z - y_1 = \Delta_{\min}(Q^{>\mu}) \leq \Delta_{i,j}$ for all $1\leq i<j\leq 3$ and $2z - y_1 = \Delta_{\max}(Q^{>\mu})$.
The dual objective becomes $W = (\Delta_{\max}(Q^{>\mu}) - \Delta_{\min}(Q^{>\mu}))(3p_\gamma - 1) - p_\gamma(\Delta_{\max}(Q^{>\mu}) - 2\Delta_{\min}(Q^{>\mu})) = (1-p_\gamma)\Delta_{\min}(Q^{>\mu}) + (2p_\gamma - 1)\Delta_{\max}(Q^{>\mu})$, which matches the primal value achieved by setting $b_{i_1,j_1} = 1-p_\gamma$, $b_{i_2,j_2}=b_{i_3,j_3}=0$, and $t = 2p_\gamma - 1$.

\textbf{Subcase~3b: $\Delta_{\min}(Q^{>\mu}) > \frac{1}{2}\Delta_{\max}(Q^{>\mu})$.} 
We set $z = \frac{1}{2}\Delta_{\max}(Q^{>\mu})$ and $y_i = 0$ for all $i\in[3]$. This choice remains dual feasible. The dual objective becomes $W = \frac{3p_\gamma - 1}{2}\Delta_{\max}(Q^{>\mu})$, which matches the primal value achieved by setting $b_{i,j} = 0$ for all $1\leq i<j\leq 3$ and $t = \frac{3p_\gamma - 1}{2}$.

Combining these two subcases, we obtain 
\[
S \geq \min\lrC{(1-p_\gamma)\Delta_{\min}(Q^{>\mu}) + (2p_\gamma - 1)\Delta_{\max}(Q^{>\mu}), \frac{3p_\gamma - 1}{2}\Delta_{\max}(Q^{>\mu})}.
\]

Substituting the tight lower bounds of $S$ established in Cases 1, 2, and 3 back into \eqref{eq:exp_cost_3}, we obtain the final piecewise upper bound for $\EE{c(Q^{>\mu}_{\mathrm{left}})}$, which completes the proof.
\end{proof}

\subsection{Proof of Lemma~\ref{expt3+}}\label{pexpt3+}
\begin{proof}
We consider the following cases based on the value of $p_\gamma$.

\textbf{Case~1: $0\leq p_\gamma\leq\frac{1}{2}$.}
Since $\Delta_{\min}(Q^{>\mu})\geq 0$ by \eqref{eqDelta}, the result holds trivially by Lemma~\ref{expt3}.

For the subsequent cases where $p_\gamma \geq\frac{1}{2}$, let
\[
L=(1-p_\gamma)\Delta_{\min}(Q^{>\mu})+(2p_\gamma-1)\Delta_{\max}(Q^{>\mu})\quad\text{and}\quad R=\frac{3p_\gamma-1}{2}\Delta_{\max}(Q^{>\mu}).
\]

Our goal is then to establish lower bounds on $\min\lrc{L,R}$.
More precisely, we will prove that $\min\lrc{L,R}\geq \max\lrc{\lra{2p_\gamma-1}\Delta_{\max}(Q^{>\mu}), \frac{2p_\gamma-1}{2}\Delta_{\mathrm{sum}}(Q^{>\mu})}$ when $p_\gamma\in[\frac{1}{2},\frac{2}{3}]$, and that $\min\lrc{L,R}\geq \max\lrc{\lra{2p_\gamma-1}\Delta_{\max}(Q^{>\mu}), \lra{3p_\gamma-2}\Delta_{\max}(Q^{>\mu})+\frac{1-p_\gamma}{2}\Delta_{\mathrm{sum}}(Q^{>\mu}),\frac{3p_\gamma-1}{6}\Delta_{\mathrm{sum}}(Q^{>\mu})}$ when $p_\gamma\in[\frac{2}{3},1]$.

To this end, we introduce the following auxiliary functions parameterized by $\theta$:
\begin{equation}\label{feq}
\begin{split}
f_1(\theta) &= \theta\cdot\Delta_{\mathrm{sum}}(Q^{>\mu})+\lrA{2p_\gamma-1-2\theta}\Delta_{\max}(Q^{>\mu}),\\
f_2(\theta) &= \theta\cdot\Delta_{\mathrm{sum}}(Q^{>\mu})+\lrA{p_\gamma-3\theta}\Delta_{\max}(Q^{>\mu}),\\
f_3(\theta) &= \theta\cdot\Delta_{\mathrm{sum}}(Q^{>\mu})+\lrA{\frac{3p_\gamma-1}{2}-3\theta}\Delta_{\max}(Q^{>\mu}).
\end{split}
\end{equation}

Note that $2\Delta_{\max}(Q^{>\mu})+\Delta_{\min}(Q^{>\mu})\geq \Delta_{\mathrm{sum}}(Q^{>\mu})$ and $\Delta_{\max}(Q^{>\mu})\geq \frac{1}{3}\Delta_{\mathrm{sum}}(Q^{>\mu})$. 
Combining these inequalities with \eqref{feq}, we obtain the following lower bounds on $L$ and $R$ over specific domains of $\theta$ and $p_\gamma$:
\begin{equation}\label{lreq}
\begin{split}
L&\geq f_1(\theta),\quad\forall\ \theta\in [0, \frac{2p_\gamma-1}{2}],\quad\forall\ p_\gamma\in[\frac{1}{2},\frac{3}{4}],\\
L&\geq f_1(\theta),\quad\forall\ \theta\in [0,1-p_\gamma],\quad\forall\ p_\gamma\in[\frac{3}{4},1],\\
L&\geq f_2(\theta),\quad\forall\ \theta\in [1-p_\gamma,\frac{p_\gamma}{3}],\quad\forall\ p_\gamma\in[\frac{3}{4},1],\\
R&\geq f_3(\theta),\quad\forall\ \theta\in [0, \frac{3p_\gamma-1}{6}],\quad\forall\ p_\gamma\in[\frac{1}{3},1].
\end{split}
\end{equation}

\textbf{Case~2: $\frac{1}{2}\leq p_\gamma\leq\frac{2}{3}$.}
In this case, we have $\frac{2p_\gamma-1}{2}\leq\frac{3p_\gamma-1}{6}\leq\frac{1-p_\gamma}{2}$. 
Therefore, by \eqref{lreq}, we obtain $L\geq f_1(\theta)$ and $R\geq f_3(\theta)$ for all $\theta\in[0,\frac{2p_\gamma-1}{2}]$. Moreover, by \eqref{feq}, we obtain $f_1(\theta)\leq f_3(\theta)$ for all $\theta\in[0,\frac{2p_\gamma-1}{2}]$.
Hence,
\begin{align*}
\min\lrC{L,R}&\geq\max_{\theta\in[0,\frac{2p_\gamma-1}{2}]}\min\lrC{f_1(\theta), f_3(\theta)}\\
&\geq\max_{\theta\in[0,\frac{2p_\gamma-1}{2}]}f_1(\theta)\\
&\geq\max\lrC{f_1(0),f_1(\frac{2p_\gamma-1}{2})}\\
&=\max\lrC{\lrA{2p_\gamma-1}\Delta_{\max}(Q^{>\mu}), \frac{2p_\gamma-1}{2}\Delta_{\mathrm{sum}}(Q^{>\mu})}.
\end{align*}

\textbf{Case~3: $\frac{2}{3}\leq p_\gamma\leq \frac{7}{9}$.} In this case, we have $\frac{1-p_\gamma}{2}\leq\frac{3p_\gamma-1}{6}\leq\min\lrc{\frac{2p_\gamma-1}{2},1-p_\gamma}$.
Therefore, by \eqref{lreq}, both $L\geq f_1(\theta)$ and $R\geq f_3(\theta)$ hold for all $\theta\in[0,\frac{3p_\gamma-1}{6}]$. 
Moreover, by \eqref{feq}, we obtain $f_1(\theta)\leq f_3(\theta)$ for all $\theta\in[0,\frac{1-p_\gamma}{2}]$, whereas $f_1(\theta)\geq f_3(\theta)$ for all $\theta\in[\frac{1-p_\gamma}{2},\frac{3p_\gamma-1}{6}]$.
Hence,
{\footnotesize\begin{align*}
\min\lrC{L,R}&\geq\max_{\theta\in[0,\frac{3p_\gamma-1}{6}]}\min\lrC{f_1(\theta), f_3(\theta)}\\
&\geq\max\lrC{\max_{\theta\in[0,\frac{1-p_\gamma}{2}]}f_1(\theta), \max_{\theta\in[\frac{1-p_\gamma}{2},\frac{3p_\gamma-1}{6}]}f_3(\theta)}\\
&\geq\max\lrC{f_1(0),f_1(\frac{1-p_\gamma}{2}), f_3(\frac{3p_\gamma-1}{6})}\\
&=\max\lrC{\lrA{2p_\gamma-1}\Delta_{\max}(Q^{>\mu}),\lrA{3p_\gamma-2}\Delta_{\max}(Q^{>\mu})+\frac{1-p_\gamma}{2}\Delta_{\mathrm{sum}}(Q^{>\mu}),\frac{3p_\gamma-1}{6}\Delta_{\mathrm{sum}}(Q^{>\mu})}.
\end{align*}}

\textbf{Case~4: $\frac{7}{9}\leq p_\gamma\leq 1$.} In this case, we have $\frac{1-p_\gamma}{2}\leq 1-p_\gamma\leq \frac{3p_\gamma-1}{6}\leq\frac{p_\gamma}{3}$.
Therefore, by \eqref{lreq}, we know that $L\geq f_1(\theta)$ and $R\geq f_3(\theta)$ for all $\theta\in[0,1-p_\gamma]$, while $L\geq f_2(\theta)$ and $R\geq f_3(\theta)$ for all $\theta\in[1-p_\gamma,\frac{3p_\gamma-1}{6}]$. 
Moreover, by \eqref{feq}, we obtain $f_1(\theta)\leq f_3(\theta)$ for all $\theta\in[0,\frac{1-p_\gamma}{2}]$, $f_1(\theta)\geq f_3(\theta)$ for all $\theta\in[\frac{1-p_\gamma}{2},1-p_\gamma]$, and $f_2(\theta)\geq f_3(\theta)$ for all $\theta\in[1-p_\gamma,\frac{3p_\gamma-1}{6}]$.
Hence,
{\footnotesize\begin{align*}
\min\lrC{L,R}&\geq\max\lrC{\max_{\theta\in[0,1-p_\gamma]}\min\lrC{f_1(\theta), f_3(\theta)}, \max_{\theta\in[1-p_\gamma,\frac{3p_\gamma-1}{6}]}\min\lrC{f_2(\theta), f_3(\theta)}}\\
&\geq\max\lrC{\max_{\theta\in[0,\frac{1-p_\gamma}{2}]}f_1(\theta), \max_{\theta\in[\frac{1-p_\gamma}{2},1-p_\gamma]}f_3(\theta), \max_{\theta\in[1-p_\gamma,\frac{3p_\gamma-1}{6}]}f_3(\theta)}\\
&=\max\lrC{\max_{\theta\in[0,\frac{1-p_\gamma}{2}]}f_1(\theta), \max_{\theta\in[\frac{1-p_\gamma}{2},\frac{3p_\gamma-1}{6}]}f_3(\theta)}\\
&\geq\max\lrC{f_1(0),f_1(\frac{1-p_\gamma}{2}), f_3(\frac{3p_\gamma-1}{6})}\\
&=\max\lrC{\lrA{2p_\gamma-1}\Delta_{\max}(Q^{>\mu}),\lrA{3p_\gamma-2}\Delta_{\max}(Q^{>\mu})+\frac{1-p_\gamma}{2}\Delta_{\mathrm{sum}}(Q^{>\mu}),\frac{3p_\gamma-1}{6}\Delta_{\mathrm{sum}}(Q^{>\mu})},
\end{align*}}
which completes the proof.
\end{proof}

\subsection{Proof of Lemma~\ref{tmu}}\label{ptmu}
\begin{proof}
Recall that, by \eqref{newVi} and \eqref{newQi}, we have
\begin{equation}\label{meq1}
\mathbb{T}^{>\mu}_{i}=\lrA{0,\tilde{p}^{>\mu}_{i,1},\ldots,\tilde{p}^{>\mu}_{i,\ell_i},0}\quad\text{and}\quad \Q^{>\mu}_i = \lrC{\lrA{v^{z}_{0},v^{z}_{{\tilde{p}^{>\mu}_{i,1}}},\ldots,v^{z}_{{\tilde{p}^{>\mu}_{i,\ell_i}}},v^{z}_{0}}}_{z\in[\size{\Q^{>\mu}_i}]}.
\end{equation}

By \eqref{meq1}, for any $i\in[N_{\mathrm{pat}}]$, we obtain
\begin{equation}\label{meq2}
\begin{split}
\sum_{Q^{>\mu}\in \Q^{>\mu}_i}\frac{\sum_{v\in V(Q^{>\mu})}\lrA{c(r,v)+c(v,r)}}{\OPT}&=\sum_{z\in[\size{\Q^{>\mu}_i}]}\sum_{j\in[\ell_i]}\frac{c(v^{z}_{0},v^{z}_{{\tilde{p}^{>\mu}_{i,j}}})+c(v^{z}_{{\tilde{p}^{>\mu}_{i,j}}},v^{z}_{0})}{\OPT}\\
&=\sum_{j\in[\ell_i]}\lrA{\chi(0, \tilde{p}^{>\mu}_{i,j})+\chi(\tilde{p}^{>\mu}_{i,j},0)},
\end{split}
\end{equation}
where the last equality follows from \eqref{funcons+}.

For any $i\in[N_{\mathrm{pat}}]$ with $\ell_i=2$, by \eqref{defDelta}, we obtain
\begin{equation}\label{meq3}
\begin{split}
\sum_{Q^{>\mu}\in \Q^{>\mu}_i}\frac{\Delta_{\max}(Q^{>\mu})}{\OPT}&=\sum_{z\in[\size{\Q^{>\mu}_i}]}\frac{c(v^{z}_{{\tilde{p}^{>\mu}_{i,1}}},v^{z}_{0})+c(v^{z}_{0},v^{z}_{{\tilde{p}^{>\mu}_{i,2}}})-c(v^{z}_{{\tilde{p}^{>\mu}_{i,1}}},v^{z}_{{\tilde{p}^{>\mu}_{i,2}}})}{\OPT}\\
&=\chi(\tilde{p}^{>\mu}_{i,1},0)+\chi(0, \tilde{p}^{>\mu}_{i,2})-\chi(\tilde{p}^{>\mu}_{i,1}, \tilde{p}^{>\mu}_{i,2}).
\end{split}
\end{equation}

Similarly, for any $i\in[N_{\mathrm{pat}}]$ with $\ell_i=3$, we obtain
\begin{equation}\label{meq4}
\begin{split}
\sum_{Q^{>\mu}\in \Q^{>\mu}_i}\frac{\Delta_{\mathrm{sum}}(Q^{>\mu})}{\OPT}&=2\chi(\tilde{p}^{>\mu}_{i,1}, 0)+\chi(0, \tilde{p}^{>\mu}_{i,2})+\chi(\tilde{p}^{>\mu}_{i,2},0)+2\chi(0, \tilde{p}^{>\mu}_{i,3})\\
&\quad\quad - \chi(\tilde{p}^{>\mu}_{i,1}, \tilde{p}^{>\mu}_{i,2})-\chi(\tilde{p}^{>\mu}_{i,1}, \tilde{p}^{>\mu}_{i,3})-\chi(\tilde{p}^{>\mu}_{i,2}, \tilde{p}^{>\mu}_{i,3}).
\end{split}
\end{equation}
Moreover, for any $i\in[N_{\mathrm{pat}}]$ with $\ell_i=3$,
\begin{equation}\label{meq5}
\sum_{Q^{>\mu}\in \Q^{>\mu}_i}\frac{\Delta_{\max}(Q^{>\mu})}{\OPT}\geq
\max
\left\{
\begin{array}{l}
\chi(\tilde{p}^{>\mu}_{i,1}, 0)+\chi(0, \tilde{p}^{>\mu}_{i,2}) - \chi(\tilde{p}^{>\mu}_{i,1}, \tilde{p}^{>\mu}_{i,2}),\\
\chi(\tilde{p}^{>\mu}_{i,1}, 0)+\chi(0, \tilde{p}^{>\mu}_{i,3}) - \chi(\tilde{p}^{>\mu}_{i,1}, \tilde{p}^{>\mu}_{i,3}),\\
\chi(\tilde{p}^{>\mu}_{i,2}, 0)+\chi(0, \tilde{p}^{>\mu}_{i,3}) - \chi(\tilde{p}^{>\mu}_{i,2}, \tilde{p}^{>\mu}_{i,3})
\end{array}
\right\}.
\end{equation}

Recall that $\Q^{>\mu}=\bigcup_{i\in[N_{\mathrm{pat}}]}\Q^{>\mu}_i$ is formed by shortcutting all customers in $V\setminus V_0$ with demands at most $\mu$ from $\T_{\mathrm{large}}$, whereas $\Q^{>\mu}_{\mathrm{left}}=\bigcup_{Q^{>\mu}\in\Q^{>\mu}}Q^{>\mu}_{\mathrm{left}}$ is constructed by modifying the tours in $\Q^{>\mu}$ based on $V^{>\mu}_{\mathrm{left}}$.

For convenience, let $\Q^{>\mu}_{i,\mathrm{left}}$ denote the set of tours in $\Q^{>\mu}_{\mathrm{left}}$ that are obtained by modifying the tours in $\Q^{>\mu}_i$ depending on $V^{>\mu}_{\mathrm{left}}$.
Consequently, the set $\lrc{\Q^{>\mu}_{1,\mathrm{left}},\ldots, \Q^{>\mu}_{N_{\mathrm{pat}},\mathrm{left}}}$ forms a valid partition of $\Q^{>\mu}_{\mathrm{left}}$. 

Using the definition of $\Q^{>\mu}_{i,\mathrm{left}}$ alongside \eqref{relc}, we have
\begin{equation}\label{meq6}
\phi_\gamma(I_\mu) \leq \frac{\EEE{c(\Q^{>\mu}_{\mathrm{left}})}}{\OPT}=\sum_{i\in[N_{\mathrm{pat}}]}\frac{\EEE{c(\Q^{>\mu}_{i,\mathrm{left}})}}{\OPT} =\sum_{j\in[3]}\sum_{i\in[N_{\mathrm{pat}}]:\ell_i = j}\frac{\EEE{c(\Q^{>\mu}_{i,\mathrm{left}})}}{\OPT}.
\end{equation}

Finally, the lemma is established by applying \eqref{meq2}--\eqref{meq6} in conjunction with Lemmas \ref{expt1}, \ref{expt2}, and \ref{expt3+}.
\end{proof}

\subsection{Proof of Lemma~\ref{bridge+}}\label{pbridge+}
\begin{proof}
Recall that $V_j$ denotes the set of all type-$j$ customers in $V\setminus V_0$ for each $j\in[N/2-N/\NM+1]$.

On one hand, by \eqref{tourcons+}, \eqref{vsetcons+}, and \eqref{funcons+}, together with the definition of $\mathbb{J}_j$, we have
\[
\frac{\sum_{v\in V_j}\lrA{c(r, v)+c(v,r)}}{\OPT} = \sum_{i\in\mathbb{J}_j}\lrA{\chi(0,i)+\chi(i,0)},\quad\forall\ j\in[N/2-N/\NM+1].
\]

On the other hand, by \eqref{int} and \eqref{defxi}, together with the definition of type-$j$ customers for each $j\in[N/2-N/\NM+1]$, we have
\begin{align*}
\xi_0(N) - \xi_0(N/2) & = \frac{\sum_{v\in V_{N/2-N/\NM+1}}\lrA{c(r, v)+c(v,r)}}{\RLB},\\
\xi_0(N/\NM+j) - \xi_0(N/\NM+j-1) & = \frac{\sum_{v\in V_{j}}\lrA{c(r, v)+c(v,r)}}{\RLB},\quad\forall\ j\in[N/2-N/\NM].
\end{align*}

Since $\RLB\leq\OPT$ by Lemma~\ref{lowerbound}, the bounds presented in the lemma follow immediately.
\end{proof}

\subsection{Proof of Lemma~\ref{dmitp}}\label{pdmitp}
\begin{proof}
For clarity, we first restate the randomized algorithm that yields the desired upper bound.

The randomized algorithm is identical to Algorithm~\ref{malg}, except that, instead of using using a minimum-cost 2-tour set $\T^{>\mu}_{\mathrm{left}}$ to serve all customers in $V^{>\mu}_{\mathrm{left}}$, it uses the 2-tour set $\Z^{>\mu}_{\mathrm{left}}$.
Then, by an argument analogous to that in the proof of Lemma~\ref{mitp}, we know that the resulting solution has expected cost at most the upper bound in \eqref{dup}. 
It therefore suffices to derandomize this algorithm while preserving the same upper bound.

Suppose $\T=\lrc{T_1,\ldots, T_m}$.
Let $y_i=1$ indicate that the tour $T_i$ is selected, and let $y_i=0$ indicate that the tour $T_i$ is rejected.
Our goal is to find a vector $\bm{y}^*\in \lrc{0,1}^{m}$ such that 
\[
\EEE{c(\T_{\mathrm{round}}) \mid \bm{y}=\bm{y}^*}+\EEE{c(\Z^{>\mu}_{\mathrm{left}}) \mid \bm{y}=\bm{y}^*}+\EEE{c(\T_{\mathrm{ITP}}) \mid \bm{y}=\bm{y}^*}\leq \SOL,
\]
where $\SOL$ denotes the upper bound in \eqref{dup}.

Note that $\EE{c(\T_{\mathrm{ITP}}) \mid \bm{y}=\bm{y}^*}$ may not be easy to compute, since the expected cost of the TSP tour $\A$ on $\lrc{r}\cup V^{\leq\mu}_{\mathrm{left}}$ is difficult to evaluate explicitly.
Fortunately, it is unnecessary to compute $\EE{c(\T_{\mathrm{ITP}})\mid \bm y=\bm y^*}$, since the bound $c(\A)\leq \alpha\cdot\TSP$ is already sufficient in the derivation of \eqref{dup}.

Define
\begin{equation}\label{newobj}
\begin{split}
c(V^{\leq\mu}_{\mathrm{left}})&\coloneq \frac{1}{1-\delta}\sum_{v\in \V_{0}^{\delta}\cap V^{\leq\mu}_{\mathrm{left}}} d_v\cdot \lrA{c(r,v)+c(v,r)}+ \frac{2}{1-\delta}\sum_{v\in \V_{\delta}^{1/2}\cap V^{\leq\mu}_{\mathrm{left}}}d_v\cdot \lrA{c(r,v)+c(v,r)}\\
&\quad\quad- \frac{\delta}{1-\delta}\sum_{v\in \V^{1/2}_\delta\cap V^{\leq\mu}_{\mathrm{left}}}{\lrA{c(r,v)+c(v,r)}}+\sum_{v\in \V^1_{1/2}\cap V^{\leq\mu}_{\mathrm{left}}}\lrA{c(r,v)+c(v,r)}.
\end{split}
\end{equation}
Then, it suffices to find a vector $\bm{y}^*\in \lrc{0,1}^{m}$ such that 
\begin{equation}\label{defobj}
\EEE{\OBJ\mid \bm{y}=\bm{y}^*}\coloneq\EEE{c(\T_{\mathrm{round}}) \mid \bm{y}=\bm{y}^*}+\EEE{c(\Z^{>\mu}_{\mathrm{left}}) \mid \bm{y}=\bm{y}^*}+\EEE{c(V^{\leq\mu}_{\mathrm{left}}) \mid \bm{y}=\bm{y}^*}\leq \SOL-\alpha\cdot\TSP.
\end{equation}

We determine each $y^*_i$ sequentially. 
Suppose that $y^*_1,\ldots, y^*_{s-1}$ have already been determined such that, during the LP-rounding procedure, we select each $T_i$ according to $y^*_i$ if $i\leq s-1$, and independently select each remaining $T_i$ with probability $\min\lrc{1,\gamma\cdot x^*_{T_i}}$ otherwise. Assume further that
\[
\EEE{\OBJ\mid y_1=y^*_1,\ldots, y_{s-1}=y^*_{s-1}}\leq \SOL-\alpha\cdot\TSP.
\]

We choose $y^*_s\in\{0,1\}$ such that
\begin{equation*}
y^*_s = \arg\min_{y^*_s\in\lrc{0,1}}\EEE{\OBJ\mid y_1=y^*_1,\ldots, y_{s}=y^*_{s}}.
\end{equation*}
By definition, 
\begin{equation*}
\EEE{\OBJ\mid y_1=y^*_1,\ldots, y_{s}=y^*_{s}}\leq \EEE{\OBJ\mid y_1=y^*_1,\ldots, y_{s-1}=y^*_{s-1}}\leq \SOL-\alpha\cdot\TSP.
\end{equation*}

Therefore, repeating this procedure determines the entire vector $\bm{y}^*$.

To compute $\EE{\OBJ\mid y_1=y^*_1,\ldots, y_{s}=y^*_{s}}$ in polynomial time, by \eqref{defobj}, it suffices to compute $\EE{c(\T_{\mathrm{round}}) \mid y_1=y^*_1,\ldots, y_{s}=y^*_{s}}$, $\EE{c(\Z^{>\mu}_{\mathrm{left}}) \mid y_1=y^*_1,\ldots, y_{s}=y^*_{s}}$, and $\EE{c(V^{\leq\mu}_{\mathrm{left}}) \mid y_1=y^*_1,\ldots, y_{s}=y^*_{s}}$.

\paragraph{\textbf{The Value $\EE{c(\T_{\mathrm{round}}) \mid y_1=y^*_1,\ldots, y_{s}=y^*_{s}}$.}} 
It can be computed by summing the expected contribution of each tour independently. 
Although duplicate visits are removed by shortcutting in the final solution, for the purpose of analysis we keep these visits without shortcutting, consistent with the probabilistic analysis in Lemma~\ref{prob}. Thus,
\begin{align*}
\EEE{c(\T_{\mathrm{round}}) \mid y_1=y^*_1,\ldots, y_{s}=y^*_{s}} &= \sum_{i=1}^{m}\PPP{T_i\ \text{chosen}\mid y_1=y^*_1,\ldots, y_{s}=y^*_{s}}\cdot c(T_i)\\
&=\sum_{i=1}^{s}y^*_i\cdot c(T_i) + \sum_{i=s+1}^{m}\min\lrc{1,\gamma\cdot x^*_{T_i}}\cdot c(T_i),
\end{align*}
where $\lrc{x^*_{T_i}}_{T_i\in\T}$ is an optimal fractional solution to the LP~\eqref{reflp}.

\paragraph{\textbf{The Value $\EE{c(\Z^{>\mu}_{\mathrm{left}}) \mid y_1=y^*_1,\ldots, y_{s}=y^*_{s}}$.}} 
It suffices to compute the contribution of each tour $T\in \Z^{>\mu}$ to $\EE{c(\Z^{>\mu}_{\mathrm{left}}) \mid y_1=y^*_1,\ldots, y_{s}=y^*_{s}}$. Since $\Z^{>\mu}$ is a 2-tour set, there are two possible cases.

\textbf{Case~1: $\size{V(T)}=1$.} Suppose $T=(r,v,r)$. 
Its contribution is
\begin{align*}
&\PPP{v\in V^{>\mu}_{\mathrm{left}}\mid y_1=y^*_1,\ldots, y_{s}=y^*_{s}}\cdot \lrA{c(r,v)+c(v,r)}\\
&=\prod_{T_i\in\T: v\in V(T_i)}\lrA{1-\PPP{T_i\ \text{chosen} \mid y_1=y^*_1,\ldots, y_{s}=y^*_{s}}}\cdot \lrA{c(r,v)+c(v,r)}.
\end{align*}
Since $\PP{T_i\ \text{chosen}\mid y_1=y^*_1,\ldots, y_{s}=y^*_{s}}$ equals $y^*_i$ if $i\leq s$ and $\min\lrc{1,\gamma\cdot x^*_{T_i}}$ otherwise, this contribution can be computed efficiently.

\textbf{Case~2: $\size{V(T)}=2$.} Suppose $T=(r,v_1,v_2,r)$. 
Its contribution is
{\small\begin{align*}
&\PPP{v_1,v_2\in V^{>\mu}_{\mathrm{left}}\mid y_1=y^*_1,\ldots, y_{s}=y^*_{s}}\cdot c(T)\\
&\quad+\lrA{\PPP{v_1\in V^{>\mu}_{\mathrm{left}}\mid y_1=y^*_1,\ldots, y_{s}=y^*_{s}}-\PPP{v_1,v_2\in V^{>\mu}_{\mathrm{left}}\mid y_1=y^*_1,\ldots, y_{s}=y^*_{s}}}\cdot \lrA{c(r,v_1)+c(v_1,r)}\\
&\quad+\lrA{\PPP{v_2\in V^{>\mu}_{\mathrm{left}}\mid y_1=y^*_1,\ldots, y_{s}=y^*_{s}}-\PPP{v_1,v_2\in V^{>\mu}_{\mathrm{left}}\mid y_1=y^*_1,\ldots, y_{s}=y^*_{s}}}\cdot \lrA{c(r,v_2)+c(v_2,r)}.
\end{align*}}
Note that
{\small\[
\PPP{v_1,v_2\in V^{>\mu}_{\mathrm{left}}\mid y_1=y^*_1,\ldots, y_{s}=y^*_{s}} = \prod_{T_i\in\T: \lrc{v_1,v_2}\cap V(T_i)\neq \emptyset}\lrA{1-\PPP{T_i\ \text{chosen} \mid y_1=y^*_1,\ldots, y_{s}=y^*_{s}}}.
\]}
By the analysis of Case~1, the probabilities $\PP{v_1\in V^{>\mu}_{\mathrm{left}}\mid y_1=y^*_1,\ldots, y_{s}=y^*_{s}}$, $\PP{v_2\in V^{>\mu}_{\mathrm{left}}\mid y_1=y^*_1,\ldots, y_{s}=y^*_{s}}$, and $\PP{T_i\ \text{chosen} \mid y_1=y^*_1,\ldots, y_{s}=y^*_{s}}$ can all be computed efficiently. 
Hence, this contribution can also be computed efficiently.

\paragraph{\textbf{The Value $\EE{c(V^{\leq\mu}_{\mathrm{left}}) \mid y_1=y^*_1,\ldots, y_{s}=y^*_{s}}$.}} 
By \eqref{newobj}, it suffices to compute $\PP{v\in V^{\leq\mu}_{\mathrm{left}}\mid y_1=y^*_1,\ldots, y_{s}=y^*_{s}}$ for each $v\in V$, which can be computed efficiently by the analysis of Case~1.

It remains to analyze the running time of the deterministic algorithm.

Determining each $y^*_s$ requires evaluating the above conditional expectations, each of which can be computed in polynomial time in $m=\size{\T}$. Hence, the entire derandomization procedure runs in polynomial time in $m$. 
Recall that the number of possible tours $m$ is bounded by $n^{\O(1/\sigma)}$ when $\sigma>0$, and by $n^{\O(k)}$ when $\sigma=0$. 
Moreover, the MW$2$SCP used to compute $\Z^{>\mu}$ can also be solved in polynomial time. 
Therefore, the overall running time remains $n^{\O(1/\sigma)}$ when $\sigma>0$ and $n^{\O(k)}$ when $\sigma=0$.
\end{proof}

\subsection{Proof of Lemma~\ref{tmu+}}\label{ptmu+}
\begin{proof}
Recall that, by \eqref{newVi} and \eqref{newQi}, we have
\begin{equation*}
\mathbb{T}^{>\mu}_{i}=\lrA{0,\tilde{p}^{>\mu}_{i,1},\ldots,\tilde{p}^{>\mu}_{i,\ell_i},0}\quad\text{and}\quad \Q^{>\mu}_i = \lrC{\lrA{v^{z}_{0},v^{z}_{{\tilde{p}^{>\mu}_{i,1}}},\ldots,v^{z}_{{\tilde{p}^{>\mu}_{i,\ell_i}}},v^{z}_{0}}}_{z\in[\size{\Q^{>\mu}_i}]}.
\end{equation*}

We first construct a 2-tour set $\E^{>\mu}\coloneq \bigcup_{i\in[N_{\mathrm{pat}}]}\E^{>\mu}_i$ serving all customers in $V^{>\mu}$, where $\E^{>\mu}_i$ is a 2-tour set serving all customers in $V(\Q^{>\mu}_i)$ for each $i\in[N_{\mathrm{pat}}]$. Then, we analyze the expected cost of $\E^{>\mu}_{\mathrm{left}}$, obtained via shortcutting all customers not in $V^{>\mu}_{\mathrm{left}}$ from $\E^{>\mu}$.
Moreover, for each $i\in[N_{\mathrm{pat}}]$, let $\E^{>\mu}_{i,\mathrm{left}}$ denote the set of tours obtained via shortcutting all customers not in $V^{>\mu}_{\mathrm{left}}$ from $\E^{>\mu}_i$.

Next, we show the construction of $\E^{>\mu}_i$ and analyze $\EE{c(\E^{>\mu}_{i,\mathrm{left}})}$.

\textbf{Case~1: $\ell_i=1$.} We directly define $\E^{>\mu}_i\coloneq \Q^{>\mu}_i$. By Lemma~\ref{expt1}, we have 
\begin{equation*}
\EEE{c(\E^{>\mu}_{i,\mathrm{left}})}\leq \sum_{Q^{>\mu}\in \Q^{>\mu}_i}\lrA{p_\gamma\sum_{v\in V(Q^{>\mu})}\lrA{c(r,v)+c(v,r)}}.
\end{equation*}

\textbf{Case~2: $\ell_i=2$.} We still define $\E^{>\mu}_i\coloneq \Q^{>\mu}_i$. By Lemma~\ref{expt2}, we have 
\begin{equation*}
\EEE{c(\E^{>\mu}_{i,\mathrm{left}})}\leq \sum_{Q^{>\mu}\in \Q^{>\mu}_i}\lrA{p_\gamma\sum_{v\in V(Q^{>\mu})}\lrA{c(r,v)+c(v,r)}-\max\lrC{0,2p_\gamma-1}\cdot \Delta_{\max}(Q^{>\mu})}.
\end{equation*}

\textbf{Case~3: $\ell_i=3$.}
Initialize $\E^{>\mu}_i=\emptyset$.
For each tour $Q^{>\mu}\in \Q^{>\mu}_i$, assume that $Q^{>\mu}=\lra{r, v_1,v_2,v_3, r}$ and $\Delta_{\max}(Q^{>\mu}) = c(v_i,r) + c(r, v_j) - c(v_i,v_j)$ for distinct $i,j\in[3]$ with $i<j$, and let $z$ be the remaining index in $[3]\setminus\lrc{i,j}$, and then we construct one 2-tour serving $v_i$ and $v_j$ and one 1-tour serving $v_z$ and select them into $\E^{>\mu}_i$, i.e., $\E^{>\mu}_i\gets \E^{>\mu}_i\cup\lrc{(r, v_z, r), (r, v_i,v_j,r)}$.
Note that $\Delta_{\max}((r, v_i,v_j,r))=\Delta_{\max}(Q^{>\mu})$.
Hence, by Lemmas~\ref{expt1} and \ref{expt2}, we still have 
\begin{equation*}
\EEE{c(\E^{>\mu}_{i,\mathrm{left}})}\leq \sum_{Q^{>\mu}\in \Q^{>\mu}_i}\lrA{p_\gamma\sum_{v\in V(Q^{>\mu})}\lrA{c(r,v)+c(v,r)}-\max\lrC{0,2p_\gamma-1}\cdot \Delta_{\max}(Q^{>\mu})}.
\end{equation*}

By the proof of Lemma~\ref{tmu}, we obtain 
\begin{equation*}
\frac{\EEE{c(\E^{>\mu}_{\mathrm{left}})}}{\OPT}\leq \tilde{L}(p_\gamma,\psi_1(I_\mu),\psi_2(I_\mu),\psi_3(I_\mu)).
\end{equation*}

Finally, by the definition of $\Z^{>\mu}$, we have $\tilde{\phi}_\gamma(I_\mu) = \frac{\EE{c(\Z^{>\mu}_{\mathrm{left}})}}{\OPT}\leq \frac{\EE{c(\E^{>\mu}_{\mathrm{left}})}}{\OPT}$, which completes the proof.
\end{proof}

\section{Additional Numerical Results}\label{additional}
As mentioned earlier, two additional properties can be exploited in the metric case.
First, the number of distinct patterns can be reduced.
Second, the bounds in part (b) of Lemma~\ref{zwp+} can be strengthened to
$\chi(i,j)=\chi(j,i)\leq \chi(i,\ell)+\chi(\ell,j)$.
Therefore, the factor-revealing LPs for the metric case can be obtained by slightly modifying the factor-revealing LPs for the asymmetric case presented in Theorems~\ref{thm4} and \ref{thm5}.

Tables~\ref{res3}--\ref{res6} summarize the optimal values and running times of the factor-revealing LPs under both the asymmetric and metric cases for $M=1000$ and $M=1500$.
Tables~\ref{res3} and \ref{res5} report the results for \textsc{R-General-Alg} and \textsc{R-Fixed-Alg},
while Tables~\ref{res4} and \ref{res6} report the results for \textsc{D-General-Alg} and \textsc{D-Fixed-Alg}.

\paragraph{\textbf{Comparison of the Asymmetric and Metric Cases.}}

From Tables~\ref{res3}--\ref{res6}, we observe no significant differences in the optimal values of the factor-revealing LPs between the asymmetric and metric cases for both the randomized and deterministic algorithms.
The minor discrepancies are likely due to numerical errors arising during computation. 
In fact, we conjecture that the exact theoretical values are identical in both cases.

\paragraph{\textbf{The Effect of the Discretization Parameter $M$.}}
By comparing the results for $M=1000$ (Tables~\ref{res3} and \ref{res4}) with those for $M=1500$ (Tables~\ref{res5} and \ref{res6}), we find that the optimal values of the factor-revealing LPs remain virtually unchanged for both the randomized and deterministic algorithms.
This suggests that increasing $M$ beyond $1000$ yields little further improvement, confirming that $M=1000$ is sufficient for our main analysis.

\begin{table}[t]
\centering
\small
\renewcommand{\arraystretch}{0.89}
\caption{
Values and running times of the factor-revealing LPs for \textsc{R-General-Alg} and \textsc{R-Fixed-Alg} with $M=1000$ and varying $N$ under the asymmetric and metric cases.
}
\setlength{\tabcolsep}{5pt}
\begin{tabular}{ccccccccc}
\toprule
\multirow{3}{*}{$N$} 
&\multicolumn{4}{c}{Asymmetric Case} & \multicolumn{4}{c}{Metric Case}\\
\cmidrule{2-9}
& \multicolumn{2}{c}{\textsc{R-General-Alg}} & \multicolumn{2}{c}{\textsc{R-Fixed-Alg}} & \multicolumn{2}{c}{\textsc{R-General-Alg}} & \multicolumn{2}{c}{\textsc{R-Fixed-Alg}}\\
\cmidrule{2-9}
& Value & Runtime (s) & Value & Runtime (s) & Value & Runtime (s) & Value & Runtime (s)\\
\midrule
\midrule
$60$ & $1.650110$ & $46$ & $1.552077$ & $13$ & $1.650110$ & $26$ & $1.552077$ & $11$ \\
$120$ & $1.647908$ & $642$ & $1.549126$ & $286$ & $1.647908$ & $83$ & $1.549126$ & $69$ \\
$180$ & $1.647120$ & $4664$ & $1.548119$ & $1149$ & $1.647169$ & $2637$ & $1.548119$ & $1008$ \\
$240$ & $1.646793$ & $10111$ & $1.547613$ & $3993$ & $1.646786$ & $9280$ & $1.547612$ & $2770$ \\
$300$ & $1.646325$ & $32053$ & $1.547308$ & $7231$ & $1.646370$ & $35545$ & $1.547308$ & $6952$ \\
$360$ & $1.646418$ & $54449$ & $1.547104$ & $23290$ & $1.646312$ & $44988$ & $1.547104$ & $15368$ \\
\bottomrule
\end{tabular}
\label{res3}
\end{table}

\begin{table}[t]
\centering
\small
\renewcommand{\arraystretch}{0.89}
\caption{
Values and running times of the factor-revealing LPs for \textsc{D-General-Alg} and \textsc{D-Fixed-Alg} with $M=1000$ and varying $N$ under the asymmetric and metric cases.
}
\setlength{\tabcolsep}{5pt}
\begin{tabular}{ccccccccc}
\toprule
\multirow{3}{*}{$N$} 
&\multicolumn{4}{c}{Asymmetric Case} & \multicolumn{4}{c}{Metric Case}\\
\cmidrule{2-9}
& \multicolumn{2}{c}{\textsc{D-General-Alg}} & \multicolumn{2}{c}{\textsc{D-Fixed-Alg}} & \multicolumn{2}{c}{\textsc{D-General-Alg}} & \multicolumn{2}{c}{\textsc{D-Fixed-Alg}}\\
\cmidrule{2-9}
& Value & Runtime (s) & Value & Runtime (s) & Value & Runtime (s) & Value & Runtime (s)\\
\midrule
\midrule
$60$ & $1.651271$ & $12$ & $1.552437$ & $9$ & $1.651271$ & $11$ & $1.552437$ & $10$ \\
$120$ & $1.648878$ & $302$ & $1.549425$ & $246$ & $1.648878$ & $79$ & $1.549425$ & $84$ \\
$180$ & $1.648018$ & $1951$ & $1.548403$ & $1216$ & $1.648037$ & $1524$ & $1.548403$ & $1039$ \\
$240$ & $1.647607$ & $7701$ & $1.547889$ & $3368$ & $1.647599$ & $8004$ & $1.547889$ & $3977$ \\
$300$ & $1.647136$ & $30106$ & $1.547581$ & $9636$ & $-1$ & $10975$ & $1.547581$ & $7605$ \\
$360$ & $-1$ & $40586$ & $1.547375$ & $16651$ & $-1$ & $20443$ & $1.547151$ & $13560$ \\
\bottomrule
\end{tabular}
\label{res4}
\end{table}

\begin{table}[t]
\centering
\small
\renewcommand{\arraystretch}{0.89}
\caption{
Values and running times of the factor-revealing LPs for \textsc{R-General-Alg} and \textsc{R-Fixed-Alg} with $M=1500$ and varying $N$ under the asymmetric and metric cases.
}
\setlength{\tabcolsep}{5pt}
\begin{tabular}{ccccccccc}
\toprule
\multirow{3}{*}{$N$} 
&\multicolumn{4}{c}{Asymmetric Case} & \multicolumn{4}{c}{Metric Case}\\
\cmidrule{2-9}
& \multicolumn{2}{c}{\textsc{R-General-Alg}} & \multicolumn{2}{c}{\textsc{R-Fixed-Alg}} & \multicolumn{2}{c}{\textsc{R-General-Alg}} & \multicolumn{2}{c}{\textsc{R-Fixed-Alg}}\\
\cmidrule{2-9}
& Value & Runtime (s) & Value & Runtime (s) & Value & Runtime (s) & Value & Runtime (s)\\
\midrule
\midrule
$60$ & $1.650110$ & $29$ & $1.552078$ & $28$ & $1.650110$ & $31$ & $1.552077$ & $19$ \\
$120$ & $1.647908$ & $740$ & $1.549126$ & $186$ & $1.647908$ & $156$ & $1.549126$ & $190$ \\
$180$ & $1.646905$ & $4548$ & $1.548119$ & $2257$ & $-1$ & $5599$ & $1.548120$ & $1126$ \\
$240$ & $1.646798$ & $15237$ & $1.547613$ & $5929$ & $-1$ & $7044$ & $1.547612$ & $6757$ \\
$300$ & $1.646521$ & $30628$ & $1.547297$ & $14296$ & $1.646135$ & $25323$ & $1.547308$ & $17058$ \\
$360$ & $-1$ & $20045$ & $1.547104$ & $39867$ & $1.646047$ & $56340$ & $1.547100$ & $29929$ \\
\bottomrule
\end{tabular}
\label{res5}
\end{table}

\begin{table}[t]
\centering
\small
\renewcommand{\arraystretch}{0.89}
\caption{
Values and running times of the factor-revealing LPs for \textsc{D-General-Alg} and \textsc{D-Fixed-Alg} with $M=1500$ and varying $N$ under the asymmetric and metric cases.
}
\setlength{\tabcolsep}{5pt}
\begin{tabular}{ccccccccc}
\toprule
\multirow{3}{*}{$N$} 
&\multicolumn{4}{c}{Asymmetric Case} & \multicolumn{4}{c}{Metric Case}\\
\cmidrule{2-9}
& \multicolumn{2}{c}{\textsc{D-General-Alg}} & \multicolumn{2}{c}{\textsc{D-Fixed-Alg}} & \multicolumn{2}{c}{\textsc{D-General-Alg}} & \multicolumn{2}{c}{\textsc{D-Fixed-Alg}}\\
\cmidrule{2-9}
& Value & Runtime (s) & Value & Runtime (s) & Value & Runtime (s) & Value & Runtime (s)\\
\midrule
\midrule
$60$ & $1.651271$ & $21$ & $1.552438$ & $15$ & $1.651271$ & $16$ & $1.552437$ & $12$ \\
$120$ & $1.648878$ & $132$ & $1.549425$ & $482$ & $1.648878$ & $137$ & $1.549425$ & $127$ \\
$180$ & $1.648026$ & $3546$ & $1.548403$ & $2597$ & $1.648039$ & $1846$ & $1.548403$ & $2138$ \\
$240$ & $1.647603$ & $14792$ & $1.547889$ & $5609$ & $1.647606$ & $11023$ & $1.547889$ & $5279$ \\
$300$ & $-1$ & $30308$ & $1.547205$ & $12762$ & $-1$ & $33932$ & $1.547581$ & $10485$ \\
$360$ & $-1$ & $29158$ & $1.547322$ & $26130$ & $-1$ & $22871$ & $1.547376$ & $22929$ \\
\bottomrule
\end{tabular}
\label{res6}
\end{table}

\end{document}